\documentclass[11pt]{article}
\usepackage{amsmath}
\usepackage{fullpage}
\usepackage{graphicx}
\usepackage{enumerate}
\usepackage{natbib}
\usepackage{url} 
\usepackage{amssymb, amsthm, bbm}
\usepackage{caption,subcaption,placeins,color,dsfont,xcolor}
\usepackage{fancyhdr}
\usepackage{hyperref}
\usepackage{accents}
\usepackage{mathtools}
\usepackage{algorithm}
\usepackage{algcompatible}
\usepackage{multirow}
\usepackage{booktabs}
\newcommand{\blind}{1}

\newtheorem{lemma}{Lemma}
\newtheorem{theorem}{Theorem}
\newtheorem{proposition}{Proposition}
\newtheorem{assumption}{Assumption}
\newtheorem{definition}{Definition}
\newtheorem{example}{Example}
\newtheorem{remark}{Remark}

\def\arrvline{\hfil\kern\arraycolsep\vline\kern-\arraycolsep\hfilneg}

\DeclareMathOperator*{\argmin}{arg\,min}

\DeclarePairedDelimiter\floor{\lfloor}{\rfloor}

\begin{document}

\def\spacingset#1{\renewcommand{\baselinestretch}%
{#1}\small\normalsize} \spacingset{0}


\if1\blind
{
  \title{\bf Population Interference In Panel Experiments}
  \author{Kevin Wu Han \vspace{.2cm} \hspace{.2cm}\\
    Department of Statistics, Stanford University \vspace{.2cm}\\
    and \vspace{.2cm} \\
    Iavor Bojinov \vspace{.2cm}\\ 
    Harvard Business School \vspace{.2cm}\\
    and \vspace{.2cm}\\
    Guillaume Basse \vspace{.2cm}\\
    Department of MS\&E and Department of Statistics, Stanford University}
  \maketitle
} \fi

\if0\blind
{
  \bigskip
  \bigskip
  \bigskip
  \begin{center}
    {\LARGE\bf Population Interference In Panel Experiments}
\end{center}
  \medskip
} \fi

\bigskip
\begin{abstract}
The phenomenon of population interference, where a treatment assigned to one experimental unit affects another experimental unit's outcome, has received considerable attention in standard randomized experiments. The complications produced by population interference in this setting are now readily recognized, and partial remedies are well known. Less understood is the impact of population interference in panel experiments where treatment is sequentially randomized in the population, and the outcomes are observed at each time step. This paper proposes a general framework for studying population interference in panel experiments and presents new finite population estimation and inference results. Our findings suggest that, under mild assumptions, the addition of a temporal dimension to an experiment alleviates some of the challenges of population interference for certain estimands. In contrast, we show that the presence of carryover effects --- that is, when past treatments may affect future outcomes --- exacerbates the problem. Our results are illustrated through both an empirical analysis and an extensive simulation study.  
\end{abstract}

\noindent%
\vfill


\newpage
\spacingset{1.25} 

\section{Introduction}

When researchers estimate causal effects from randomized experiments, they almost always make assumptions that restrict the number of counterfactual outcomes to simplify the subsequent inference. In standard experiments, where units are randomly assigned to either a treatment or control, researchers commonly assume that one unit’s assignment does not affect another unit’s response; this is usually referred to as no interference \citep[Chapter~2]{cox58}. In panel experiments, where units are exposed to different interventions over time \citep{bojinov2020panel}, in addition to no interference, researchers regularly assume that the observed outcomes were not impacted by past assignments; this is often called the no carryover assumption \citep[Chapter~13]{cox58}. Although these two assumptions are useful, there are numerous empirical examples where they are violated. This mismatch between practical applications and theoretical assumptions has catalyzed a growing amount of literature dedicated to studying relaxations of these stringent conditions for either standard or panel experiments, but not both.

In standard experiments without evoking the no interference assumption, each unit’s outcome depends on the assignments received by all other experimental units. Allowing for such arbitrary population interference\footnote{We use the term \emph{population interference} to emphasize that the interference occurred across units.} makes causal inference challenging \citep{basseairoldi2018}. In practice, researchers look for an underlying structure that limits the scope of interference. For example, when studying electoral participation during a special election in 2009 in Chicago, \cite{sinclair2012} assumed that interference occurred within-household but not across; more broadly, this type of interference has been found in many other applications, including education (\cite{hong2006}; \cite{rosenbaum2007}), economics (\cite{sobel2006}; \cite{manski2013}) and public health (\cite{halloran1995}). Inference in this setting is challenging because interference increases the number of potential outcomes and makes observations dependent. \cite{aronow2017} introduce a general framework for studying causal inference with interference: they introduce the concept of exposure mapping, define useful estimands, and construct asymptotically valid confidence intervals based on the Horvitz-Thompson estimator. 

The literature on panel experiments has similarly shifted towards relaxing the no carryover effects assumption that precludes outcomes from being impacted by past assignments. For example, in the most extreme case, \cite{bojinov19} allows for arbitrary carryover effects when studying whether algorithms or humans are better at executing large financial orders; more generally, these types of relaxations have also been studied in economics \citep{angrist2011causal, rambachan2019econometric}, epidemiology \citep{robins1986new,RobinsGreenlandHu99}, public health \citep{boruvka2018assessing}, and political science \citep{blackwell2018make}. Similarly to relaxing the no interference assumption, removing the no carryover assumption enables researchers to develop and explore a richer class of causal estimands that capture both the contemporaneous and delayed causal effects \citep{bojinov19}. The latter is particularly important for technology companies seeking to understand the long-term impact of their interventions \citep{basse2019minimax,hohnhold2015focusing}. Similarly to the population interference setting, researchers use the analogous Horvitz-Thompson type estimator estimators to analyze experiments with carryover effects.

An apparent gap in the literature is an understanding of whether the possibility of running a panel experiment alleviates the challenges posed by population interference or makes them worse. This is particularly important for researchers wishing to run field experiments for two reasons. First, it is often the case the researchers are constrained on the maximal number of experimental units they can recruit, for instance, because of costs or limits in the total population. Second, population interference often leads to increased uncertainty that reduces only by increasing the sample size. Panel experiments can alleviate these as it is often cheaper to change an experimental unit's treatment than to recruit more subjects, and uncertainty tends to decrease as the sample size and the number of time period increase. However, what happens when there is population interference and carryover effects is unclear. 

To address this gap, we introduce a unifying framework for studying panel experiments with population interference. We begin by focusing on panel experiments with population interference but no carryover effects (Section~\ref{sec: panel}). Here, we provide asymptotically valid confidence intervals for estimands defined at specific time periods and estimands that average contrasts over multiple time periods. We also introduce a novel class of assumptions enabling us to leverage past data to improve inference at a given time. Together, our results show that using panel experiments when there is population interference allows us to achieve valid inference under much weaker conditions on the population interference and even drop all restrictions for large time horizons. These results should be particularly encouraging for researchers wishing to run field experiments when the number of experimental subjects is constrained, as is often the case in Economics (for example, \cite{RePEc:eee:jetheo:v:127:y:2006:i:1:p:117-154}) and Management Science (for example, \cite{choudhury2021virtual}). 

We then tackle the most general setting featuring both population and temporal interference (Section~\ref{subsecsec: mixed_int}). Under additional assumptions, we derive a general central limit theorem, which fails to provide the same clear benefit because of the data complexity caused by carryover effects. We also state asymptotic results for a restricted type of mixed interference that generalizes the usual stratified interference to panel experiments and provides a blueprint for deriving additional results in specific contexts. Here we show a clear benefit that incorporating a temporal dimension allows us to relax the main restriction on the maximal cluster size to obtain valid inference. For researchers, these results are slightly less encouraging but, nevertheless, provide an essential next step in understanding how to leverage panel experiments in real-world settings. 

Finally, Section~\ref{sec: setup} details our setup by introducing the potential outcomes framework, our causal estimands and corresponding estimators, and the randomization-based framework that we leverage for all our results. We conclude with simulations (Section \ref{sec: simulations}), empirical applications (Section~\ref{sec: real-data-ex}), and a discussion (Section~\ref{sec: conclusion}). The Appendix contains a detailed discussion of inference under population interference for standard experiments, all proofs, and additional simulations.  
\section{Setup}
\label{sec: setup}
\subsection{Assignments}
\label{subsec: assignments}

Consider a randomized experiment occurring over $T$ periods on a finite population 
of $n$ experimental units. At each time step $t \in \{1, \cdots, T\}$, unit 
$i \in \{1, \cdots, n\}$ can be assigned to treatment ($W_{i, t} = 1$) or 
control ($W_{i, t} = 0$); extensions to non-binary treatments are straightforward. 
We denote by $W_{i, 1:t} = (W_{i, 1}, W_{i, 2}, \cdots, W_{i, t})$ the assignment 
path up to time $t$ for unit $i$, $W_{1:n, t}$ the assignment vector for all $n$ 
units at time step $t$ and $W_{1:n, 1:t} \in \{0, 1\}^{n \times t}$ the assignment 
matrix. Hence, for each $i$ and $t$, $W_{i, 1:t}$ is a vector of length $t$, 
$W_{1:n, t}$ is a vector of length $n$ and $W_{1:n, 1:t}$ is a matrix of 
dimension $n \times t$.

We define an assignment mechanism (or design) to be the probability distribution of the assignment matrix $\mathbb{P}(W_{1:n, 1:T})$.

For example, in a Bernoulli design, the assignment mechanism is independent across time and units such that $\mathbb{P}(W_{i,t} = 1) = p$ for all $i,t$. More generally, we will often work with assignments that are temporally independent.
\begin{definition}[Temporally independent assignment mechanism]\label{def temporal indep assignment}
    We say an assignment mechanism is temporally independent if $W_{1:n, t}$ and $W_{1:n, t^{'}}$ are independent for any $t$ and $t^{'}$
\end{definition}

 Following much of the literature on analyzing complex experiments, we adopt the randomization-based approach to inference, in which the assignment mechanism is the only source of randomness \citep{kempthorne1955}; see \cite{abadie_2020sampling} for a recent review. Throughout, we use lower cases $w$ with the appropriate subscript for realizations of the assignment matrix $W$.  
\subsection{Potential outcomes and exposure mappings}
\label{subsec: po, est}
The goal of causal inference is to study how an intervention impacts an outcome of interest. Following the potential outcomes formulation, for panel experiments without any assumptions, each unit $i$ at time $t$ has $2^{nT}$ potential outcomes corresponding to the total number of distinct realizations of the assignment matrix, denoted by $Y_{i,t} (w_{1:n, 1:T})$. For simplicity, we assume that the potential outcomes are one-dimensional, although it is straightforward to relax this assumption. 

In randomized experiments, where we control the assignment mechanism, the outcomes at time $t$ are not impacted by future assignments that have yet to be revealed to the units \citep{bojinov19}. This assumption drastically reduces the total number of potential outcomes\footnote{The assumption, known as non-anticipating potential outcomes \citep{bojinov19}, can be violated if experimental units are told what their future assignments will be and modify their present behavior as a result. For instance, this could occur for shoppers who expect to receive a considerable discount on a subsequent day and may curtail their spending until they receive the discount.} and will be implicitly made throughout this paper. We now formally define the no carryover effects and no population interference assumption. 
\begin{definition} [No carryover effect and population interference]
We say that there is no carryover effect if and only if
\begin{equation*}
    Y_{i, t}(w_{1:n, 1:t}) = Y_{i, t}(w_{1:n, 1:t}^{'}) \text{ whenever } w_{1:n, t} = w_{1:n, t}^{'}.
\end{equation*}
And we say that there is no population interference if and only if
\begin{equation*}
    Y_{i, t}(w_{1:n, 1:t}) = Y_{i, t}(w_{1:n, 1:t}^{'}) \text{ whenever } w_{i, 1:t} = w_{i, 1:t}^{'}.
\end{equation*}
\end{definition}

If we make both assumptions, inference is relatively straightforward. However, if we only assume that there is population interference and no carryover effects, inference is still impossible without any assumptions on the population interference structure  \citep{basseairoldi2018}. One way forward is to assume that the outcomes of unit $i$ depend only on the treatments assigned to a subset of the population. This intuition extends more generally to the assertion that the outcome of unit $i$ at time $t$ depends on a low-dimensional representation of $w_{1:n, 1:t}$. Formally, for each unique $i,t$ pair we define the exposure mapping $f_{i,t}: \{0,1\}^{n\times t} \to \Delta$, where $\Delta$ is the set of possible exposures\footnote{To make exposure mappings useful, we assume the cardinality of $\Delta$ is (substantially) smaller than $n\times t$.} \citep{aronow2017}. On the other hand, if we assume that there is no population interference, again, we would need additional assumptions to obtain valid randomization-based inference \citep{bojinov2020panel}.

Defining exposure mappings in this flexible manner allows us to unify and transparently consider restrictions on the population interference and the duration of the carryover effect. Throughout this paper, we restrict our focus to properly specified time-invariant exposure mappings, which are formally defined below.

\begin{assumption}[Properly specified time-invariant exposure mapping] \label{assump:expmap}
    The exposure mappings are properly specified if, for all pairs $i \in \{1, \cdots, n\}$ and $t\in\{1, \cdots, T\}$, and any two assignment matrices $w_{1:n, 1:t}$ and $w_{1:n, 1:t}^{'}$,
\begin{align*}
    Y_{i, t}(w_{1:n, 1:t}) &= Y_{i, t}(w_{1:n, 1:t}^{'}) \text{ whenever } \\
    f_{i, t}(w_{1:n, 1:t}) &= f_{i, t}(w_{1:n, 1:t}^{'}).
\end{align*}
For $p \in \{1, \cdots, T\}$, we say the exposure mappings are $p$-time-invariant if for any $t, t^{'} \in \{p, \cdots, T\}$ and any unit $i$,
\begin{align*}
    &f_{i, t}(w_{1:n, 1:t}) = f_{i, t^{'}}(w_{1:n, 1:t^{'}}) \text{ whenever }\\
    &w_{1:n, t - p + 1:t} = w_{1:n, t^{'} - p + 1:t^{'}}.
\end{align*}
The exposure mappings are time-invariant if the exposure mappings are $p$-time-invariant for some $p \in \{1, \cdots, T\}$.
We say the exposure mappings are properly specified time-invariant exposure mappings if they are both properly specified and time-invariant.  
\end{assumption}
Properly specified exposure mappings can be thought of as defining ``effective treatments,'' allowing us to write:
\begin{equation*}
        Y_{i,t}(w_{1:n,1:t}) = Y_{i,t}(f_{i,t}(w_{1:n,1:t})) = Y_{i,t}(h_{i,t}),
\end{equation*}
where $h_{i,t} = f_{i,t}(w_{1:n,1:t}) \in \Delta$. Time-invariant exposure mappings constrain the relationship between experimental units to be invariant over time. Specifically, it does not allow the exposure mappings to change across time. For example, if at time $t = 1$, the outcomes depend on the fraction of treated neighbors in the graph, then it cannot be the case that at time $t = 2$ the outcomes now depend on the number of treated neighbors in the graph. We will see why such an invariance assumption is necessary for the next section when we define causal effects. Of course, the validity of Assumption~\ref{assump:expmap} depends on the exact definition of the exposure mapping and should be informed by the empirical context. 

Throughout this paper, we consider a special class of exposure mappings that restrict the outcomes of unit $i$ to depend only on the assignments of a predefined subset of units that we refer to as $i$'s neighborhood and index by $\mathcal{N}_i$; note that the index set is not dynamically changing over time. For example, for units connected through a social network, $\mathcal{N}_i$ indexes the set of nodes connected to $i$ by an edge; for units organized households,  $\mathcal{N}_i$ indexes the set of units that live in the same household as $i$; and for units located in space, $\mathcal{N}_i$ indexes the set of units who are at most a certain distance away from unit $i$. 
\begin{definition} [Locally Effective Assignments (LEA)]
\label{def: leot}
We say the assignments and exposure mappings are locally effective if the exposure mappings are $p$-time-invariant for some $p \in \{1, \cdots, T\}$ and
\begin{equation*}
    f_{i, t}(w_{1:n, 1:t}) = f_{i, t}(w_{\mathcal{N}_i, t - p + 1:t}),
\end{equation*}
with the convention that $w_{\mathcal{N}_i, t - p + 1:t} = w_{\mathcal{N}_i, 1:t}$ for $t - p + 1 \leq 0$.
\end{definition}

Although LEA imposes further structure, it still provides a great deal of flexibility as it incorporates all notions of traditional population interference and temporal carryover effects as special cases. For example, fixing $p = 1$ makes the exposure values depend only on current assignments, which is equivalent to usual population interference. On the other hand, fixing $\mathcal{N}_i = \{i\}$ is equivalent to the no interference assumption imposed on panel experiments in \cite{bojinov2020panel}. Of course, these special cases are interesting and extensively studied, but our general formulation's real benefit is to consider scenarios where there is both population interference and carryover effects.

\begin{example} [Example of Locally Effective Assignments]
We consider an example where the exposure values depend on past assignments. In particular, let 
\begin{equation*}
    f_{i, t}(w_{1:n, 1:t}) = (w_{i,t - 1}, w_{i,t}, u_{i, t-1}, u_{i, t})
\end{equation*} 
where $u_{i, t - 1} = \frac{1}{|\mathcal{N}_i|}\sum_{j \in \mathcal{N}_i}w_{j, t-1}$ 
and $u_{i, t} = \frac{1}{|\mathcal{N}_i|}\sum_{j \in \mathcal{N}_i}w_{j, t}$; we use $|\mathcal{A}|$ to denote the cardinality of the set $\mathcal{A}$. Hence, one unit's assignment and the fraction of treated neighbors at the previous time step matter as well. This is a special case of LEA with $p = 2$. In this example, the exposure mappings are 2-time-invariant: for $t, t^{'} \geq 2$, if $w_{1:n, (t-1):t} = w_{1:n, (t^{'}-1):t^{'}}$ then $f_{i, t}(w_{1:n, 1:t}) = f_{i, t^{'}}(w_{1:n, 1:t^{'}})$. 
\end{example}
One limitation of LEA($p$) assumption is that it can not directly capture long-range decaying dependency on past assignments and population interference beyond a unit's neighborhood. Such long-range decaying dependency on time is uncommon in econometrics literature \citep{JUDSON19999, wooldridge2010}. For example, if we consider the following parametric model \citep{wooldridge2010}:
\begin{equation*}
    Y_{i, t} = \rho Y_{i, t-1} + \beta W_{i, t} + \epsilon_{i, t}, \text{ where } \rho \in (-1, 1) \text{ and } \mathbb{E}[\epsilon_{i, t} | Y_{i, 1}, \cdots, Y_{i, t}, W_{i, 1}, \cdots, W_{i, t}] = 0,
\end{equation*}
we have an infinite-long decaying dependency on past assignments on the current outcomes. To capture this, we would need to set $p=T$ in the LEA($p$) causal effect, which, although possible, is unlikely to be practically useful as we would not be able to estimate this estimand with any reasonable precision. Despite this limitation, the LEA$(p)$ assumption allows for a great deal of flexibility as it does not require imposing modeling assumptions on the outcomes and is still useful in many real-world situations. 

Finally, population interference beyond local interference has also been studied in the econometrics literature \citep{manski1993, BRAMOULLE200941, leung2022}. We leave the extensions of our work to this setting as future work.

\subsection{Causal effects}
\label{subsec: expcontrast}
Causal effects, within the potential outcomes framework, are defined as contrasts of each unit's potential outcomes under alternate assignments \citep{imbens2015causal}. As the number of possible contrasts grows exponentially with the number of distinct potential outcomes, we focus on two important special cases.

The first---which is well-defined regardless of the interference structure---compares the difference in the potential outcomes across two extreme scenarios: assigning every unit to treatment, $W_{1:n, 1:t} = 1_{1:n, 1:t}$, as opposed to control,  $W_{1:n, 1:t} = 0_{1:n, 1:t}$. 

\begin{definition} [Total effect at time $t$]
\label{def: taonc}
The total effect at time $t$ is
\begin{equation*}
    \tau^{TE}_t = \frac{1}{n}\sum_{i = 1}^{n}Y_{i, t}(1_{1:n, 1:t}) - \frac{1}{n}\sum_{i = 1}^{n}Y_{i, t}(0_{1:n, 1:t}).
\end{equation*}
\end{definition}
Our total effect at time $t$ corresponds to the Global Average Treatment Effect sometimes used in single time experiments \citep{ug2020randomized}. In the absence of interference and carryover effects, the total effect at time $t$ reduces to the usual average treatment effect at time $t$. 

The second---which requires Assumption \ref{assump:expmap}---provides a much richer class of causal effects with important practical applications. The TEC estimand is the generalization of the usual exposure contrast estimands \citep{aronow2017} to the panel experiment setting. Hereafter, the letter $k$ will always represent values in $\Delta$. 
\begin{definition} [Temporal exposure contrast (TEC)]
\label{def: expcontrast}
For any time step $t$ and exposure values $k, k^{'} \in \Delta$, we define the temporal exposure contrast between $k$ and $k^{'}$ to be
\begin{equation*}
    \tau^{k, k^{'}}_{t} = \frac{1}{n}\sum_{i = 1}^{n}Y_{i, t}(k) - \frac{1}{n}\sum_{i = 1}^{n}Y_{i, t}(k^{'})
\end{equation*}
\end{definition}

Here, we implicitly assume that for every unit, the potential outcome is well defined for all values of $k\in\Delta$. This assumption precludes situations where the range of the exposure mappings depends on $\mathcal{N}_i$. Further, note that if there exist carryover effects, then TEC may not be well-defined for the first few time steps. In this case, we may assume that all units are in the control group prior to the first time step in the panel experiment. 

In panel experiments, researchers are often less interested in the idiosyncratic effects at each point in time and instead focus on the temporal average causal effect that captures the intervention's average impact across both time and units \citep{boruvka2018assessing, bojinov19, bojinov2020panel, bojinov2022design, hu2022switchback, xiong2019optimal}. For example, \cite{bojinov19} are not interested in the relative difference between an algorithm or a human executing a large financial order on an arbitrary day of the experiment but are instead interested in the average difference across multiple trades on the same market. Similarly, technology companies like DoorDash \citep{tang2020control}, Lyft \cite{chamandy2016experimentation}, and Uber \cite{farronato2018innovation} use switchback design for panel experiments and consider average effect across time to make product decisions.  Such companies rarely are interested in the effect at time $t$—for instance, on Tuesday 2-3 pm—but instead, want to understand the average performance throughout the experiment.

\begin{definition} [Average total effect]
\label{def:ataonc}
The average total effect is
\begin{equation*}
    \overline{\tau^{TE}} = \frac{1}{T}\sum_{t = 1}^{T}\tau^{TE}_t.
\end{equation*}
\end{definition}
Similar to the total effect, in many applications, we are interested in the TEC's temporal average.  

\begin{definition} [Average temporal exposure contrast (ATEC)]
\label{def: avgexpcontrast}
For any exposure values $k, k^{'} \in \Delta$, we define the average temporal exposure contrast between $k$ and $k^{'}$ to be 
\begin{equation*}
    \bar{\tau}^{k, k^{'}} = \frac{1}{T}\sum_{t = 1}^{T}\tau^{k, k^{'}}_{t}
\end{equation*}
\end{definition}

Without assuming that the exposure mappings are time-invariant, the definition of the ATEC becomes more cumbersome as an exposure $k\in \Delta$ may be in the image of $f_{i,t}$ for some $t$, but not in the image of $f_{i,t'}$. That is, $Y_{i, t}(k)$ might be well-defined while $Y_{i, t'}(k)$ is not, which makes taking temporal averages difficult. 

Of course, our causal estimands are not exhaustive, and there are many other causal estimands of interest. For example, there is a vast literature in econometrics and statistics studying estimation and inference of spillover effects under either different designs or different model assumptions \citep{robins1986new,robins1999,leung2020treatment, bramoulle2020peer, VAZQUEZBARE2022}.

\subsection{Estimation and inference}

%
\subsubsection{The observed data}

For any choice of exposure mappings $\{f_{i,t}\}$, the observed assignment path 
$W_{1:n,1:t}$ induces the exposure $H_{i,t} = f_{i,t}(W_{1:n,1:t})$ for each $i$ and 
$t$; in particular, the assignment mechanism $\mathds{P}(W_{1:n,1:t})$ induces 
a distribution for the exposures $\mathds{P}(H_{i,t})$ for each $i$ and $t$. 
Under Assumption~\ref{assump:expmap}\footnote{We additionally assume that each unit fully complies with the assignment, leaving the relaxation of this assumption as future work.}, the observed outcomes $Y_{i,t}$ for unit $i$ 
at time $t$ can therefore be written:
\begin{equation*}
 Y_{i,t} = \sum_{k\in \Delta} \mathbf{1}(H_{i,t} = k) Y_{i,t}(k), 
 \quad \forall i \in {1, \cdots, n}, \forall t \in {1, \cdots, T},
\end{equation*}
%

We use these observed data to estimate the causal effects defined in \ref{subsec: expcontrast}.

\subsubsection{Estimation}
For the different interference structures studied in the following sections, 
we will rely on Horvitz-Thompson estimators \citep{ht1952}, or variations of it; e.g., to estimate $\tau^{k,k^{'}}_{t}$, we will use:
\begin{equation}
    \hat{\tau}^{k, k^{'}}_{t} = \frac{1}{n}\sum_{i = 1}^{n}\frac{\mathbf{1}(H_{i,t} = k)}{\mathbb{P}(H_{i, t} = k)}Y_{i, t} - \frac{1}{n}\sum_{i = 1}^{n}\frac{\mathbf{1}(H_{i, t} = k^{'})}{\mathbb{P}(H_{i, t} = k^{'})}Y_{i, t}. \label{est: HT_at_t}
\end{equation}
Taking the temporal average of \eqref{est: HT_at_t} provides a natural estimator of $\bar{\tau}^{k, k^{'}}$,
\begin{equation}
\hat{\bar{\tau}}^{k, k^{'}}_{} = \frac{1}{T}\sum_{t = 1}^{T}\hat{\tau}^{k, k^{'}}_{t}. \label{est: HT_avg}
\end{equation}
Similarly, if we let $h^{1}_{i, t} := f_{i, t}(1_{1:n, 1:t})$ and $h^{0}_{i, t} := f_{i, t}(0_{1:n, 1:t})$, then we can estimate total effect at time $t$ (c.f. Definition \ref{def: taonc}) by the following estimator
\begin{equation}
    \hat{\tau}^{TE}_{t} = \frac{1}{n}\sum_{i = 1}^{n}\frac{\mathbf{1}(H_{i, t} = h_{i, t}^{1})}{\mathbb{P}(H_{i, t} = h_{i, t}^{1})}Y_{i, t} - \frac{1}{n}\sum_{i = 1}^{n}\frac{\mathbf{1}(H_{i, t} = h_{i, t}^{0})}{\mathbb{P}(H_{i, t} = h_{i, t}^{0})}Y_{i, t}. \label{est: taonc_HT}
\end{equation}
Again, we have a natural estimator of average total effect induced by the above estimator
\begin{equation}
    \hat{\bar{\tau}}^{TE} = \frac{1}{T}\sum_{t = 1}^T\hat{\tau}^{TE}_{t}. \label{est: ataonc}
\end{equation}
The properties of these estimators are discussed in details in the rest of this 
manuscript.

\subsubsection{Randomization-based inference}

Throughout this paper, we adopt the randomization-based (sometimes called design-based) framework--- that is, we consider the potential outcomes as fixed, with the assignment being the only source of randomness. Equivalently, we can view the randomization-based framework as conditioning on the full set of potential outcomes and only using the randomness in the assignment for inference. This framework has seen a recent uptake in causal inference \citep{lin2013,li2017,liding2019, basse2020general} and has become standard for analyzing experiments with population interference \citep{aronow2017,svje2017average,basse2018,chin2018central} and unbounded carryover effects \citep{bojinov19,rambachan2019econometric,bojinov2020panel, bojinov2022design}. An additional benefit to adopting a randomization-based approach in the context of population interference is that it explicitly removes the challenges posed by correlations between the potential outcomes and the social relationships, as the potential outcomes are fixed (unknown) constants. 

There are two dominant inferential strategies within the randomization framework. The first is to use Fisher (conditional) randomization tests for sharp null hypotheses of 
no exposure effects, or for pairwise null hypotheses contrasting two exposures. While 
these tests deliver p-values that are exact and non-asymptotic, they are challenging to run with complex exposure mappings \citep{athey2018exact,basse2019randomization,puelz2019graph}. 

The second, which we focus on in this paper, is to construct confidence intervals based on the asymptotic distribution of our estimators, which can be used for testing if there is an effect on average. Intuitively, the asymptotic distribution represents a sequence of hypothetical randomized experiments in which either the number of units increases, the number of time steps increases, or both \citep{li2017, bojinov2020panel}. Within each step, we apply the analogous assignment mechanism, obtain the observed data, and compute our proposed estimand to estimate the causal effect of interest \citep{aronow2017, chin2018central}.  

Under the randomization framework, it is easy to show that the  
Horvitz-Thompson estimators $\hat{\tau}^{k, k^{'}}_{t}$, $\hat{\bar{\tau}}^{k, k^{'}}$, $\hat{\tau}^{TE}_{t}$ and $\hat{\bar{\tau}}^{TE}$ are unbiased for $\tau^{k, k^{'}}_{t}$, $\bar{\tau}^{k, k^{'}}$, $\tau^{TE}_{t}$ and $\bar{\tau}^{TE}$, respectively\footnote{For example, see \cite{bojinov19} and \cite{aronow2017} for explicit proof.}; obtaining central limit theorems in this setting, however, is 
notoriously challenging. In the next two sections, we develop such results for the 
above four estimators under different experiment assumptions.

\section{Panel experiments with population interference and no carryover 
effects}
\label{sec: panel}
Panel experiments are particularly helpful when there is population interference but no carryover effects--- a setting we refer to as puper population interference. This situation commonly occurs when the treatment has a relatively short-lived effect; for example, as is the case for most digital experiments on networks and platforms (see the discussion in \cite{kohavi2020trustworthy, bojinov2022online}). In this setting, inference for the TEC at a fixed point in time is equivalent to the standard experimental setup (see Appendix~\ref{sec: cross-sec} for a full discussion of this setting, including a new central limit theorem that illustrated a fundamental trade-off between the interference structure and the design of the experiment). 

We can leverage time and move beyond the standard experimental setup in two ways. First, we can focus on inference for the ATEC that captures the average effect across units and time. Here, we show that varying the treatment over time allows us to handle settings with more expansive population interference. Second, we revisit the TEC and provide a variance reduction technique that leverages a stability assumption that limits the change in the potential outcomes across time for the same unit receiving the same treatment. Together, our results demonstrate the potential of leveraging panel experiments when there are no carryover effects.



\subsection{Average temporal exposure contrast}
\label{subsec: atec}

There are three distinct asymptotic regimes when considering inference for the ATEC and its natural estimator $\hat{\bar{\tau}}^{k,k'}$: (1) $T$ fixed and $n\rightarrow \infty$; (2) $T \rightarrow \infty$ and $n \rightarrow \infty$; (3) $T \rightarrow \infty$ and $n$ fixed. An important insight from this section is that inference in these three regimes requires different constraints on the population interference mechanism. Roughly speaking, the larger $T$ is relative to $n$, the more interference we can tolerate. 

\subsubsection{Assumptions}

For most of our central limit theorem results, we require a notion of the dependency graph for a collection of random variables. We define the dependency graph $G_{n,t}$ for $H_{1,t}, \cdots, H_{n,t}$ to be the graph with vertices $V_t = \{1, \cdots, n\}$ and edges $E_t$ such that $(i, j) \in E$ if and only if $H_{i,t}$ and $H_{j,t}$ are not independent. The graph $G_n$ models the dependency relationship among $n$ random variables $H_{1,t}, \cdots, H_{n,t}$. Notice that the dependency graph depends both on the interference structure and on the assignment mechanism. Finally, denote by $d_{n,t}$ the maximal number of dependent exposure values for any unit $i$ at time step $t$ and let $d_n = \limsup_{t \rightarrow \infty}d_{n,t}$ with the convention that for fixed $T$, $d_n = \max\{d_{n,1}, \cdots, d_{n,T}\}$. See Appendix~\ref{subsec: design trade-off} for the derivation of $d_{n,t}$ in several specific contexts. 

Throughout this subsection, we work exclusively with temporally independent assignment mechanisms, see Definition~\ref{def temporal indep assignment}.
Our central limit theorems also require three additional assumptions. The first two assumptions bound the potential outcomes and the inverse probabilities of exposure.
\begin{assumption} [Uniformly bounded potential outcomes]
Assume that all the potential outcomes are uniformly bounded, i.e., $|Y_{i,t}(k)| \leq M$ for some $M$ and for all $i\in\{1,n\}$, $k\in \Delta$, and $t\ge 1$. \label{assump: bddpo}
\end{assumption}
\begin{assumption} [Overlap]
Assume all the exposure probabilities are bounded away from 0 and 1, i.e., $\exists \eta > 0$ such that for all $i\in\{1,n\}$, $k\in \Delta$, and $t\ge 1$, $0 < \eta \leq \pi_i(k) \leq 1 - \eta < 1$. \label{assump: overlap}
\end{assumption}
Assumptions \ref{assump: bddpo} and \ref{assump: overlap} are standard in the causal inference literature (\cite{aronow2017}; \cite{leung2022}). Assumption \ref{assump: bddpo} holds in most practical applications as realizations of the outcome variables are almost always bounded. Assumption \ref{assump: overlap} is necessary as vanishing exposure probabilities make the causal question ill-defined as we cannot observe the associated potential outcomes.

The next assumption rules out the existence of a pathological subsequence $n_k$ along which the limiting variance of our estimator is zero.
\begin{assumption} [Nondegenerate asymptotic variance]
Assume that $\liminf_{n \rightarrow \infty}\text{Var}(\sqrt{n}\hat{\tau}^{k, k^{'}}_{t}) > 0$ for any $t$. \label{assump: nondegeneratevar}
\end{assumption}
As a consequence of this assumption, for each $t$, there exists a constant $c > 0$ such that $\text{Var}(\sqrt{n}\hat{\tau}^{k, k^{'}}_{t}) \geq c$ for all sufficiently large $n$. This type of assumption seems unavoidable, even in settings without interference (see, e.g., Corollary~1 in \cite{guo2020generalized}, and subsequent discussion).
\subsubsection{Three central limit theorems}
We now state and discuss our three new central limit theorems for the ATEC under population interference but no carryover effects. 

\begin{theorem} 
\label{thm: temporalavgclt_fixedT}
Suppose we have pure population interference and a temporally independent assignment mechanism, then for any $T$, under Assumption \ref{assump:expmap}-\ref{assump: nondegeneratevar} and the condition that $d_n = o(n^{1/4})$, we have
\begin{equation*}
    \frac{\sqrt{nT}(\hat{\bar{\tau}}^{k, k^{'}} - \bar{\tau}^{k, k^{'}})}{\sqrt{\frac{1}{T}\sum_{t=1}^{T}\sigma^2_{n, t}}} \xrightarrow{d} \mathcal{N}(0, 1),
\end{equation*}
as $n \rightarrow \infty$, where $\sigma_{n, t}^2 = \text{Var}(\sqrt{n}\hat{\tau}^{k, k^{'}}_{t})$.
\end{theorem}
This first theorem states a central limit theorem for the regime where $T$ is fixed and $n \rightarrow \infty$, making it relevant for applications where $n$ is much larger than $T$. The assumption that $d_n = o(n^{1/4})$ quantifies the dependence among observations due to interference. If we compare this result to the analogous CLT in the non-temporal experimental setup with interference, Theorem~\ref{thm: tec_clt_pure_spat_inter} in Appendix~\ref{sec: cross-sec}, we have the same requirement, namely $d_n = o(n^{1/4})$. Intuitively, this is because this asymptotic regime is closest to the standard setting with no temporal dimension—--any finite number of time periods $T$ is negligible compared with infinitely many observations $n$. 

At the other extreme, we consider the regime where $T\rightarrow \infty$ and 
$n$ is fixed:
\begin{theorem}
\label{thm: ateclyapunov}
Suppose we have pure population interference, a temporally independent assignment mechanism, and Assumptions \ref{assump:expmap}, \ref{assump: bddpo}, \ref{assump: overlap} are satisfied. Let $\sigma_{n, t}^2 = \text{Var}(\sqrt{n}\hat{\tau}^{k, k^{'}}_{t})$, we further assume that $\frac{1}{T}\sum_{t=1}^{T}\sigma^2_{n, t}$ is bounded away from 0 for any $T$. We then have
\begin{equation*}
    \frac{\sqrt{nT}(\hat{\bar{\tau}}^{k, k^{'}} - \bar{\tau}^{k, k^{'}})}{\sqrt{\frac{1}{T}\sum_{t=1}^{T}\sigma^2_{n, t}}} \xrightarrow{d} \mathcal{N}(0, 1),
\end{equation*}
as $T \rightarrow \infty$.
\end{theorem}
This central limit theorem makes no assumption whatsoever on the interference 
mechanism, beyond assuming that there are no carryover effects: in particular, we allow 
a unit's outcome to depend on any other unit's assignment. This perhaps 
surprising fact sheds some light into the nature of inference for the 
ATEC, and how it differs from the TEC. Intuitively, a central limit 
theorem requires enough ``nearly independent'' observations: this means 
that even if at any time step $t$, the observations are all correlated, we 
can still have infinitely many independent observations if: (1) 
observations are uncorrelated across time and (2) we observe infinitely 
many time periods.

The next theorem formalizes this intuition, by making the trade-off between the 
growth rates of $T$ and $d_n$ explicit:
\begin{theorem}
\label{thm: atecclt}
Suppose we have pure population interference, a temporally independent assignment mechanism, and Assumption \ref{assump:expmap}-\ref{assump: nondegeneratevar} are satisfied, then for $T = T(n)$ such that either 
\begin{equation}
    \frac{n}{T} \rightarrow 0 \label{atecclt: nodegcond}
\end{equation}
or
\begin{equation}
    \frac{\min\{d_n^2, n\}}{\sqrt{nT}} \rightarrow 0 \label{atecclt: degcond}
\end{equation}
holds, we have
\begin{equation*}
    \frac{\sqrt{nT}(\hat{\bar{\tau}}^{k, k^{'}} - \bar{\tau}^{k, k^{'}})}{\sqrt{\frac{1}{T}\sum_{t=1}^{T}\sigma^2_{n, t}}} \xrightarrow{d} \mathcal{N}(0, 1),
\end{equation*}
as $n \rightarrow \infty$, where $\sigma_{n, t}^2 = \text{Var}(\sqrt{n}\hat{\tau}^{k, k^{'}}_{t})$.
\end{theorem}
Condition (\ref{atecclt: nodegcond}) is actually a special case of 
condition (\ref{atecclt: degcond}): if we do not impose any 
assumptions on the interference, $\min\{d_n^2, n\}$ is just $n$, so we 
need $\frac{n}{\sqrt{nT}} \rightarrow 0$, which is equivalent to 
$\frac{n}{T} \rightarrow 0$. Condition~\ref{atecclt: degcond} gives 
us more subtle control over the rate of growth required of $T$ for any given level 
of interference. For instance, while for finite $T$ we would require $d_n = o(n^{1/4})$, if $T$ grows as $T(n) = \sqrt{n}$ we only require 
$d_n = o(n^{1/2})$. As with the previous theorem, the intuition behind this 
result is that as $d_n$ becomes larger, the number of ``nearly independent'' 
observations at each time point shrinks --- this must be counterbalanced by 
an increase in the the number of temporal observation, i.e, an increase in 
the rate of $T = T(n)$.

\subsubsection{Variance bounds and estimation}

Unfortunately, as is typical in finite population causal inference, $\text{Var}(\hat{\tau}^{k, k^{'}}_t)$ contains terms that are products of potential outcomes that can never be simultaneously observed from a single experiment, making it non-identifiable \citep{basse2020general}. Instead, researchers derive an upper bound to the variance and compute unbiased estimates for this bound, allowing them to conduct conservative inference (i.e., derive confidence intervals with higher coverage than the nominal level). Without making assumptions on the assignment mechanism, we can obtain a simple bound by replacing all non-observable products of potential outcomes with the sum of their squares, we denote the estimate of the bound by $\widehat{\text{Var}}(\sqrt{n}\hat{\tau}_t^{k, k^{'}})$. The specific expression can be found in the following proposition:
\begin{proposition} (Estimator of variance) 
\label{prop: varest}
Let,
\small
\begin{equation*}
\begin{split}
    \widehat{\text{Var}}(\sqrt{n}\hat{\tau}^{k, k^{'}}_t) = \frac{1}{n}\bigg\{\sum_{i = 1}^{n}\mathbf{1}(H_{i,t} = k)(1 - \pi_{i,t}(k))\left[\frac{Y_{i,t}}{\pi_{i,t}(k)}\right]^2
    + \sum_{i = 1}^{n}\sum_{j \neq i, \pi_{ij}(k) = 0}\left[\frac{\mathbf{1}(H_{i,t} = k)Y_{i,t}^2}{2\pi_{i,t}(k)} + \frac{\mathbf{1}(H_{j,t} = k)Y_{j,t}^2}{2\pi_{j,t}(k)}\right] \\
    + \sum_{i = 1}^{n}\sum_{j \neq i, \pi_{ij}(k) > 0}\mathbf{1}(H_{i,t} = k)\mathbf{1}(H_{j,t} = k)
    \times \frac{\pi_{ij}(k) - \pi_{i,t}(k)\pi_{j,t}(k)}{\pi_{ij}(k)}\frac{Y_{i,t}}{\pi_{i,t}(k)}\frac{Y_{j,t}}{\pi_{j,t}(k)} \\
    + \sum_{i = 1}^{n}\mathbf{1}(H_{i,t} = k^{'})(1 - \pi_{i,t}(k^{'}))\left[\frac{Y_{i,t}}{\pi_{i,t}(k^{'})}\right]^2
    + \sum_{i = 1}^{n}\sum_{j \neq i, \pi_{ij}(k^{'}) = 0}\left[\frac{\mathbf{1}(H_{i,t} = k^{'})Y_{i,t}^2}{2\pi_{i,t}(k^{'})} + \frac{\mathbf{1}(H_{j,t} = k^{'})Y_{j,t}^2}{2\pi_{j,t}(k^{'})}\right] \\
    + \sum_{i = 1}^{n}\sum_{j \neq i, \pi_{ij}(k^{'}) > 0}\mathbf{1}(H_{i,t} = k^{'})\mathbf{1}(H_{j,t} = k^{'})
    \times\frac{\pi_{ij}(k^{'}) - \pi_{i,t}(k^{'})\pi_{j,t}(k^{'})}{\pi_{ij}(k^{'})}\frac{Y_{i,t}}{\pi_{i,t}(k^{'})}\frac{Y_{j,t}}{\pi_{j,t}(k^{'})} \\
    -2\sum_{i = 1}^{n}\sum_{j \neq i, \pi_{ij}(k, k^{'}) > 0}\left(\pi_{ij}(k, k^{'}) - \pi_{i,t}(k)\pi_{j,t}(k^{'})\right)
    \times\frac{\mathbf{1}(H_{i,t} = k)\mathbf{1}(H_{j,t} = k^{'})}{\pi_{ij}(k, k^{'})}\frac{Y_{i,t}}{\pi_{i,t}(k)}\frac{Y_{j,t}}{\pi_{j,t}(k^{'})} \\
    +2\sum_{i = 1}^{n}\sum_{j \neq i, \pi_{ij}(k, k^{'}) = 0}\left[\frac{\mathbf{1}(H_{i,t} = k)Y_{i,t}^2}{2\pi_{i,t}(k)} + \frac{\mathbf{1}(H_{j,t} = k^{'})Y_{j,t}^2}{2\pi_{j,t}(k^{'})}\right]\bigg\},
\end{split}
\end{equation*}
\normalsize
where $\pi_{i}(k) = mathbb{P}(H_{i,t} = k)$ and $\pi_{i,j} = \mathbb{P}(H_{i,t} = k \& H_{j,t} = k)$. Then $\mathbb{E}\left[\widehat{\text{Var}}(\sqrt{n}\hat{\tau}^{k, k^{'}})\right] \geq \text{Var}(\sqrt{n}\hat{\tau}^{k, k^{'}})$.
\end{proposition}
We drop the subscript $t$ to ease notations. With the above proposition and central limit theorems, inference proceeds as follows:
\begin{proposition}
\label{prop: ci_atec}
Suppose Theorem \ref{thm: temporalavgclt_fixedT} or \ref{thm: atecclt} holds, then for any $\delta > 0$,
\begin{equation*}
    \mathbb{P}\left(\bar{\tau}^{k, k^{'}} \in \left[\hat{\bar{\tau}}^{k, k^{'}} - \frac{z_{1 - \frac{\alpha}{2}}}{\sqrt{1 - \delta}}\sqrt{\frac{1}{T^2}\sum_{t = 1}^{T}\widehat{\text{Var}}(\hat{\tau}^{k, k^{'}}_t)}, \hat{\tau}^{k, k^{'}} + \frac{z_{1 - \frac{\alpha}{2}}}{\sqrt{1 - \delta}}\sqrt{\frac{1}{T^2}\sum_{t = 1}^{T}\widehat{\text{Var}}(\hat{\tau}^{k, k^{'}}_t)}\right]\right) \geq 1 - \alpha
\end{equation*}
for large $n$. Moreover, suppose Theorem \ref{thm: ateclyapunov} holds, then for any $\delta > 0$,
\begin{equation*}
    \mathbb{P}\left(\bar{\tau}^{k, k^{'}} \in \left[\hat{\bar{\tau}}^{k, k^{'}} - \frac{z_{1 - \frac{\alpha}{2}}}{\sqrt{1 - \delta}}\sqrt{\frac{1}{T^2}\sum_{t = 1}^{T}\widehat{\text{Var}}(\hat{\tau}^{k, k^{'}}_t)}, \hat{\tau}^{k, k^{'}} + \frac{z_{1 - \frac{\alpha}{2}}}{\sqrt{1 - \delta}}\sqrt{\frac{1}{T^2}\sum_{t = 1}^{T}\widehat{\text{Var}}(\hat{\tau}^{k, k^{'}}_t)}\right]\right) \geq 1 - \alpha
\end{equation*}
for large $T$.
\end{proposition}
The proof of the above proposition builds on the proof of Proposition~\ref{prop: ci} in Appedix~\ref{sec: cross-sec}.

The variance bound and inference for the regime when $T\to\infty$ is identical to what is given in \cite{bojinov19} and is therefore omitted. 
\subsection{Shrinkage estimator under stability assumption}
\label{subsec: epsilonstab}

We now focus on deriving a better estimator of the TEC, at a fixed point in time. Suppose a unit receives the same treatment for two consecutive periods. If the potential outcomes are similar across time, then we can borrow information from past outcomes to reduce the variance of our eastimate of the TEC. Intuitively, this section provides a bias-variance trade-off, where we introduce some bias in our inference for a potentially substantial reduction in the variance. 
For a specific treatment, the similar across time assumption can be formalized as follows.


\begin{assumption} [Weak stability of potential outcomes]
\label{def: weakstability}
We say the potential outcome matrix $Y_{i, t}, i = 1, \cdots, N$, $t = 1, \cdots, T$ is $\epsilon$-weakly stable if for each $i$ and exposure value $k$, we have $|Y_{i, t}(k) - Y_{i, t+1}(k)| \leq \epsilon, \forall t \in \{1, \cdots, T - 1\}$. If we further assume that $\epsilon = 0$, we then say that the potential outcome matrix is strongly stable.
\end{assumption}
All results in this section easily generalize to the case where the uniform bound $\epsilon$ is replaced by a time dependent bound $\epsilon_t$. Throughout, we focus on the estimation of the total effect at time $t$ as an example to illustrate how we can leverage temporal information under weak 
stability.

Under pure population interference and time-invariant exposure mappings,
\begin{align}
    \tau_t^{TE} = \frac{1}{n}\sum_{i = 1}^{n}Y_{i, t}(h_i^1) - \frac{1}{n}\sum_{i = 1}^{n}Y_{i, t}(h_i^0), \label{equa: tau_t}
\end{align}
where $h_i^1 = f_i(1_{t, 1:n})$ and $h_i^0 = f_i(0_{t, 1:n})$. 

To build some intuition, we first describe how to leverage a single past time period, $t'= t-1$ to improve estimation at time $t$. The idea is that 
by considering a convex combination $\hat{\tau}^c_t = \alpha\hat{\tau}_t^{TE} + (1 - \alpha)\hat{\tau}_{t - 1}^{TE}$,
for some $\alpha \in [0, 1]$ as an estimator of $\hat{\tau}_t^{TE}$, we 
introduce some bias but reduce the variance --- the hope being that under weak 
stability, the bias introduced will be modest compared to the reduction in variance. 
This is formalized in the following proposition.
\begin{proposition} [Bound on the bias of $\hat{\tau}^c_t$]
\label{prop: k=2cvxbias}
\begin{equation}
    |\mathbb{E}[\hat{\tau}^c_t] - \tau_t^{TE}| \leq 2(1 - \alpha)\epsilon
\end{equation}
\end{proposition}
As we can see, the absolute bias of $\hat{\tau}_t^c$ is bounded by a quantity that grows linearly with $\epsilon$: if $\epsilon$ is very small, then so will the maximum bias. In particular, $\hat{\tau}_t^c$ is unbiased for $\tau_t^{TE}$ if $\epsilon = 0$, which corresponds to the somewhat unrealistic assumption that the potential outcomes do not vary across time. Under some conditions, it can be guaranteed that the gain in bias is more than counterbalanced by a reduction in variance, making it a worthwhile trade-off in terms of the mean squared error (MSE).

\begin{proposition}
Suppose $Var(\hat{\tau}_t^{TE}) > Cov(\hat{\tau}_t^{TE}, \hat{\tau}_{t-1}^{TE})$, then there exists some $\alpha \in (0, 1)$ such that $\hat{\tau}^c_t = \alpha\hat{\tau}_t^{TE} + (1 - \alpha)\hat{\tau}_{t-1}^{TE}$ has lower MSE than $\hat{\tau}_t^{TE}$. Moreover, if we have $Var(\hat{\tau}_t^{TE}) - Var(\hat{\tau}_{t-1}^{TE}) > 4\epsilon^2$ then we know that $\hat{\tau}^c_t = \frac{1}{2}\hat{\tau}_t^{TE} + \frac{1}{2}\hat{\tau}_{t-1}^{TE}$ has lower MSE than $\hat{\tau}_t^{TE}$. \label{prop: msereduction, k=2}
\end{proposition}
In other words, if the current variance is larger, by choosing some $\alpha$, the convex combination type estimator would give us a better estimator in terms of MSE. Moreover, as the proposition suggests, if we know the difference is bigger than $4\epsilon^2$, we know that $\alpha = \frac{1}{2}$ is sufficient.

If we further assume that the assignment mechanism is temporally independent (Definition~\ref{def temporal indep assignment}), then $\text{Cov}(\hat{\tau}_t^{TE}, \hat{\tau}_{t-1}^{TE}) = 0$, hence we have the following result. 
\begin{proposition}
Suppose that the assignments mechanism is temporally independent, then there exists some $\alpha \in (0, 1)$ such that $\hat{\tau}^c_t = \alpha\hat{\tau}_t^{TE} + (1 - \alpha)\hat{\tau}_{t-1}^{TE}$ has lower MSE than $\hat{\tau}_t^{TE}$. The optimal $\alpha$ is given by $\alpha = 1 - \frac{Var(\hat{\tau}_t^{TE})}{4\epsilon^2 + Var(\hat{\tau}_t^{TE}) + Var(\hat{\tau}_{t - 1}^{TE})}$. \label{prop: msereductionwithindtreat}
\end{proposition}

\begin{algorithm}[t]
    \caption{Algorithm to estimate $\epsilon$}
    \label{alg: epsilon}
  \begin{algorithmic}[1]
    \STATE Initialize $\hat{\epsilon} = 0$
    \STATE For $t$ = 1 to $T - 1$:
    \begin{itemize}
        \item [(a)] For $i = 1, 2,\cdots,n$ compute $h_{i, t}$ and $h_{i, t + 1}$.
        \item [(b)] If $h_{i, t} = h_{i, t+1} = k$, compute $\epsilon_{i, t} = |y_{i, t} - y_{i, t+1}|$. 
        \item [(c)] If $\epsilon_{i, t} > \hat{\epsilon}$, set  $\hat{\epsilon} = \epsilon_{i, t}$.
    \end{itemize}
    \STATE Output $\hat{\epsilon}$.
  \end{algorithmic}
\end{algorithm}

Under the $\epsilon-$stability assumption, Algorithm \ref{alg: epsilon} 
provides a data dependent approach to estimate $\epsilon$ and allows us to obtain estimate $\hat{\alpha}$ of the weight parameter 
$\alpha$,

\begin{equation*}
    \hat{\alpha} = 1 - \frac{\widehat{\text{Var}}(\hat{\tau}_t^{TE})}{\widehat{\text{Var}}(\hat{\tau}_t^{TE}) + \widehat{\text{Var}}(\hat{\tau}_{t^{'}}^{TE}) + 4(t - t^{'})^2\hat{\epsilon}^2},
\end{equation*}
where $\widehat{\text{Var}}(\hat{\tau}_t^{TE})$ can be any estimator of 
the variance $\text{Var}(\hat{\tau}_t^{TE})$: we discuss a few options in Proposition \ref{prop: totaleffectvarest} of Appendix \ref{appA}. 
In addition, under pure population interference and temporally independent assignments, 
\begin{align*}
    \text{Var}(\hat{\tau}_t^c) &= \text{Var}\bigg(\alpha\hat{\tau}_t^{TE} + (1 - \alpha)\hat{\tau}_{t - 1}^{TE}\bigg) \\
    &= \alpha^2\text{Var}(\hat{\tau}_t^{TE}) + (1 - \alpha)^2\text{Var}(\hat{\tau}_{t - 1}^{TE}),
\end{align*}
which suggests the following plug-in estimator of the variance:
\begin{equation*}
    \widehat{\text{Var}}(\hat{\tau}_t^c) = \hat{\alpha}^2\widehat{\text{Var}}(\hat{\tau}_t^{TE}) + (1 - \hat{\alpha})^2\widehat{\text{Var}}(\hat{\tau}_{t - 1}^{TE}).
\end{equation*}
We also give the expression of $\text{Cov}(\hat{\tau}_t^{TE}, \hat{\tau}_{t-1}^{TE})$ and an estimator for it in Proposition~\ref{prop: var_and_cov_of_tauhat} and Proposition~\ref{prop: totaleffectcovest} of Appendix~\ref{appA} respectively. Equipped with the variance and the covariance estimators, we can directly check the condition in Proposition~\ref{prop: totaleffectcovest}. The optimal $\alpha$ is given in the proof and can be estimated in the similar way as in the independent assignment case.

We can easily generalize the above discussion to a general version of this estimator such that we combine $\hat{\tau}_t^{TE}$ and $\hat{\tau}_{t^{'}}^{TE}$ for arbitrary $t^{'} < t$. The above results trivially generalize and are therefore omitted for brevity.

%

%
%
Based on the variance estimator above, we propose two ways to construct confidence intervals. The first one ignores the bias of $\hat{\tau}_t^c$ and uses Gaussian confidence interval. The second one takes advantage of Chebyshev's inequality and incorporates the bias. Specifically, note that
\begin{equation*}
    \mathbb{P}(|\hat{\tau}_t^c - (\mathbb{E}[\hat{\tau}_t^c] - \tau_t^{TE}) - \tau_t^{TE}| \geq \epsilon) \leq \frac{\text{Var}(\hat{\tau}_t^c)}{\epsilon^2},
\end{equation*}
hence $\forall \delta > 0$, 
\begin{equation*}
\mathbb{P}\left(\tau_t^{TE} \in \left[\hat{\tau}_t^c - (\mathbb{E}[\hat{\tau}_t^c] - \tau_t^{TE}) - \epsilon, 
\hat{\tau}_t^c - (\mathbb{E}[\hat{\tau}_t^c] - \tau_t^{TE}) + \epsilon\right]\right) \geq 1 - \delta
\end{equation*}
for $\epsilon = \sqrt{\frac{\text{Var}(\hat{\tau}_t^c)}{\delta}}$. Let $b(\hat{\tau}_t^c) = \mathbb{E}[\hat{\tau}_t^c] - \tau_t^{TE} = (1 - \alpha)(\tau_{t - 1}^{TE} - \tau_t^{TE})$ be the bias of our convex combination estimator. If we estimate $b(\hat{\tau}_t^c)$ by $\hat{b}(\hat{\tau}_t^c) = (1 - \hat{\alpha})(\hat{\tau}_{t - 1}^{TE} - \hat{\tau}_t^{TE})$, then we can use the following interval as an approximate $(1 - \delta)$-level confidence interval of $\tau_t^{TE}$:
\begin{equation*}
    \left[\hat{\tau}_t^c - \hat{b}(\hat{\tau}_t^c) - \sqrt{\frac{\widehat{\text{Var}}(\hat{\tau}_t^c)}{\delta}}, \hat{\tau}_t^c - \hat{b}(\hat{\tau}_t^c) + \sqrt{\frac{\widehat{\text{Var}}(\hat{\tau}_t^c)}{\delta}}\right].
\end{equation*}
We explore empirically the coverage of the above approximate confidence 
intervals with a simulation study in Section~\ref{sec: simulations}.

The approach we have described in this section naturally extends to using the $k - 1$ 
previous time steps, yielding the weighted combination estimator:
\begin{equation*}
    \hat{\tau}_t^c = \alpha_1\hat{\tau}_{t-k+1}^{TE} + \cdots + \alpha_k\hat{\tau}_{t}^{TE},
\end{equation*}
where $\alpha_1, \ldots, \alpha_k$ can be estimated by solving a slightly more involved 
convex optimization problem. We describe this approach in full details in Appendix \ref{suppl_sec_b}.

\section{Panel experiments with population interference and carryover effects}
\label{subsecsec: mixed_int}

Section \ref{sec: panel} shows that adding a temporal dimension is particularly useful if no carryover effects exist. We now consider the setting where there are both carryover effects and population interference, which we call mixed interference. As we show below, mixed interference affects our ability to draw inference both for the TEC and ATEC, albeit in different ways. For temporal exposure contrasts (TEC), the same theorem as in Section 
\ref{sec: cross-sec} holds. 
\begin{theorem}
\label{thm: tec_mixed_inter}
Assume we have a temporally independent assignment mechanism, under Assumption \ref{assump:expmap}-\ref{assump: nondegeneratevar} and the condition that $d_n = o(n^{1/4})$, we have
\begin{equation*}
    \frac{\sqrt{n}(\hat{\tau}^{k, k^{'}}_t - \tau^{k, k^{'}}_t)}{\text{Var}(\sqrt{n}\hat{\tau}^{k, k^{'}}_t)^{1/2}} \xrightarrow{d} \mathcal{N}(0, 1)
\end{equation*}
as $n\to\infty$.
\end{theorem}
The difference with the pure population setting, discussed in Appendix~\ref{sec: cross-sec}, is not mathematical but conceptual: in the mixed setting, the exposures involve the assignments over previous time steps. Consequently, there are generally many more exposures than in the pure population setting, and each unit has a lower probability of receiving each. This leads to Horvitz-Thompson estimators with a much larger variance.

For the average temporal exposure contrast, the difference between population 
interference and mixed interference is starker. The main difficulty is that mixed 
interference breaks the temporal independence that powered the results of 
Section~\ref{subsec: atec}. 

Nevertheless, we can still establish a general central limit theorem, be it under some additional assumptions. In particular, we require a restriction on the rate at which the variance shrinks.
\begin{assumption}
\label{assump: nonvanishvarmixed}
Assume that 
\begin{equation*}
\liminf_{n \rightarrow \infty}\text{Var}(\sqrt{nT}\hat{\bar{\tau}}^{k, k^{'}}) \geq \epsilon > 0
\end{equation*}
for some $\epsilon$.
\end{assumption}
This technical assumption rules out the pathological case that the variance vanishes as $n \rightarrow \infty$. To state our theorem, we also need to introduce the notion of $s$-dependent sequence. 
\begin{definition}[$s$-dependant sequences]
We say that $\{H_{i, t}\}_{i = 1}^{n}$ is an $s$-dependent sequence of random variables if and only if for any index set $I, J \subseteq 1,\dots,n$, $\{H_{i, t}\}_{i = 1}^{n}$ and $\{H_{i, t}\}_{i = 1}^{n}$ are independent so long as $\min_{j \in J}j - \max_{i \in I}i > s$.    
\end{definition}
Intuitively, this assumption limits the depends across units. 
\begin{theorem} \label{thm: mixedgeneral}
Assume we have a temporally independent assignment mechanism, under Assumption \ref{assump:expmap}-\ref{assump: nondegeneratevar} and Assumption \ref{assump: nonvanishvarmixed}, suppose $\{H_{i, t}\}_{i = 1}^{n}$ is an $s$-dependent sequence of random variables for a fixed $t$ and LEA$(p)$ assumption is satisfied with some finite $p$. If $s, n, T$ are such that $s^5T^4 = o(n^{1 - \alpha})$ for some $0 < \alpha < 1$, then we have that
\begin{equation*}
    \frac{\sqrt{nT}(\hat{\bar{\tau}}^{k, k^{'}} - \bar{\tau}^{k, k^{'}})}{\sqrt{\text{Var}(\sqrt{nT}\hat{\bar{\tau}}^{k, k^{'}})}} \xrightarrow{d} \mathcal{N}(0, 1)
\end{equation*}
as $n \rightarrow \infty$.
\end{theorem}

 The above theorem requires general assumptions on exposure values as well as the asymptotic variance though independent assignments are not required. However, it is somewhat difficult to apply this result to practical settings. Therefore, we now focus on a specific setting to illustrate the type of results that can be derived under mixed interference.
\subsection{Stratified interference}

Consider the following natural temporal extension of the stratified interference (\cite{hudgens2008}; \cite{basse2018})
setting (\cite{hudgens2008}; \cite{basse2018}):
\begin{equation*}
    f_{i, t}(w_{1:n, 1:t}) = f(w_{i, t-1}, w_{i, t}, \{w_{j, t-1}\}_{j \in \mathcal{N}_i, j \neq i}, \{w_{j, t}\}_{j \in \mathcal{N}_i, j \neq i})
\end{equation*}
where $\mathcal{N}_i$ is the group to which unit $i$ belongs. For convenience, 
we fix each group to be of size $r$ \footnote{This ensures that each unit is associated with exactly the same set of exposure values so that the exposure contrast between two arbitrary exposure values is well-defined.}.

\begin{theorem}
\label{thm: groupclt}
With the above setting and temporally independent assignments, under Assumption \ref{assump:expmap}-\ref{assump: nondegeneratevar} and Assumption \ref{assump: nonvanishvarmixed}, suppose $n, r, T$ are such that $r = o((nT)^{\frac{1}{4}})$, then we have that
\begin{equation*}
    \frac{\sqrt{nrT}(\hat{\bar{\tau}}^{k, k^{'}} - \bar{\tau}^{k, k^{'}})}{\sqrt{\text{Var}(\sqrt{nrT}\hat{\bar{\tau}}^{k, k^{'}})}} \xrightarrow{d} \mathcal{N}(0, 1)
\end{equation*}
as $n \rightarrow \infty$.
\end{theorem}

The theorem holds for heterogeneous group sizes as long as $\max_i r_i = o((nT)^{\frac{1}{4}})$ where $r_i = |\mathcal{N}_i|$ is the size of the group unit $i$ belongs to.
To do inference, we consider a specific example of stratified interference:
\begin{equation*}
    f_{i, t}(w_{1:n, 1:t}) = \left(w_{i, t-1}, w_{i, t}, \sum_{j \in \mathcal{N}_i, j \neq i}w_{j, t-1}, \sum_{j \in \mathcal{N}_i, j \neq i}w_{j, t}\right).
\end{equation*}
We focus on the Bernoulli design where each unit is independently assigned to treatment with probability $\frac{1}{2}$. We consider the exposures $k = (1,1,r-1,r-1)$ and 
$k' = (0,0,0,0)$. Such exposure contrast is exactly the same as the total effect since essentially we are comparing the world of everyone getting treatment to the world of everyone getting control. Notice that in this case, $r$ cannot be infinite, otherwise the overlap assumption would be violated. To ease notations, we index each unit $i$ by a tuple $(l, q)$, meaning that unit $i$ is the $q$-th unit in the $l$-th group\footnote{If we use a tuple $(l, q)$ to represent the $q-$th unit in the $l-$th household, then we note by passing that $0 < C_1 \leq Y_{(l, q), t}(k) \leq C_2$ for all $l, q, t, k$ for some $C_1, C_2$ is sufficient for Assumption~\ref{assump: nonvanishvarmixed}.}.

\begin{proposition} \label{prop: householdclt-ci}
Assuming the above setup, let  $X_{n,t}=\sqrt{n}(\hat{\tau}^{k, k^{'}}_{t} - \tau^{k, k^{'}}_t)$, then we can estimate the asymptotic variance by 
\begin{equation*}
\widehat{B_n}^2 =  \sum_{t = 1}^{T}\widehat{\text{Var}}(X_{n, t}) + 2\sum_{t = 1}^{T - 1}\widehat{\text{Cov}}(X_{n, t}, X_{n, t + 1}),
\end{equation*}
where 
\begin{multline}
    \widehat{\text{Var}}(X_{n, t}) = \frac{1}{nrT}\left[\sum_{l = 1}^{n}\sum_{q = 1}^{r}(2^{2r} - 1)\frac{\mathbf{1}(H_{(l, q), t} = k)Y_{(l, q), t}^2}{\mathbb{P}(H_{(l, q), t} = k)}
    + \sum_{l = 1}^{n}\sum_{q = 1}^{r}(2^{2r} - 1)\frac{\mathbf{1}(H_{(l, q), t} = k^{'})Y_{(l, q), t}^2}{\mathbb{P}(H_{(l, q), t} = k^{'})} \right.\\
    + \sum_{l = 1}^n\sum_{q = 1}^{r}\left(\frac{\mathbf{1}(H_{(l, q), t} = k)Y_{(l, q), t}^2}{\mathbb{P}(H_{(l, q), t} = k)} + \frac{\mathbf{1}(H_{(l, q), t} = k^{'})Y_{(l, q), t}^2}{\mathbb{P}(H_{(l, q), t} = k^{'})}\right)\\
    + \sum_{l = 1}^{n}\sum_{q_1 = 1}^{r}\sum_{q_2 \neq q_1}\left((2^{2r} - 1)\frac{\mathbf{1}(H_{(l, q_1), t} = k, H_{(l, q_2), t)} = k)Y_{(l, q_1), t}Y_{(l, q_2), t}}{\mathbb{P}(H_{(l, q_1), t} = k, H_{(l, q_2), t)} = k)} + \right. \\
    + \left.(2^{2r} - 1)\frac{\mathbf{1}(H_{(l, q_1), t} = k^{'}, H_{(l, q_2), t)} = k^{'})Y_{(l, q_1), t}Y_{(l, q_2), t}}{\mathbb{P}(H_{(l, q_1), t} = k^{'}, H_{(l, q_2), t)} = k^{'})}\right)\\
    \left. + \sum_{l = 1}^{n}\sum_{q_1 = 1}^{r}\sum_{q_2 \neq q_1}\frac{\mathbf{1}(H_{(l, q_1), t} = k)Y_{(l, q_1), t}^2}{\mathbb{P}(H_{(l, q_1), t} = k)} + \sum_{l = 1}^{n}\sum_{q_1 = 1}^{r}\sum_{q_2 \neq q_1}\frac{\mathbf{1}(H_{(l, q_2), t} = k^{'})Y_{(l, q_2), t}^2}{\mathbb{P}(H_{(l, q_2), t} = k^{'})}\right] \label{est: mixedinter_var}
\end{multline}
and
\begin{multline}
    \widehat{\text{Cov}}(X_{n, t}, X_{n, t+1}) = \frac{1}{nrT}\sum_{l = 1}^n\sum_{q_1 = 1}^r\sum_{q_2 = 1}^r
    \left((2^r - 1)\frac{\mathbf{1}(H_{(l, q_1), t} = k, H_{(l, q_2), t+1} = k)Y_{(l, q_1), t}Y_{(l, q_2), t+1}}{\mathbb{P}(H_{(l, q_1), t} = k, H_{(l, q_2), t+1} = k)}\right.\\
    \left.+ (2^r - 1)\frac{\mathbf{1}(H_{(l, q_1), t} = k^{'}, H_{(l, q_2), t+1} = k^{'})Y_{(l, q_1), t}Y_{(l, q_2), t+1}}{\mathbb{P}(H_{(l, q_1), t} = k^{'}, H_{(l, q_2), t+1} = k^{'})} \right.\\
    + \left.\frac{\mathbf{1}(H_{(l, q_1), t} = k^{'})Y_{(l, q_1), t}^2}{\mathbb{P}(H_{(l, q_1), t} = k^{'})} + \frac{\mathbf{1}(H_{(l, q_2), t+1} = k)Y_{(l, q_2), t+1}^2}{\mathbb{P}(H_{(l, q_2), t+1} = k)} \right. \\
    \left. + \frac{\mathbf{1}(H_{(l, q_1), t} = k)Y_{(l, q_1), t}^2}{\mathbb{P}(H_{(l, q_1), t} = k)} + \frac{\mathbf{1}(H_{(l, q_2), t+1} = k^{'})Y_{(l, q_2), t+1}^2}{\mathbb{P}(H_{(l, q_2), t+1} = k^{'})}\right) \label{est: mixedinter_cov}
\end{multline}
\end{proposition}
The expression of the variance is immediate from the setup, \eqref{varxnt} and \eqref{covnt}. The estimator is obtained by replacing the non-identifiable terms with an upper bound and estimating the upper bound accordingly.

The difficulty in doing inference under a more general setting comes from the fact that it is hard to give the explicit expression of the variance. Since there is dependence across time, the variance of $\hat{\bar{\tau}}^{k, k^{'}}$ also involves covariance between Horvitz-Thompson estimators across different times. Hence, in this case, we need at least more assumptions on the assignment mechanism in order to express the variance term explicitly.

\section{Simulations}
\label{sec: simulations}

We now use a simulation to explore some of our theoretical results. 
Section~\ref{section:simul-clt} explores some of the finite sample properties of our central limit theorems in different realistic settings. 
Section~\ref{section:simul-stab} explores empirically some properties of the convex combination estimator proposed in Section~\ref{subsec: epsilonstab}: in particular, we show that confidence intervals based on normal approximations behave well in our simulation is a reasonable candidate for practical use.

\subsection{Simulations for central limit theorems}
\label{section:simul-clt}

We first explore the finite sample behavior of our central limit 
theorems. To make our simulations relevant, we consider a version 
of the popular stratified interference setting (\cite{duflo2003}; 
\cite{basse2018}), in which individuals 
are nested in groups of varying sizes, and interference may occur 
within but not across groups. Specifically, we consider the exposure 
mapping $f_{i,t}(w_{1:n,t}) = (w_{i,t}, u_{i,t})$, where at time $t$ $u_{i,t} = 1$ if unit $i$ 
has at least one treated neighbor and $u_{i,t} = 0$ otherwise, so each 
unit may receive one of four exposures: (0,0), (0,1), (1,0) and (1,1). 
Throughout, we consider a two-stage design whereby each group is assigned 
independently with probability $\frac{1}{2}$ to a high-exposure or 
low-exposure arm, and then each unit is assigned to treatment independently 
with probability $0.9$ in high-exposure groups, and $0.1$ in low-exposure groups.

We focus on the central limit theorems for ATEC. Theorem~\ref{thm: atecclt} establishes asymptotic results for ATEC under less constraining assumptions on the interference mechanism than for TEC (Theorem~\ref{thm: tec_clt_pure_spat_inter} in Appendix~\ref{sec: cross-sec}). To illustrate this point, we consider the stratified interference setting. We assume that the size of each group is bounded by $n^{1/3}$. In this case, $d_n = n^{1/3}$ and hence $T = \sqrt{n}$ suffices for Theorem \ref{thm: atecclt}. Compared to $d_n = o(n^{1/4})$ in the cross-sectional setting, we are able to have larger group size. We consider the exposure mapping $f_{i,t}(w_{1:n}) = (w_{i,t}, u_{i,t})$ where $u_{i,t} = 0$ if less than 25\% of the neighbors are treated; $u_{i,t} = 1$ if between 25\% and 50\% of the neighbors are treated; $u_{i,t} = 2$ if between 50\% and 75\% of the neighbors are treated and $u_{i,t} = 3$ if more than 75\% of the neighbors are treated. We generate the potential outcomes for unit $i$ at time step $t$ according to $\mathcal{N}(3w_{i,t} + 2u_{i,t} + 5 + \epsilon_{i,t}, 1)$, where $\epsilon_{i,t}$ is $\text{uniform}\{-1, 1\}$.
\begin{figure}
    \centering
    \begin{minipage}{0.45\textwidth}
        \centering
        \includegraphics[width=0.9\textwidth]{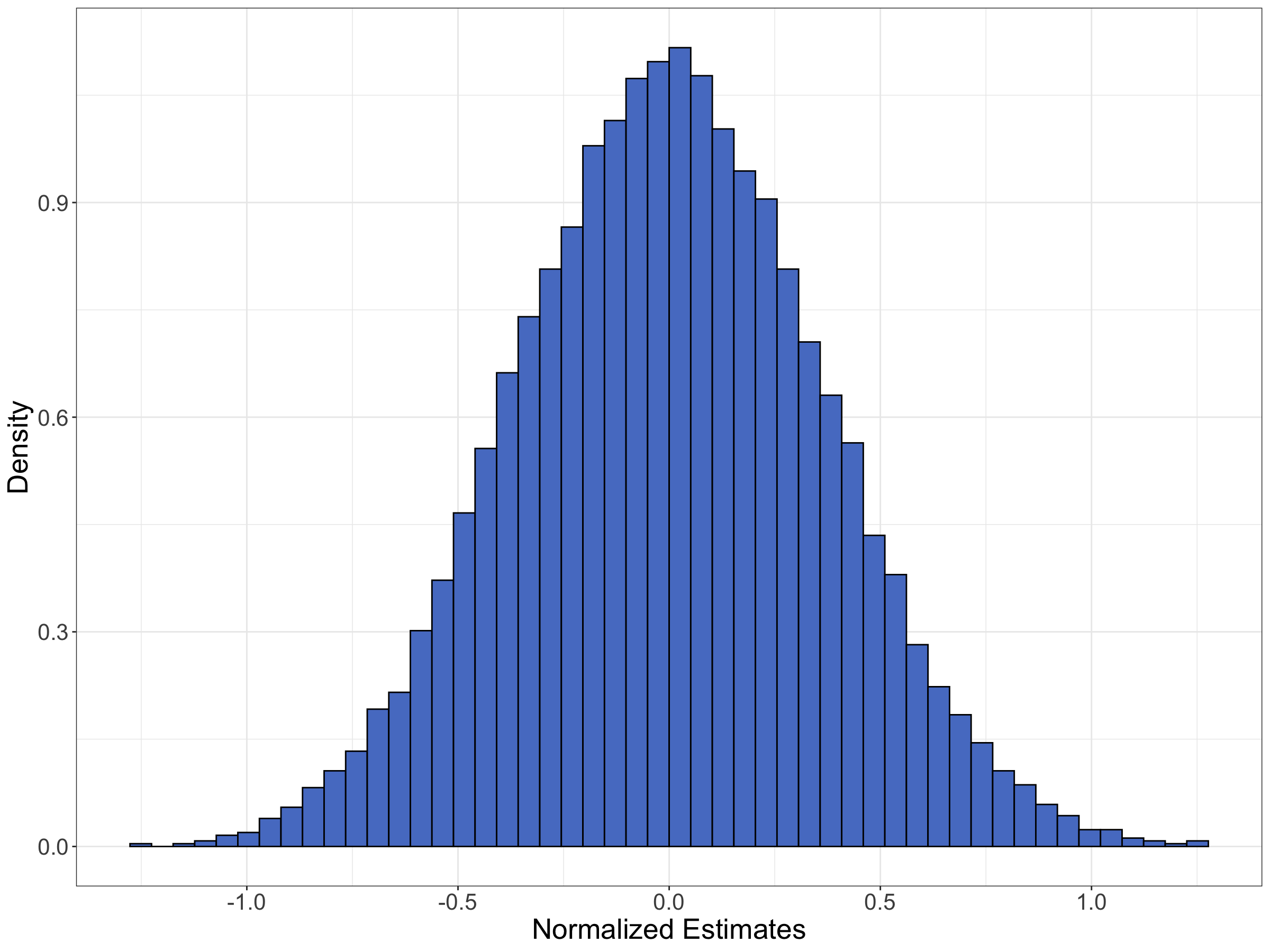} 
        \caption{Histogram, $n = 1000$}
        \label{fig: atec_hist}
    \end{minipage}\hfill
    \begin{minipage}{0.45\textwidth}
        \centering
        \includegraphics[width=0.9\textwidth]{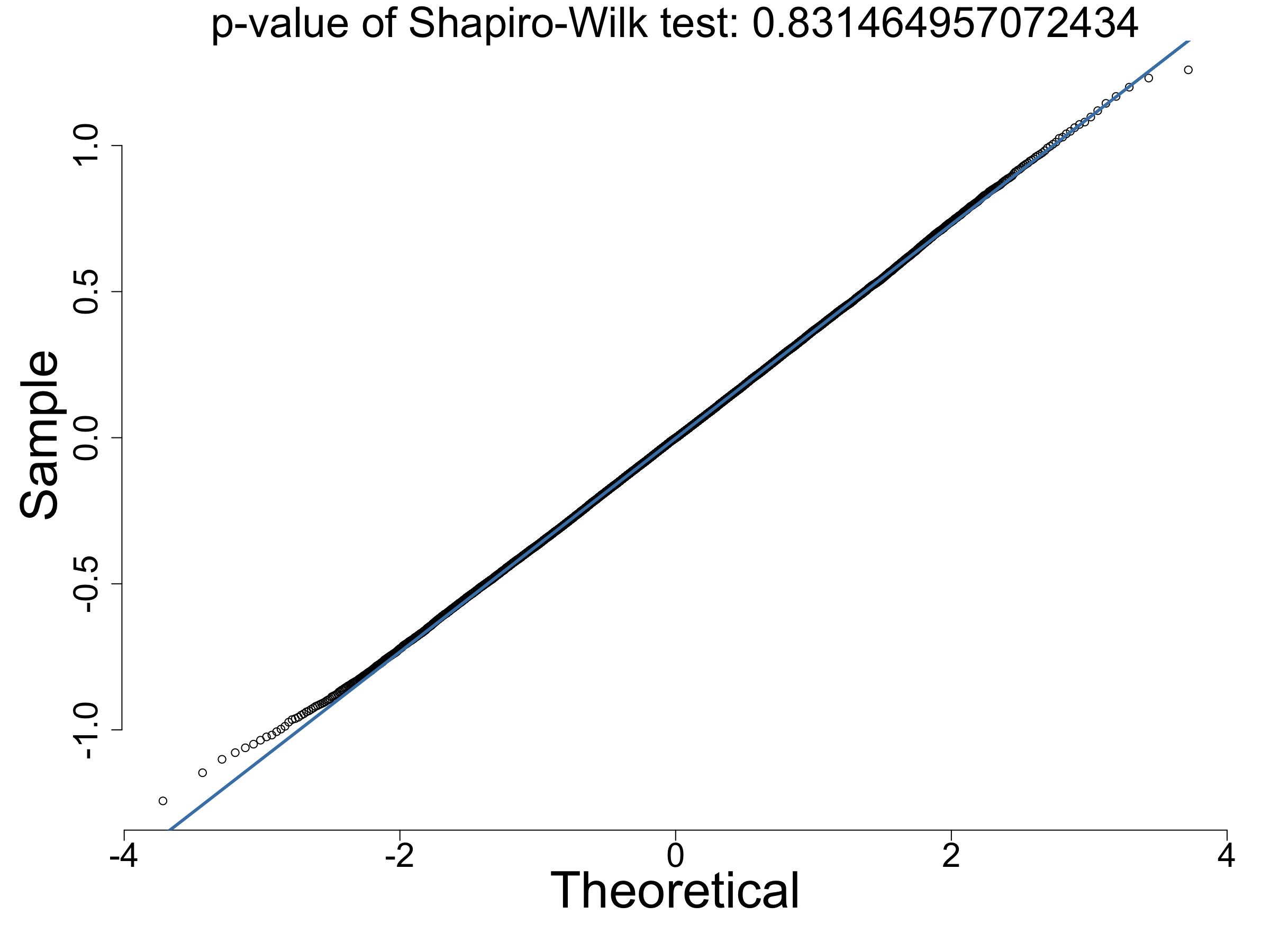} 
        \caption{Q-Q normal plot, $n = 1000$}
        \label{fig: atec_qq}
    \end{minipage}
\end{figure}
Figure \ref{fig: atec_hist} and \ref{fig: atec_qq} show that $n = 1000$ suffices for a good approximation. Moreover, the coverage of our 95\% confidence interval is 95.4\%. 

\subsection{Estimation under the stability assumption}
\label{section:simul-stab}
In Section \ref{subsec: epsilonstab}, we showed that with an appropriate choice of weights, the family of convex combination estimators outperforms the Horvitz-Thompson estimator. We illustrate this with a simulation study and show that our proposed confidence intervals perform well in our simulated setting.


\subsubsection{Estimation under stability assumption for total effects}
\label{subsubsec:simul_stab_est}

We consider a social network generated according to an Erd\H{o}s-R\'{e}nyi model, in which the units are assigned to treatment or control following a Bernoulli(1/2) design at each time step. We assume a local, pure population form of interference, summarized by the following exposure mappings:
\begin{equation}
f_{i}(w_{1:n, t}) = \left(w_{i, t}, \frac{1}{|\mathcal{N}_i|}\sum_{j \in \mathcal{N}_i}w_{j, t}\right) \label{exposure-model-frac}
\end{equation}
%
where $\mathcal{N}_i$ is the neighborhood of the $i$-th unit; that is, we assume that only direct neighbors affect one's potential outcomes. For each unit $i$, we generate 
the potential outcomes at $t = 1$ randomly from $\mathcal{N}(10, 1)$. Then, for each time $t > 1$, we generate the potential outcome $Y_{i, t}(k)$ uniformly from the interval $(Y_{i,t-1}(k) - \epsilon, Y_{i, t-1}(k) + \epsilon)$, so $\epsilon$-stability holds. 
Throughout our simulations, we assume that $T = 20$ and we are interested in the total effect at time step $t = 20$. We compare the performance of the standard Horvitz-Thompson 
estimator and the performance of the convex combination estimator for estimating the 
total effect $\tau_T^{TE}$ at time $t=T=20$, varying both the population size $n$ and 
the number of time steps $k$ used in the convex combination. We estimate $\epsilon$ using  Algorithm~\ref{alg: epsilon} described in Section \ref{subsec: epsilonstab}; we use 
Proposition~\ref{prop: msereductionwithindtreat} to estimate $\alpha$ when $k = 2$, and 
solve the optimization problem introduced in Appendix \ref{suppl_sec_b} for $k \geq 3$.
\begin{table} 
\centering
\scalebox{0.8} {
\begin{tabular}{c l c c c c c c c c} 
\toprule 
\multicolumn{2}{c}{Sample Size} & n = 50 & n = 100 & n = 250 & n = 500 & n = 750 & n = 1000\\
\midrule
    \multirow{1}{*}{RMSE for $\hat{\tau}_{20}^{TE}$} 
            & & 64.68 & 28.98 & 5.80 & 1.84 & 0.95 & 0.70\\
            \midrule
    \multirow{1}{*}{RMSE for $\hat{\tau}^c_{20}$, $k$ = 2} 
            & & 14.17 & 9.18 & 3.72 & 1.42 & 0.68 & 0.52\\
            \midrule
    \multirow{1}{*}{RMSE for $\hat{\tau}^c_{20}$, $k$ = 5} 
            & & 4.39 & 4.58 & 3.01 & 1.17 & 0.58 & 0.45\\
\bottomrule 
\end{tabular}}
\caption{Root mean squared errors (RMSE) for $\hat{\tau}_{20}^{TE}$, $\hat{\tau}^c_{20}$ with $k = 2$ and $\hat{\tau}^c_{20}$ with $k = 5$} 
\label{table:cvxest_diffn} 
\end{table}

We first fix $\epsilon$ to be 3 and vary the sample size. To make each unit have the same expected number of neighbors, we scale the probability $p$ in Erd\H{o}s-R\'{e}nyi model accordingly. For each $n$, we fix the graph and generate 100 realizations of assignments. Table \ref{table:cvxest_diffn} shows the root mean squared errors for three kinds of estimators for the total effect: the usual Horvitz-Thompson estimator, the convex combination type estimator with $k = 2$, and the convex combination estimator with $k = 5$. We see that the convex combination type estimators effectively reduce the mean squared error. Moreover, when $n$ is relatively small, the reduction in mean squared error is significant.

\subsubsection{Coverage of two approximate confidence intervals}
Recall that in Section~\ref{subsec: epsilonstab}, we gave two approximate confidence intervals of $\tau_t^{TE}$ based on our convex combination estimator $\hat{\tau}_t^c$ and variance estimator. We now provide coverage results of these two approximate confidence intervals. We assume a social network generated from the Erd\H{o}s-R\'{e}nyi Model with $n = 100$ and $p = 0.05$. We fix the stability parameter $\epsilon$ to be 3 and generate the data in the same way as in the previous section. To calculate the coverage, we generate 1000 realizations of the assignments and construct approximate confidence intervals accordingly.
\begin{table} 
\centering
\scalebox{0.75} {
\begin{tabular}{c l c c c c c} 
\toprule 
\multicolumn{2}{c}{Confidence Interval} & Network 1 & Network 2 & Network 3\\
\midrule
    \multirow{1}{*}{Gaussian CI with variance estimated by $\widehat{\text{Var}}^d$} 
            & & 92.9\% & 98.4\% & 95.9\%\\
            \midrule
    \multirow{1}{*}{Gaussian CI with variance estimated by $\widehat{\text{Var}}^u$} 
            & & 97.2\% & 99.8\% & 100\%\\
            \midrule
    \multirow{1}{*}{Chebyshev CI with variance estimated by $\widehat{\text{Var}}^d$} 
            & & 91.4\% & 94.1\% & 96.4\%\\
            \midrule
    \multirow{1}{*}{Chebyshev CI with variance estimated by $\widehat{\text{Var}}^u$} 
            & & 94.6\% & 95.6\% & 97.7\%\\
\bottomrule 
\end{tabular}}
\caption{Coverage of two approximate confidence intervals for $\tau_t^{TE}$ with $k = 2$} 
\label{table:cvx_est_ci} 
\end{table}

Table~\ref{table:cvx_est_ci} shows the two approximate confidence intervals provide reasonable coverage across the three different social networks. Although the Gaussian confidence interval ignores the bias of $\hat \tau_{t}^{c}$, it tends to provide better coverage than the confidence intervals obtained from the Chebyshev approach. Moreover, the Gaussian intervals tend to be shorter, making them practically more useful. Appendix \ref{appC} provides an additional table showing the average lengths of the confidence intervals in Table~\ref{table:cvx_est_ci}.

\section{Two real data examples}
\label{sec: real-data-ex}
We now apply our methods to two empirical applications. In the first application, we use the convex combination estimator to analyze a panel experiment and show it reduces the variance and leads to more reliable estimates of the temporal exposure contrast. In the second application, we run a semi-synthetic experiment on a social network to demonstrate the necessity of our assumptions for the validity of the analysis and provide further empirical evidence of the advantage of the convex combination estimator.
\subsection{Rational cooperation}

The panel experiment we analyze is from \cite{RePEc:eee:jetheo:v:127:y:2006:i:1:p:117-154}. The authors test a game-theoretic model of ``rational cooperation" through a panel experiment. Specifically, in each experiment session, they recruited 22 subjects to play 20 twice-played prisoners' dilemmas, ensuring that no player would meet the same partner twice. The twice-played prisoners' dilemma consists of two periods with different pay-off structures, as shown in Table~\ref{tab: andreoni-samuelson stage game}. The parameters $x_1, x_2$ satisfy $x_1, x_2 \geq 0$, $x_1 + x_2 = 10$.

\begin{table}[!htbp]
    \centering
    \begin{tabular}{c c c}
         \arrvline{} & $C$ & $D$ \\
         \hline \hline
        $C$  \arrvline{} & $(3x_1, 3x_1)$ & $(0, 4x_1)$  \\
        $D$  \arrvline{} & $(4x_1, 0)$  & $(x_1, x_1)$ \\ \\
        \multicolumn{3}{c}{Period one}
        \end{tabular}\quad
        \begin{tabular}{c c c}
        \arrvline{} & $C$ & $D$ \\
        \hline \hline
        $C$  \arrvline{} & $(3x_2, 3x_2)$  & $(0, 4x_2)$  \\
        $D$  \arrvline{} & $(4x_2, 0)$  & $(x_2, x_2)$ \\ \\
        \multicolumn{3}{c}{Period two}
        \end{tabular}  
    \caption{Payoff structure in the experiment conducted by \cite{RePEc:eee:jetheo:v:127:y:2006:i:1:p:117-154}. The choice $C$ denotes ``cooperate'' and the choice $D$ ``defect.''}
    \label{tab: andreoni-samuelson stage game}
\end{table}

Let $\lambda = \frac{x_1}{x_1 + x_2}$, then for each round of the experiment, 22 subjects were grouped into 11 pairs, and each pair was randomly assigned with a $\lambda \sim \text{Unif}\{0, 0.1, \cdots, 0.9, 1\}$. The outcomes were the total payoffs. Since there are five sessions in total, we have 110 subjects and 2200 outcomes. We use this panel experiment to illustrate that the convex combination estimator effectively reduces the estimates' variance and thus produces more reliable estimates. To this end, following \cite{bojinov2020panel}, we define treatment to be $\lambda > 0.6$ and control to be $\lambda \leq 0.6$. This results in a panel experiment with binary treatments and Bernoulli design with treated probability $\frac{5}{11}$. Under this setup, we generally expect a positive treatment affect as the payoffs are more concentrated in period two. 

We next build a social network among all subjects in the experiment. If the players have played each other in the first few rounds, then they should have some influence on each other for the later rounds. Hence, we consider any players that played each other in the first five rounds of the game as being connected. We then use the remaining 15 rounds as our experimental data. So, for our panel experiment, we have $n = 110$ and $T = 15$. As \cite{bojinov2020panel} showed little evidence of carryover effects, we assume there is only population interference. Then, we use the exposure model in \eqref{exposure-model-frac} and the temporal exposure contrast we are interested in is the exposure contrast between $(0, \leq 0.2)$ and $(1, \geq 0.8)$ for each time step. We now report the Horvitz-Thompson estimates of the temporal exposure contrast for the last 10 time steps, the estimates from the 2-step, and the 5-steps convex combination estimator estimates. Table~\ref{table:results_andreoni} shows the results.

\begin{table} 
\centering
\scalebox{0.7} {
\begin{tabular}{c l c c c c c c c c c c c} 
\toprule 
\multicolumn{2}{c}{Time Step} & $T = 6$ & $T = 7$ & $T = 8$ & $T = 9$ & $T = 10$ & $T = 11$ & $T = 12$ & $T = 13$ & $T = 14$ & $T = 15$ & Variance\\
\midrule
    \multirow{1}{*}{Horvitz-Thompson} 
            & & -8.08 &  11.03 &  29.39 &  -3.53 &  57.48 &  -3.76 &  10.29 &  23.40 & -13.70 &  16.19 & 452.80\\
\midrule
    \multirow{1}{*}{2-step} 
            & & -6.84 &   9.48 &  25.40 &  -2.99 &  33.81 &  -2.98 &   9.83 &  22.10 & -11.20 &  14.64 & 226.13\\
\midrule
    \multirow{1}{*}{5-steps} 
            & & -6.83 &   9.46 &  25.38 &  -2.98 &  33.74 &  -2.97 &   9.82 & 22.09 & -11.17 &  14.63 & 225.37\\
\bottomrule 
\end{tabular}}
\caption{Estimates for temporal exposure contrasts from the panel experiment in \cite{RePEc:eee:jetheo:v:127:y:2006:i:1:p:117-154}} 
\label{table:results_andreoni} 
\end{table}

In general, we do not expect the temporal exposure contrasts to be different for different time steps since all the 15 rounds of games were done together in one session. And as we can see from the table, the convex combination estimator leads to estimates with much smaller variance. Note that the estimates from 2-step and 5-steps convex combination estimators are similar, illustrating that the choice of $k$ is not crucial since the estimator itself takes care of it. Moreover, as we pointed out earlier, we would expect positive exposure contrast and the estimates from convex combination estimator are more reliable in the sense that it shrinks the estimates towards zero when the Horvitz-Thompson estimator gives a negative value (this is possible since we only have $n = 110$ subjects which is a small number).

\subsection{Facebook network semi-synthetic experiment}
We now describe a semi-synthetic experiment using the Swarthmore College social network from the Facebook 100 dataset \citep{TRAUD20124165}. All networks in this dataset are complete online friendship networks for one hundred colleges and universities collected from a single-day snapshot of Facebook in September 2005. The network we use is of size 1657 with 61049 edges. We use this network as the graph that describes population interference among units and generate an assignment vector using a Bernoulli design with a success probability of 1/2. We first show mean squared error reduction of using convex combination estimator to estimate temporal exposure contrast between $(0, 0)$ and $(1, 1)$ at $T = 20$. Let $\rho_{i, t} = \frac{1}{|\mathcal{N}_i|}\sum_{j \in \mathcal{N}_i}w_{j, t}$, we assume the following exposure mappings: 
\begin{equation*}
    f_i(w_{1:n, t}) = (w_{i, t}, \tilde{\rho}_{i, t}), \quad \text{where } \tilde{\rho}_{i, t} = \begin{cases}
    0 & \text{ if } \rho_{i, t} \leq 0.3, \\
    \rho_{i, t} & \text{ if } 0.3 \leq \rho_{i, t} \leq 0.7, \\
    1 & \text { if } \rho_{i, t} > 0.7.
    \end{cases}
\end{equation*}
Now we make a panel experiment with $T = 20$. We generate the outcomes at each time step according to a linear model that is linear in $w_{i, t}$ and $\tilde{\rho}_{i, t}$ and add a time-varying component $\epsilon_t$ that is uniformly distributed on $[-0.5, 0.5]$. Table~\ref{table:cvxest_semi-synthetic} shows the empirical bias, the variance, and the root mean squared errors (RMSE) of the estimates of the temporal exposure contrast at time step $T = 20$ by using Horvitz-Thompson estimator, 2-step convex combination estimator, and 5-step convex combination estimator. As expected, the convex combination estimator reduces the RMSE significantly. Though the biases seem large compared to the Horvitz-Thompson estimator, as we mentioned previously, we can also control the amount of bias we tolerate, which implicitly accounts for the time effect.

\begin{table} 
\centering
\scalebox{1.0} {
\begin{tabular}{c c c c} 
\toprule 
\multicolumn{1}{c}{Estimator of TEC} & Horvitz-Thompson & 2-step CVX & 5-step CVX\\
\midrule
    \multirow{1}{*}{RMSE} 
            & 50.48 & 3.57 & 3.31\\
    \midrule
    \multirow{1}{*}{Empirical bias} 
            & 0.26 & 2.75 & 2.80\\
    \midrule
    \multirow{1}{*}{Empirical variance} 
            & 2548.495 & 5.19 & 3.09\\
\bottomrule 
\end{tabular}}
\caption{RMSE for different estimators of temporal exposure contrast at $T = 20$} 
\label{table:cvxest_semi-synthetic} 
\end{table}

The maximal degree for the network is 577, which is far greater than the $\sqrt{n}$. To make Theorem~\ref{thm: atecclt} hold approximately for this network, we require having an extremely large $T$. Below we illustrate this empirically through a semi-synthetic experiment. Let 
\begin{equation*}
    f_i(w_{1:n, t}) = (w_{i, t}, \tilde{\rho}_{i, t}), \quad \text{where } \tilde{\rho}_{i, t} = \begin{cases}
    0 & \text{ if } \rho_{i, t} \leq 0.35, \\
    1 & \text{ if } 0.35 < \rho_{i, t} \leq 0.5, \\
    2 & \text{ if } 0.5 < \rho_{i, t} \leq 0.65, \\
    3 & \text { if } \rho_{i, t} > 0.65.
    \end{cases}
\end{equation*}
We are interested in the average temporal exposure contrast between $(1, 3)$ and $(0, 0)$. Since the network is dense, with an average degree of 73.69, we expect the Horvitz-Thompson estimator to be inaccurate since units with exposure values $(1, 3)$ or $(0, 0)$ will unlikely to be those units with many neighbors. Figure~\ref{fig: semi-synthetic1} and \ref{fig: semi-synthetic2} show the histograms of Horvitz-Thompson estimates for $T = 20$ and $T = 100$ respectively. Here, we calculate ATEC for 10,000 realizations, and since the computation of the variance estimate is time-consuming, we do not rescale the estimates.

\begin{figure}[t]
    \centering
    \begin{minipage}{0.5\textwidth}
        \centering
        \includegraphics[width=\textwidth]{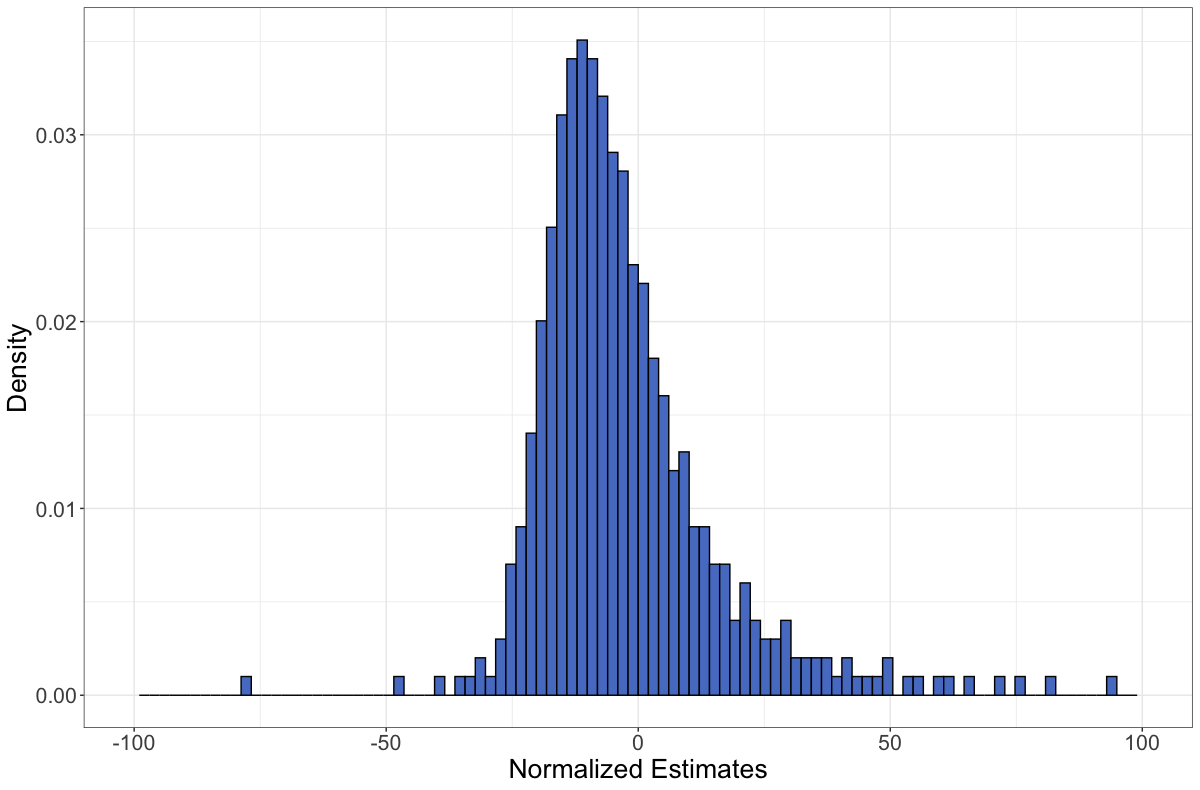} 
        \caption{Histogram, $T = 20$}
        \label{fig: semi-synthetic1}
    \end{minipage}\hfill
    \begin{minipage}{0.5\textwidth}
        \centering
        \includegraphics[width=\textwidth]{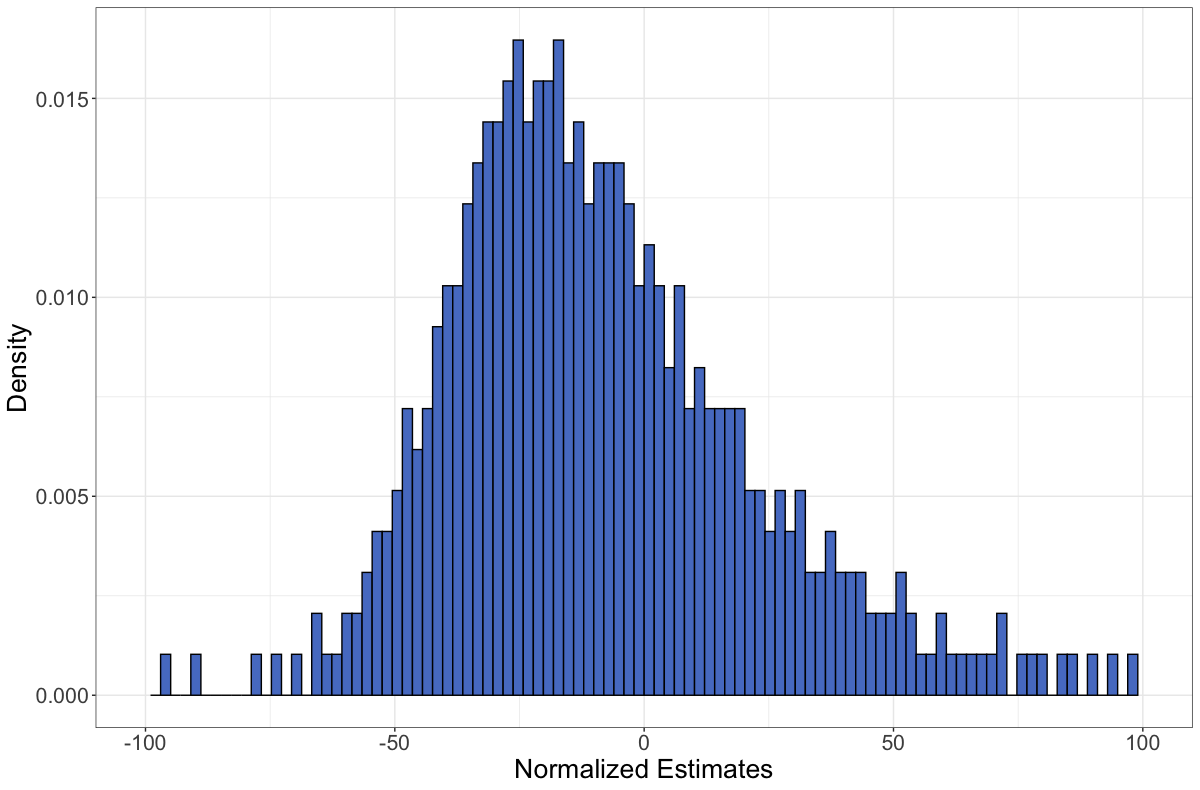} 
        \caption{Histogram, $T = 100$}
        \label{fig: semi-synthetic2}
    \end{minipage}
\end{figure}

Figure \ref{fig: semi-synthetic1} shows that when $T = 20$ the histogram is far from normally distributed. Figure \ref{fig: semi-synthetic2} shows that when $T = 100$, although the data are much closer to being normally distributed, they still are not. Also, note that the centers of these two histograms are away from 0; as we stated above, since some of the neighborhoods are extremely large, we cannot observe the exposure value we would need for units with large neighborhoods. This illustrates the necessity of condition \eqref{atecclt: degcond} --- reliable inference requires more experiments if we have a dense network. We also report the coverage using a Gaussian confidence interval here. For both $T = 20$ and $T = 100$, the empirical coverage of naive Gaussian confidence interval is around 80\%.


\section{Conclusion}
\label{sec: conclusion}


In this paper, we developed estimation and inference results for panel experiments with population interference. Our work shows that the added temporal dimension in experiments with population interference may either help or hurt our ability to do inference. In the absence of carryover effects, causal effects averaged across time can be estimated more easily than at single time points. Here we also introduced a variance reduction technique that incorporated a novel stability assumption and showed that if the potential outcomes are well-behaved, a simple convex combination type estimator can effectively reduce the mean squared error. In the presence of carryover effects, we can still obtain novel central limit theorems; however, these require additional assumptions and are not as flexible. Finally, our results also explicitly capture the relationship between the growth rate of the dependency graph of the exposure values, allowing us to consider the trade-off between the assignment mechanism and the interference.  

Many interesting avenues of investigation around interference in panel experiments have been left unexplored in this manuscript and will be the object of future work. First, our results mostly consider the temporally independent assignment mechanism: this is, of course, limiting as it removes all dynamic designs, but it does present a useful benchmark. Second, while our simulations show that our convex combination estimators seem to behave well, our formal results under this new stability assumption are still limited. Third, we do not explicit discussions hypothesis testing. Although our results do provide a way to test specific hypotheses by inverting the confidence intervals, there has been literature that discusses how to conduct a Fisher randomization test for the sharp null hypothesis in panel experiments \citep{bojinov2020panel}. 



\bibliographystyle{agsm}

\bibliography{submission}

\newpage

\appendix
\numberwithin{equation}{section}





\section{Standard population interference}
\label{sec: cross-sec}

This appendix focuses on estimating TEC under population interference and assumes that either the experiment was conducted over a single time point or that there are no carryover effects. In both cases, we drop the subscript $t$ for the remainder of the section. Our setup is now equivalent to the one studied in \cite{liuhudgens2014}, \cite{aronow2017}, \cite{chin2018central} and \cite{leung2022}. Our Horvitz-Thompson type estimator $\hat{\tau}^{k, k^{'}}$ now simplifies to,
\begin{align}
    \hat{\tau}^{k, k^{'}} &= \frac{1}{n}\sum_{i = 1}^{n}\left[\frac{\mathbf{1}(H_i = k)}{\pi_i(k)}Y_i(k) - \frac{\mathbf{1}(H_i = k^{'})}{\pi_i(k^{'})}Y_i(k^{'})\right] \label{est: HT_at_t_beta},
\end{align}
where $\pi_i(k) = \mathbb{P}(H_i = k)$ and $\pi_i(k^{'}) = \mathbb{P}(H_i = k^{'})$.

\cite{aronow2017} showed that if the potential outcomes and inverse exposure probabilities are bounded, and the number of dependent pairs of $H_i$'s is of order $o(n^2)$, then the estimator $\hat{\tau}^{k, k^{'}}$ is consistent,
\begin{equation*}
    \left(\hat{\tau}^{k, k^{'}} - \tau^{k, k^{'}}\right) \rightarrow_\mathbb{P} 0.
\end{equation*}
In addition, the authors provided an asymptotically conservative confidence interval of $\hat{\tau}^{k, k^{'}}$ and implicitly outlined a version of a central limit theorem in the proof. However, the conditions stated in their derivations were sufficient but not necessary. Below, we establish a central limit theorem for $\hat{\tau}^{k, k^{'}}$ under weaker conditions and provide a detailed proof that builds on recent results by \cite{chin2018central}. We then illustrate the trade-offs between the strength of the interference structure assumption and the assignment mechanism's flexibility.

\subsection{A central limit theorem}

We can now state the following central limit theorem for temporal exposure contrast.
\begin{theorem}
\label{thm: tec_clt_pure_spat_inter}
Under Assumptions \ref{assump: bddpo}-\ref{assump: nondegeneratevar} and the condition that $d_n = o(n^{1/4})$, we have
\begin{equation*}
    \frac{\sqrt{n}(\hat{\tau}^{k, k^{'}} - \tau^{k, k^{'}})}{\text{Var}(\sqrt{n}\hat{\tau}^{k, k^{'}})^{1/2}} \xrightarrow{d} \mathcal{N}(0, 1)
\end{equation*}
as $n \rightarrow \infty$.
\end{theorem}

The assumption that $d_n = o(n^{1/4})$ quantifies the dependence among observations due to interference. Let $d_n$ be the maximal degree in this graph, which is equal to the maximal number of dependent exposure values for each unit; this is different from the definition of $d_n$ in the main paper, where it is defined as a limit or a maximal over $t$.

Theorem \ref{thm: tec_clt_pure_spat_inter} strengthens the result of \cite{aronow2017} in two ways. First, our Assumption \ref{assump: nondegeneratevar} weakens Condition 6 of \cite{aronow2017}, which requires the convergence of $\text{Var}(\sqrt{n}\hat{\tau}^{k, k^{'}}_{t})$. Second, we allow for a higher range of dependence ($d_n = o(n^{1/4})$ compared to $d_n = O(1)$ as in \cite{aronow2017}) among exposure values. The proof of this theorem relies on recent results in \cite{chin2018central}. 

\subsection{Design and interference structure: a trade-off}\label{subsec: design trade-off}

Intuitively, Theorem~\ref{thm: tec_clt_pure_spat_inter} asserts that asymptotic normality holds so long as the dependency relations among the $H_i$'s are moderate. However, since $H_i = f_i(W_{1:n})$ is determined by both function $f_i$ and assignment $W$, the dependence structure among the $H_i$'s --- and therefore the value of $d_n$ --- depends on both the exposure specification and the assignment mechanism.

This suggests that there exists a trade-off between the strength of the dependence in the $W_i$'s induced by the assignment mechanism and the dependence induced by the interference structure. The less restricted the interference structure is, the more restricted the assignment mechanism must be; in reverse, the more restricted the interference structure, the more flexible one can be with the design. We illustrate these insights with three special cases of Theorem~\ref{thm: tec_clt_pure_spat_inter}, applied to popular settings. We should also note that our condition on $d_n$ is not a sufficient condition for the central limit theorem. For example, if we consider $f_i(W_i) = W_i$ (i.e., there is no interference) and $W$ follows completely randomized design, then the central limit theorem still holds (see Theorem 1 in \cite{ding2017}). The discussion here mainly illustrates the entanglement between the assignment mechanism and the interference structure from a general perspective.

\begin{example}
\label{cor: cor1}
Suppose that the interference structure among $n$ units is adequately described by a social network $\mathcal{A}_n$, and 
assume that the exposure mapping is of the form $f_{i}(W_{1:n}) = f_i(W_{\mathcal{N}_i})$; that is, only the neighbors' assignments  matter.
Let $\delta_n$ be the maximal number of neighbors a unit can have in the network 
$\mathcal{A}_n$ --- which is distinct from the dependency graph. Then if 
$\delta_n = o(n^{1/8})$ and the $W_i$'s are independent (i.e., the design is Bernoulli),
then $d_n=o(n^{1/4})$ as required by Theorem~\ref{thm: tec_clt_pure_spat_inter}.
\end{example}
This first example explores one extreme end of the trade-off, in which the 
assignment mechanism is maximally restricted --- the $W_i$'s are independent --- 
which allows for a comparatively large amount of interference. 

\begin{example}
\label{cor: cor3}
We consider the graph cluster randomization approach (\cite{ugander2013}) in which case we group units into clusters and randomize at the cluster level. Following the notations in \cite{ugander2013}, we let the vertices be partitioned into $n_c$ clusters $C_1, \cdots, C_{n_c}$. The graph cluster randomization approach assigns either treatment or control to all the units in each cluster. Suppose one's potential outcomes depend only on the assignments of its neighbors. Let $\delta_n$ be the maximal number of neighbors one can have and $c_n$ be the maximal size of the cluster. Then $d_n = o(n^{1/4})$ for $\delta_n^2+ \delta_nc_n = o(n^{1/4})$.
\end{example}
\begin{example}
\label{cor: cor2}
Another commonly studied scenario is the ``household" interference (\cite{basse2018}; \cite{duflo2003}). In household interference, we assume that each unit belongs to a ``household" and their potential outcomes depend only on the assignments of the units within the ``household". Suppose we have a two-stage design such that we first assign each household into treatment group or control group independently and then we assign treatments to units in each household depending on the assignment of their associated household. Let $r_n$ be the maximal size of the ``household", then $d_n = o(n^{1/4})$ for $r_n = o(n^{1/4})$.
\end{example}

Table \ref{table:trade-off} summarizes the above three examples. In Example \ref{cor: cor1}, to have a general network interference setting with the maximum possible number of neighbors for each unit, we constrain the design to be the Bernoulli design. Further limiting the interference, like in Example \ref{cor: cor2} where the interference is restricted within households, we can have a more complex two-stage design. In the same spirit, Example~\ref{cor: cor3} shows that for a highly dependent design, we need an even stronger condition on the interference structure, indicated by a stronger rate condition on $\delta_n$. In general, a weaker assumption on the interference structure induces a more complex dependence graph for the exposures, which in turn reduces our flexibility in the choice of design. 
\begin{table}[H]
\centering
\scalebox{0.8} {
\begin{tabular}{c l c c c c c c c} 
\toprule 
\multicolumn{2}{c}{Interference} & Design & Conditions\\
\midrule
    \multirow{1}{*}{Network Interference} 
            & & Bernoulli Design & $\delta_n = o(n^{1/8})$\\
    \midrule
    \multirow{1}{*}{Network Interference} 
            & & Graph Cluster Randomization & $\delta_n^2 + \delta_nc_n = o(n^{1/4})$\\
    \midrule
    \multirow{1}{*}{Group Interference} 
            & & Two-stage Design & $r_n = o(n^{1/4})$\\
\bottomrule 
\end{tabular}}
\caption{Trade-off between design and interference} 
\label{table:trade-off} 
\end{table}

\subsection{Inference}
\label{subsec: tec_inference}

The central limit theorem stated in Theorem~\ref{thm: tec_clt_pure_spat_inter} serves as our basis for inference. 

\begin{proposition}
\label{prop: ci}
Assuming all the assumptions in Theorem \ref{thm: tec_clt_pure_spat_inter}, then for any $\delta > 0$, we have that
\begin{equation*}
    \mathbb{P}\left(\frac{\widehat{\text{Var}}(\hat{\tau}^{k, k^{'}})}{\text{Var}(\hat{\tau}^{k, k^{'}})} \geq 1 - \delta\right) \rightarrow 1.
\end{equation*}
where $\widehat{\text{Var}}(\hat{\tau}^{k, k^{'}}) = n^{-1} \widehat{\text{Var}}(\sqrt{n}\hat{\tau}^{k, k^{'}})$.
Therefore, we can construct asymptotically conservative confidence interval based on the variance estimator:
for any $\delta > 0$,
\begin{equation*}
\begin{split}
    \mathbb{P}\left(\tau^{k, k^{'}} \in \left[\hat{\tau}^{k, k^{'}} - \frac{z_{1 - \frac{\alpha}{2}}}{\sqrt{1 - \delta}}\sqrt{\widehat{\text{Var}}(\hat{\tau}^{k, k^{'}})}, \right.\right.\\
    \left.\left.\hat{\tau}^{k, k^{'}} + \frac{z_{1 - \frac{\alpha}{2}}}{\sqrt{1 - \delta}}\sqrt{\widehat{\text{Var}}(\hat{\tau}^{k, k^{'}})}\right]\right) \geq 1 - \alpha
    \end{split}
\end{equation*}
for large $n$.
\end{proposition}
$\widehat{\text{Var}}(\sqrt{n}\hat{\tau}^{k, k^{'}})$ is the same as the one given in Proposition~\ref{prop: ci_atec}.
Once again, this result strengthens that of \cite{aronow2017} by both removing the 
requirement that $n\text{Var}(\hat{\tau}^{k, k^{'}})$ converge, and by relaxing the 
constraint on the interference mechanism. Note that here $\delta > 0$ is arbitrary and we present detailed simulations in Section~\ref{sec: simulations} with $\delta = 0.04$.

\newpage
\section{Proofs and additional discussions} \label{appA}

To begin with, we provide technical tools that we will use in our proofs. We first state a lemma from \cite{ross2011}:
\begin{lemma}
\label{lemma: steinbound}
Let $X_1, \cdots, X_n$ be a collection of random variables such that $\mathbb{E}\left[X_i^4\right] < \infty$ and $\mathbb{E}\left[X_i\right] = 0$. Let $\sigma^2 = $\textup{Var}$(\sum_iX_i)$ and $S = \sum_iX_i$. Let $d$ be the maximal degree of the dependency graph of $(X_1, \cdots, X_n)$. Then for constants $C_1$ and $C_2$ which do not depend on $n, d$ or $\sigma^2$,
\begin{equation}
    d_\mathcal{W}(S/\sigma) \leq C_1\frac{d^{3/2}}{\sigma^2}\left(\sum_{i = 1}^{n}\mathbb{E}\left[X_i^4\right]\right)^{1/2} + C_2\frac{d^2}{\sigma^3}\sum_{i = 1}^{n}\mathbb{E}|X_i|^3, \label{bound}
\end{equation}
where $d_\mathcal{W}(S/\sigma)$ is the Wasserstein distance between $S/\sigma$ and standard Gaussian.
\end{lemma}
Second, we provide the expression for the variance of $\hat{\tau}^{k, k^{'}}$:
\begin{lemma}[Variance of Horvitz-Thompson estimator]
We have that (\cite{aronow2017}):
\begin{equation*}
\begin{split}
    \text{Var}(\sqrt{n}\hat{\tau}^{k, k^{'}}) = \frac{1}{n}\sum_{i = 1}^{n}\pi_i(k)(1 - \pi_i(k))\left(\frac{Y_i(k)}{\pi_i(k)}\right)^2 \\
    + \frac{1}{n}\sum_{i = 1}^{n}\pi_i(k^{'})(1 - \pi_i(k^{'}))\left(\frac{Y_i(k^{'})}{\pi_i(k^{'})}\right)^2 + \frac{2}{n}\sum_{i = 1}^{n}Y_i(k)Y_i(k^{'}) \\
    + \frac{1}{n}\sum_{i = 1}^{n}\sum_{j \neq i}\left\{\left[\pi_{ij}(k) - \pi_i(k)\pi_j(k)\right]\frac{Y_i(k)}{\pi_i(k)}\frac{Y_j(k)}{\pi_j(k)} \right.\\
    + \left.\left[\pi_{ij}(k^{'}) - \pi_i(k^{'})\pi_j(k^{'})\right]\frac{Y_i(k^{'})}{\pi_i(k^{'})}\frac{Y_j(k^{'})}{\pi_j(k^{'})} \right\}\\
    - \frac{2}{n}\sum_{i = 1}^{n}\sum_{j \neq i}\left\{\left[\pi_{ij}(k, k^{'}) - \pi_i(k)\pi_j(k^{'})\right]\frac{Y_i(k)}{\pi_i(k)}\frac{Y_j(k^{'})}{\pi_j(k^{'})}\right\}
\end{split}
\end{equation*} \label{htvar}
\end{lemma}
Here $\pi_{ij}(k) = \mathbb{P}(H_i = k \text{ and } H_j = k)$
\begin{proof} [Proof of Theorem \ref{thm: tec_clt_pure_spat_inter}]
Note that $\hat{\tau}^{k, k^{'}} = \sum_{i = 1}^{n}\tilde{\tau}_i$ where
\begin{equation*}
    \tilde{\tau}_i = \frac{1}{n}\left[\frac{\mathbf{1}(H_i = k)}{\pi_i(k)}Y_i(k) - \frac{\mathbf{1}(H_i = k^{'})}{\pi_i(k^{'})}Y_i(k^{'})\right]
\end{equation*}
and $\mathbb{E}\left[\tilde{\tau}_i\right] = \frac{1}{n}\left[Y_i(k) - Y_i(k^{'})\right]$, hence if we let $X_i = \sqrt{n}(\tilde{\tau}_i - \mathbb{E}\left[\tilde{\tau}_i\right])$, then $\sqrt{n}(\hat{\tau}^{k, k^{'}} - \tau^{k, k^{'}}) = \sum_{i = 1}^{n}X_i = S$. By Assumption \ref{assump: bddpo} and Assumption \ref{assump: overlap}, we know that $X_i = O_p(n^{-1/2})$, hence there exist some constants $C_1$ and $C_2$ such that for sufficiently large $n$, both
\begin{equation*}
    \left(\sum_{i = 1}^n\mathbb{E}\left[X_i^4\right]\right)^{1/2} \leq C_1n^{-1/2}
\end{equation*}
and
\begin{equation*}
    \sum_{i = 1}^{n}\mathbb{E}|X_i|^3 \leq C_2n^{-1/2}
\end{equation*}
hold. Moreover, by Assumption \ref{assump: nondegeneratevar},
\begin{equation*}
    \sigma^2 = \text{Var}(\sum_iX_i) = n\text{Var}(\hat{\tau}^{k, k^{'}})
\end{equation*}
is at least $O(1)$. Note that $X_i$ is a function of $H_i$, hence $X_i$ and $X_j$ are not independent if and only if $H_i$ and $H_j$ are not independent. Since $d_n = o(n^{1/4})$, we know that the maximal degree of the dependency graph of $X_i$'s is $o(n^{1/4})$. Now we apply Lemma \ref{lemma: steinbound}. Since $\sigma^2$ is at least $O(1)$, we get:
\begin{equation*}
    \text{RHS of } (\ref{bound}) = o(n^{-1/8}) + o(1) \rightarrow 0
\end{equation*}
We're done.
\end{proof}
\begin{remark}
In fact, with the tools in \cite{leung2022}, we can proof this theorem with a weaker condition on $d_n$: $d_n = O(\log n)$.
\end{remark}
\begin{proof}[Proof of Example \ref{cor: cor1}]
Note that $H_i$ is a function of $W_i$ and $W_j$'s for $j$ being a neighbor of $i$. If $H_i$ and $H_j$ are dependent, there must be the case that $(\{i\} \cup \mathcal{N}_i) \cap (\{j\} \cup \mathcal{N}_j)$ is nonempty since we have the Bernoulli design. Hence, for each fixed unit $i$, there are at most $\delta_n$ units such that the above intersection is nonempty.
\end{proof}
\begin{proof}[Proof of Example \ref{cor: cor2}]
We use the same reasoning as in the above proof. The only change is that now we know that each unit is belonged to a group and units in the group are connected. Therefore, for each fixed unit $i$, all the units outside the group will not have effect on unit $i$. As a result, we can have $r_n = o(n^{1/4})$.
\end{proof}
\begin{proof}[Proof of Example \ref{cor: cor3}]
Since we do not have Bernoulli design anymore, there might be the case that $W_i$ and $W_j$ are dependent, hence except $(\{i\} \cup \mathcal{N}_i) \cap (\{j\} \cup \mathcal{N}_j)$ is nonempty, there is another case that makes $H_i$ and $H_j$ dependent: a neighbor of $i$ is in the same cluster as a neighbor of $j$. For this case, we have at most $\delta_n c_n$ such $j$'s for a fixed unit $i$. Hence, in total, there are at most $\delta_n^2 + \delta_n c_n$ $j$'s such that $H_i$ and $H_j$ are dependent.
\end{proof}

\begin{proof} [Proof of Proposition \ref{prop: ci}]
We first prove the first part of the proposition. The proof is based on A.7 in \cite{aronow2017}. To start with, for any $(i, j) \in \{1, \cdots, \} \times \{1, \cdots, n\}$, we define $e_{ij} = 1$ if $H_i$ and $H_j$ are dependent and $0$ otherwise. Let $a_{ij}(H_i, H_j)$ be the sum of the elements in $\widehat{\text{Var}}(\hat{\tau}^{k, k^{'}})$ that incorporate $i$ and $j$, then
\begin{equation*}
\begin{split}
    \text{Var}\left(\widehat{\text{Var}}(\hat{\tau}^{k, k^{'}})\right) \leq n^{-4}\text{Var}\left[\sum_{i = 1}^{n}\sum_{j = 1}^{n}e_{ij}a_{ij}(H_i, H_j)\right] \\
    = n^{-4}\sum_{i = 1}^n\sum_{j = 1}^n\sum_{k = 1}^n\sum_{l = 1}^n\text{Cov}\left[e_{ij}a_{ij}(H_i, H_j), e_{kl}a_{kl}(H_k, H_l)\right]
\end{split}
\end{equation*}
Note that $\text{Cov}\left[e_{ij}a_{ij}(H_i, H_j), e_{kl}a_{kl}(H_k, H_l)\right]$ is nonzero if and only if $e_{ij} = 1, e_{kl} = 1$ and at least one of $e_{ik}, e_{il}, e_{jk}, e_{jl}$ is 1. In total, there are at most $4nd_n^3$ $(i, j, k, l)$'s satisfying this condition. And by Assumption \ref{assump: bddpo} and \ref{assump: overlap}, each covariance term is bounded, so we know that $\text{Var}\left(\widehat{\text{Var}}(\hat{\tau}^{k, k^{'}})\right) = o(n^{-4}\times n\times n^{3/4}) \rightarrow 0$ as $n \rightarrow \infty$. Then by Chebyshev's inequality,
\begin{equation*}
\left|\widehat{\text{Var}}(\sqrt{n}\hat{\tau}^{k, k^{'}}) - \mathbb{E}\left[\widehat{\text{Var}}(\sqrt{n}\hat{\tau}^{k, k^{'}})\right]\right| = o_p(1). 
\end{equation*}
Since $\mathbb{E}\left[\widehat{\text{Var}}(\hat{\tau}^{k, k^{'}})\right] \geq \text{Var}(\hat{\tau}^{k, k^{'}})$,
\begin{equation*}
    \mathbb{P}\left(\frac{\widehat{\text{Var}}(\hat{\tau}^{k, k^{'}})}{\text{Var}(\hat{\tau}^{k, k^{'}})} \geq 1 - \delta\right) \rightarrow 1
\end{equation*}
for any $\delta > 0$.

Now we can prove the second part of the proposition. We have that
\small
\begin{align}
     LHS &= \mathbb{P}\left(\left|\frac{\sqrt{n}(\hat{\tau}^{k, k^{'}} - \tau^{k, k^{'}})}{\sqrt{\text{Var}(\sqrt{n}\hat{\tau}^{k, k^{'}})}}\right| \leq \frac{z_{1 - \frac{\alpha}{2}}}{\sqrt{1 - \delta}}\sqrt{\frac{\widehat{\text{Var}}(\sqrt{n}\hat{\tau}^{k, k^{'}})}{\text{Var}(\sqrt{n}\hat{\tau}^{k, k^{'}})}}\right) \notag \\ 
     &\geq \mathbb{P}\left(\left|\frac{\sqrt{n}(\hat{\tau}^{k, k^{'}} - \tau^{k, k^{'}})}{\sqrt{\text{Var}(\sqrt{n}\hat{\tau}^{k, k^{'}})}}\right| \leq \frac{z_{1 - \frac{\alpha}{2}}}{\sqrt{1 - \delta}}\sqrt{\frac{\widehat{\text{Var}}(\sqrt{n}\hat{\tau}^{k, k^{'}})}{\text{Var}(\sqrt{n}\hat{\tau}^{k, k^{'}})}}
     \text{ and } \frac{\widehat{\text{Var}}(\sqrt{n}\hat{\tau}^{k, k^{'}})}{\text{Var}(\sqrt{n}\hat{\tau}^{k, k^{'}})} \geq 1 - \delta\right) \notag \\
     &\geq \mathbb{P}\left(\left|\frac{\sqrt{n}(\hat{\tau}^{k, k^{'}} - \tau^{k, k^{'}})}{\sqrt{\text{Var}(\sqrt{n}\hat{\tau}^{k, k^{'}})}}\right| \leq z_{1 - \frac{\alpha}{2}} 
     \text{ and } \frac{\widehat{\text{Var}}(\sqrt{n}\hat{\tau}^{k, k^{'}})}{\text{Var}(\sqrt{n}\hat{\tau}^{k, k^{'}})} \geq 1 - \delta\right) \notag \\
     &= \mathbb{P}\left(\left|\frac{\sqrt{n}(\hat{\tau}^{k, k^{'}} - \tau^{k, k^{'}})}{\sqrt{\text{Var}(\sqrt{n}\hat{\tau}^{k, k^{'}})}}\right| \leq z_{1 - \frac{\alpha}{2}}\right)
     - \mathbb{P}\left(\left|\frac{\sqrt{n}(\hat{\tau}^{k, k^{'}} - \tau^{k, k^{'}})}{\sqrt{\text{Var}(\sqrt{n}\hat{\tau}^{k, k^{'}})}}\right| \leq z_{1 - \frac{\alpha}{2}}
     \text{ and }\frac{\widehat{\text{Var}}(\sqrt{n}\hat{\tau}^{k, k^{'}})}{\text{Var}(\sqrt{n}\hat{\tau}^{k, k^{'}})} < 1 - \delta\right) \label{eq: prop1_temp}
\end{align}
Now,
\begin{align*}
    \eqref{eq: prop1_temp} &\geq \mathbb{P}\left(\left|\frac{\sqrt{n}(\hat{\tau}^{k, k^{'}} - \tau^{k, k^{'}})}{\sqrt{\text{Var}(\sqrt{n}\hat{\tau}^{k, k^{'}})}}\right| \leq z_{1 - \frac{\alpha}{2}}\right) \\ 
     &- \mathbb{P}\left(\frac{\widehat{\text{Var}}(\sqrt{n}\hat{\tau}^{k, k^{'}})}{\text{Var}(\sqrt{n}\hat{\tau}^{k, k^{'}})} < 1 - \delta\right) \\
     &= \mathbb{P}\left(\left|\frac{\sqrt{n}(\hat{\tau}^{k, k^{'}} - \tau^{k, k^{'}})}{\sqrt{\text{Var}(\sqrt{n}\hat{\tau}^{k, k^{'}})}}\right| \leq z_{1 - \frac{\alpha}{2}}\right)
     - \mathbb{P}\left(\frac{\widehat{\text{Var}}(\hat{\tau}^{k, k^{'}})}{\text{Var}(\hat{\tau}^{k, k^{'}})} < 1 - \delta\right) \\
     &\rightarrow 1 - \alpha
\end{align*}
\normalsize
as $n \rightarrow \infty$ by the first part and Theorem \ref{thm: tec_clt_pure_spat_inter}.
\end{proof}

\begin{proof} [Proof of Theorem \ref{thm: temporalavgclt_fixedT}]
We use a characteristic function argument. We first note that
\begin{align*}
    \frac{\sqrt{nT}(\hat{\bar{\tau}}^{k, k^{'}} - \bar{\tau}^{k, k^{'}})}{\sqrt{\frac{1}{T}\sum_{t=1}^{T}\sigma^2_{n, t}}} &= \frac{\sqrt{nT}(\frac{1}{T}\sum_{t = 1}^{T}\hat{\tau}^{k, k^{'}}_{t} - \frac{1}{T}\sum_{t = 1}^{T}\tau^{k, k^{'}}_t)}{\sqrt{\frac{1}{T}\sum_{t=1}^{T}\sigma^2_{n, t}}} \\
    &= \frac{\sqrt{T}\frac{1}{T}\sum_{t = 1}^{T}\sqrt{n}(\hat{\tau}^{k, k^{'}}_{t} - \tau^{k, k^{'}}_t)}{\sqrt{\frac{1}{T}\sum_{t=1}^{T}\sigma^2_{n, t}}} \\
    &= \frac{\frac{1}{\sqrt{T}}\sum_{t = 1}^{T}X_{n, t}}{\sqrt{\frac{1}{T}\sum_{t=1}^{T}\sigma^2_{n, t}}},
\end{align*}
where $X_{n, t} = \sqrt{n}(\hat{\tau}^{k, k^{'}}_{t} - \tau^{k, k^{'}}_t)$. Now,
\begin{align}
    &\mathbb{E}\left[\exp\Bigg\{i\lambda\frac{\sqrt{nT}(\hat{\bar{\tau}}^{k, k^{'}} - \bar{\tau}^{k, k^{'}})}{\sqrt{\frac{1}{T}\sum_{t=1}^{T}\sigma^2_{n, t}}}\Bigg\}\right] \\
    &= \mathbb{E}\left[\exp\Bigg\{i\lambda\frac{\frac{1}{\sqrt{T}}\sum_{t = 1}^{T}X_{n, t}}{\sqrt{\frac{1}{T}\sum_{t=1}^{T}\sigma^2_{n, t}}}\Bigg\}\right] \notag \\
    &= \prod_{t = 1}^{T}\mathbb{E}\left[\exp\Bigg\{i\lambda\frac{\frac{1}{\sqrt{T}}X_{n, t}}{\sqrt{\frac{1}{T}\sum_{t=1}^{T}\sigma^2_{n, t}}}\Bigg\}\right] \notag \\
    &= \prod_{t = 1}^{T}\mathbb{E}\left[\exp\Bigg\{i\frac{\lambda\sigma_{n, t}}{\sqrt{\sum_{t=1}^{T}\sigma^2_{n, t}}}\frac{X_{n, t}}{\sigma_{n, t}}\Bigg\}\right] \notag \\
    &= \prod_{t = 1}^{T}\phi_{\frac{X_{n, t}}{\sigma_{n, t}}}\left(\frac{\lambda\sigma_{n, t}}{\sqrt{\sum_{t=1}^{T}\sigma^2_{n, t}}}\right) \label{prod}
\end{align}
The second equality follows from our assumption that assignment vectors are independent across time and $\phi_X$ denotes the characteristic function of a random variable $X$. Pick $\epsilon > 0$. Now, to conclude the proof, we note that
\begin{equation*}
    \phi_{\frac{X_{n, t}}{\sigma_{n, t}}}(\theta) \rightarrow e^{-\frac{\theta^2}{2}}
\end{equation*}
for any $t \in \{1, \cdots, T\}$. Moreover, for each $t$, the convergence is actually uniform on any bounded interval. Therefore, for any $t \in \{1, \cdots, T\}$, 
\begin{equation*}
    \phi_{\frac{X_{n, t}}{\sigma_{n, t}}}(\theta) \rightarrow e^{-\frac{\theta^2}{2}} \,\, \text{uniformly on} \,\, (0, 1).
\end{equation*}
Note that 
\begin{equation*}
    \frac{\lambda\sigma_{n, t}}{\sqrt{\sum_{t=1}^{T}\sigma^2_{n, t}}} \in (0, 1),
\end{equation*}
so for any $t$, $\exists N_t \in \mathbb{N}$ such that for any $n \geq N_t$,
\begin{equation*}
\begin{split}
    \Bigg|\phi_{\frac{X_{n, t}}{\sigma_{n, t}}}\left(\frac{\lambda\sigma_{n, t}}{\sqrt{\sum_{t=1}^{T}\sigma^2_{n, t}}}\right) - \exp\Bigg\{-\frac{1}{2}\frac{\lambda^2\sigma_{n, t}^2}{\sum_{t=1}^{T}\sigma^2_{n, t}}\Bigg\}\Bigg| = |\epsilon_t| \\
    \leq \frac{1}{2^K}.
\end{split}
\end{equation*}
Let $N = \max\{N_1, \cdots, N_T\}$, then for all $n \geq N$, and for all $t \in \{1, \cdots, T\}$, 
\begin{equation*}
    \Bigg|\phi_{\frac{X_{n, t}}{\sigma_{n, t}}}\left(\frac{\lambda\sigma_{n, t}}{\sqrt{\sum_{t=1}^{T}\sigma^2_{n, t}}}\right) - \exp\Bigg\{-\frac{1}{2}\frac{\lambda^2\sigma_{n, t}^2}{\sum_{t=1}^{T}\sigma^2_{n, t}}\Bigg\}\Bigg| = |\epsilon_t| \leq \frac{1}{2^K},
\end{equation*}
where $K$ is any big number we want. Now,
\begin{align*}
    (\ref{prod}) &= \prod_{t = 1}^{T}\left(\exp\Bigg\{-\frac{1}{2}\frac{\lambda^2\sigma_{n, t}^2}{\sum_{t=1}^{T}\sigma^2_{n, t}}\Bigg\} + \epsilon_t\right) \\
    &= \exp\bigg\{-\frac{1}{2}\lambda^2\bigg\} + R(\epsilon_t),
\end{align*}
where $R(\epsilon_t)$ is a remainder term that is the sum of several monomial terms of $\epsilon_t$'s. Note that $\exp\Bigg\{-\frac{1}{2}\frac{\lambda^2\sigma_{n, t}^2}{\sum_{t=1}^{T}\sigma^2_{n, t}}\Bigg\}$ is actually bounded by 1, hence by making $K$ sufficiently large, we can make $R(\epsilon_t)$ arbitrarily small. Pick such $K$, then we know that for sufficiently large $n$,
\begin{equation*}
    \Bigg|\mathbb{E}\left[\exp\Bigg\{i\lambda\frac{\sqrt{nT}(\hat{\bar{\tau}}^{k, k^{'}} - \bar{\tau}^{k, k^{'}})}{\sqrt{\frac{1}{T}\sum_{t=1}^{T}\sigma^2_{n, t}}}\Bigg\}\right] - \exp\bigg\{-\frac{1}{2}\lambda^2\bigg\} \Bigg| \leq \epsilon.
\end{equation*}
Hence, by standard characteristic function argument, we complete the proof of the theorem.
\end{proof}
To prove Theorem \ref{thm: atecclt}, we first state the following version of Lindeberg-Feller central limit theorem.
\begin{lemma}[Lindeberg-Feller CLT]
\label{lemma: lindebergclt}
Let $\{k_n\}_{n \geq 1}$ be a sequence of positive integers increasing to infinity. For each $n$, let $\{X_{n, i}\}_{1 \leq i \leq k_n}$ is a collection of independent random variables. Let $\mu_{n, i} \coloneqq \mathbb{E}(X_{n, i})$ and
\begin{equation*}
    s_n^2 \coloneqq \sum_{i = 1}^{k_n}\text{Var}(X_{n, i}).
\end{equation*}
Suppose that for any $\epsilon > 0$,
\begin{equation}
    \lim_{n \rightarrow \infty}\frac{1}{s_n^2}\sum_{i = 1}^{k_n}\mathbb{E}\left((X_{n, i} - \mu_{n, i})^2; |X_{n, i} - \mu_{n, i}| \geq \epsilon s_n\right) = 0. \label{lindebergcondition}
\end{equation}
Then the random variable
\begin{equation*}
    \frac{\sum_{i = 1}^{k_n}(X_{n, i} - \mu_{n, i})}{s_n} \xrightarrow{d} \mathcal{N}(0, 1)
\end{equation*}
as $n \rightarrow \infty$.
\end{lemma}
\begin{proof} [Proof of Theorem \ref{thm: atecclt}]
We first prove the theorem with condition (\ref{atecclt: nodegcond}). We note that $\sqrt{nT}(\hat{\bar{\tau}}^{k, k^{'}} - \bar{\tau}^{k, k^{'}}) = \sum_{t = 1}^{T}\sqrt{\frac{n}{T}}(\hat{\tau}^{k, k^{'}}_{t} - \tau_t^{k, k^{'}})$. Let $X_{n, t} = \sqrt{\frac{n}{T}}\hat{\tau}^{k, k^{'}}_{t}$, then $\mu_{n, t} = \sqrt{\frac{n}{T}}\tau_t^{k, k^{'}}$, so the numerator is exactly $\sum_{t = 1}^{T}(X_{n, t} - \mu_{n, t})$. Moreover, note that for any $n$, $X_{n, 1}, \cdots, X_{n, T}$ are independent by the pure population interference assumption. Now,
\begin{align*}
    s_n^2 &=  \sum_{t = 1}^{T}\text{Var}(X_{n, t}) \\
    &= \sum_{t = 1}^{T}\text{Var}\left(\sqrt{\frac{n}{T}}\hat{\tau}^{k, k^{'}}_{t}\right) \\
    &= \frac{1}{T}\sum_{t = 1}^{T}\text{Var}(\sqrt{n}\hat{\tau}^{k, k^{'}}_{t}) \\
    &= \frac{1}{T}\sum_{t=1}^{T}\sigma^2_{n, t}.
\end{align*}
Hence, to finish the proof, we only need to check (\ref{lindebergcondition}) is satisfied. Notice that for any $\epsilon > 0$,
\begin{align*}
    |X_{n, t} - \mu_{n, t}| \geq \epsilon s_n &\Leftrightarrow \left|\sqrt{\frac{n}{T}}\hat{\tau}^{k, k^{'}}_{t} - \sqrt{\frac{n}{T}}\tau_t^{k, k^{'}}\right| \geq \epsilon\sqrt{\frac{1}{T}\sum_{t=1}^{T}\sigma^2_{n, t}} \\
    &\Leftrightarrow \left|\hat{\tau}^{k, k^{'}}_{t} - \tau_t^{k, k^{'}}\right| \geq \epsilon \sqrt{\frac{1}{n}\sum_{t=1}^{T}\sigma^2_{n, t}}
\end{align*}
By Assumption \ref{assump: nondegeneratevar}, $\sigma_{n, t}^2 \geq c$ for some $c > 0$ and for all $n$ large. Hence
\begin{equation*}
    \epsilon \sqrt{\frac{1}{n}\sum_{t=1}^{T}\sigma^2_{n, t}} \geq \epsilon \sqrt{\frac{T}{n}c} \rightarrow \infty.
\end{equation*}
Note that by Assumptions \ref{assump: bddpo} and \ref{assump: overlap}, $\left|\hat{\tau}^{k, k^{'}}_{t} - \tau_t^{k, k^{'}}\right|$ is uniformly bounded. Hence for sufficiently large $n$, $\left|\hat{\tau}^{k, k^{'}}_{t} - \tau_t^{k, k^{'}}\right| < \epsilon \sqrt{\frac{1}{n}\sum_{t=1}^{T}\sigma^2_{n, t}}$ for all $t$. Therefore, for sufficiently large $n$, 
\begin{equation*}
    \frac{1}{s_n^2}\sum_{t = 1}^{T}\mathbb{E}\left((X_{n, t} - \mu_{n, t})^2; |X_{n, t} - \mu_{n, t}| \geq \epsilon s_n\right) = 0.
\end{equation*}
As a result, (\ref{lindebergcondition}) is satisfied. We're done.
The proof of this theorem with condition (\ref{atecclt: degcond}) is exactly the same as in single time step case once we notice that the numerator is just a sum of $nT$ mean 0 dependent random variables.
\end{proof}
To prove Theorem \ref{thm: ateclyapunov}, we need the following version of Lyapunov central limit theorem.
\begin{lemma} [Lyapunov CLT]
\label{lemma: lyapunovclt}
Let $\{X_n\}_{n = 1}^{\infty}$ be a sequence of independent random variables. Let $\mu_i \coloneqq \mathbb{E}(X_i)$ and
\begin{equation*}
    s_n^2 = \sum_{i = 1}^{n}\text{Var}(X_i).
\end{equation*}
If for some $\delta > 0$,
\begin{equation}
    \lim_{n \rightarrow \infty}\frac{1}{s_n^{2 + \delta}}\sum_{i = 1}^{n}\mathbb{E}|X_i - \mu_i|^{2 + \delta}= 0, \label{lyapunov}
\end{equation}
then the random variable
\begin{equation*}
    \frac{\sum_{i = 1}^n(X_i - \mu_i)}{s_n} \xrightarrow{d} \mathcal{N}(0, 1)
\end{equation*}
\end{lemma}
\begin{proof} [Proof of Theorem \ref{thm: ateclyapunov}]
This time, we let $X_t$ = $\sqrt{\frac{n}{T}}\hat{\tau}^{k, k^{'}}_{t}$ then the numerator is $\sum_{t = 1}^{T}(X_t - \mu_t)$. Since we have pure population interference, $\{X_t\}_{t = 1}^{\infty}$ are independent. Now,
\begin{align*}
    s_T^2 &= \sum_{t = 1}^{T}\text{Var}(X_t) \\
          &= \frac{1}{T}\sum_{t=1}^{T}\sigma^2_{n, t}.
\end{align*}
Hence, we only need to check (\ref{lyapunov}). We have that
\begin{align*}
\begin{split}
    &\lim_{T \rightarrow \infty}\frac{1}{s_T^{2 + \delta}}\sum_{t = 1}^{T}\mathbb{E}|X_t - \mu_t|^{2 + \delta} \\ &= \lim_{T \rightarrow \infty}\frac{1}{s_T^{2 + \delta}}\left(\frac{n}{T}\right)^{1 + \frac{\delta}{2}}\sum_{t = 1}^{T}\mathbb{E}\left|\hat{\tau}^{k, k^{'}}_{t} - \tau_t^{k, k^{'}}\right|^{2 + \delta}
\end{split}
\end{align*}
Now, by Assumptions \ref{assump: bddpo} and \ref{assump: overlap}, $\exists M > 0$ such that $\left|\hat{\tau}^{k, k^{'}}_{t} - \tau_t^{k, k^{'}}\right| \leq M$ for all $t$. Hence,
\begin{align*}
    &\frac{1}{s_T^{2 + \delta}}\left(\frac{n}{T}\right)^{1 + \frac{\delta}{2}}\sum_{t = 1}^{T}\mathbb{E}\left|\hat{\tau}^{k, k^{'}}_{t} - \tau_t^{k, k^{'}}\right|^{2 + \delta}\\
    &\leq \frac{1}{s_T^{2 + \delta}}\left(\frac{n}{T}\right)^{1 + \frac{\delta}{2}}TM^{2 + \delta} \\
    &= \frac{1}{s_T^{2 + \delta}}\frac{n^{1 + \frac{\delta}{2}}}{T}M^{2 + \delta}
\end{align*}
If $T \rightarrow \infty$, $\frac{1}{s_T^{2 + \delta}}\frac{n^{1 + \frac{\delta}{2}}}{T}M^{2 + \delta} \rightarrow 0$. Therefore, (\ref{lyapunov}) is satisfied. We're done.
\end{proof}

\begin{proof} [Proof of Proposition \ref{prop: ci_atec}]
Now we can prove the second part of the proposition. We have that
\begin{align}
     LHS &= \mathbb{P}\left(\left|\frac{\sqrt{nT}(\hat{\bar{\tau}}^{k, k^{'}} - \bar{\tau}^{k, k^{'}})}{\sqrt{\frac{1}{T}\sum_{t = 1}^T\text{Var}(\sqrt{n}\hat{\tau}^{k, k^{'}}_t)}}\right| \leq \frac{z_{1 - \frac{\alpha}{2}}}{\sqrt{1 - \delta}}\sqrt{\frac{\frac{1}{T}\sum_{t = 1}^T\widehat{\text{Var}}(\sqrt{n}\hat{\tau}^{k, k^{'}}_t)}{\frac{1}{T}\sum_{t = 1}^T\text{Var}(\sqrt{n}\hat{\tau}^{k, k^{'}}_t)}}\right) \notag \\
     &\geq \mathbb{P}\left(\left|\frac{\sqrt{nT}(\hat{\bar{\tau}}^{k, k^{'}} - \bar{\tau}^{k, k^{'}})}{\sqrt{\frac{1}{T}\sum_{t = 1}^T\text{Var}(\sqrt{n}\hat{\tau}^{k, k^{'}}_t)}}\right| \leq \frac{z_{1 - \frac{\alpha}{2}}}{\sqrt{1 - \delta}}\sqrt{\frac{\frac{1}{T}\sum_{t = 1}^T\widehat{\text{Var}}(\sqrt{n}\hat{\tau}^{k, k^{'}}_t)}{\frac{1}{T}\sum_{t = 1}^T\text{Var}(\sqrt{n}\hat{\tau}^{k, k^{'}}_t)}}\text{ and } \frac{\frac{1}{T}\sum_{t = 1}^T\widehat{\text{Var}}(\hat{\tau}^{k, k^{'}}_t)}{\frac{1}{T}\sum_{t = 1}^T\text{Var}(\hat{\tau}^{k, k^{'}}_t)} \geq 1 - \delta\right) \notag \\
     &\geq \mathbb{P}\left(\left|\frac{\sqrt{nT}(\hat{\bar{\tau}}^{k, k^{'}} - \bar{\tau}^{k, k^{'}})}{\sqrt{\frac{1}{T}\sum_{t = 1}^T\text{Var}(\sqrt{n}\hat{\tau}^{k, k^{'}}_t)}}\right| \leq z_{1 - \frac{\alpha}{2}} \text{ and } \frac{\frac{1}{T}\sum_{t = 1}^T\widehat{\text{Var}}(\hat{\tau}^{k, k^{'}}_t)}{\frac{1}{T}\sum_{t = 1}^T\text{Var}(\hat{\tau}^{k, k^{'}}_t)} \geq 1 - \delta\right) \label{eq: temp_prop5}
\end{align}
Furthermore, we have that
\begin{align*}
    \eqref{eq: temp_prop5} &= \mathbb{P}\left(\left|\frac{\sqrt{nT}(\hat{\bar{\tau}}^{k, k^{'}} - \bar{\tau}^{k, k^{'}})}{\sqrt{\frac{1}{T}\sum_{t = 1}^T\text{Var}(\sqrt{n}\hat{\tau}^{k, k^{'}}_t)}}\right| \leq z_{1 - \frac{\alpha}{2}}\right) \notag \\ 
    &\qquad - \mathbb{P}\left(\left|\frac{\sqrt{nT}(\hat{\bar{\tau}}^{k, k^{'}} - \bar{\tau}^{k, k^{'}})}{\sqrt{\frac{1}{T}\sum_{t = 1}^T\text{Var}(\sqrt{n}\hat{\tau}^{k, k^{'}}_t)}}\right| \leq z_{1 - \frac{\alpha}{2}} \text{ and } \frac{\frac{1}{T}\sum_{t = 1}^T\widehat{\text{Var}}(\hat{\tau}^{k, k^{'}}_t)}{\frac{1}{T}\sum_{t = 1}^T\text{Var}(\hat{\tau}^{k, k^{'}}_t)} < 1 - \delta\right) \notag \\
     &\geq \mathbb{P}\left(\left|\frac{\sqrt{nT}(\hat{\bar{\tau}}^{k, k^{'}} - \bar{\tau}^{k, k^{'}})}{\sqrt{\frac{1}{T}\sum_{t = 1}^T\text{Var}(\sqrt{n}\hat{\tau}^{k, k^{'}}_t)}}\right| \leq z_{1 - \frac{\alpha}{2}}\right) - \mathbb{P}\left(\frac{\frac{1}{T}\sum_{t = 1}^T\widehat{\text{Var}}(\hat{\tau}^{k, k^{'}}_t)}{\frac{1}{T}\sum_{t = 1}^T\text{Var}(\hat{\tau}^{k, k^{'}}_t)} < 1 - \delta\right)
\end{align*}
So if we can show $\mathbb{P}\left(\frac{\frac{1}{T}\sum_{t = 1}^T\widehat{\text{Var}}(\hat{\tau}^{k, k^{'}}_t)}{\frac{1}{T}\sum_{t = 1}^T\text{Var}(\hat{\tau}^{k, k^{'}}_t)} \geq 1 - \delta\right) \rightarrow 0$ then we are done. Notice that
\begin{equation}
    \text{Var}\left(\frac{1}{T}\sum_{t = 1}^T\widehat{\text{Var}}(\hat{\tau}^{k, k^{'}}_t)\right) = \frac{1}{T^2}\sum_{t = 1}^{T}\text{Var}\left(\widehat{\text{Var}}(\hat{\tau}^{k, k^{'}}_t)\right). \label{var_eq: atec}
\end{equation}
If $T$ is fixed (i.e., Theorem \ref{thm: temporalavgclt_fixedT} holds), then by what we have in Proposition \ref{prop: ci}, we immediately have that (\ref{var_eq: atec}) $\rightarrow 0$ and we are done. Now suppose Theorem \ref{thm: atecclt} holds. Recall that
\begin{equation*}
\begin{split}
    \text{Var}\left(\widehat{\text{Var}}(\hat{\tau}^{k, k^{'}})\right) \leq n^{-4}\sum_{i = 1}^n\sum_{j = 1}^n\sum_{k = 1}^n\sum_{l = 1}^n\text{Cov}\left[e_{ij}a_{ij}(H_i, H_j), \right.\\
    \left.e_{kl}a_{kl}(H_k, H_l)\right],
\end{split}
\end{equation*}
which implies that $\text{Var}\left(\widehat{\text{Var}}(\hat{\tau}^{k, k^{'}}_t)\right)$ is uniformly bounded by a constant $M$ by Assumption \ref{assump: bddpo} and \ref{assump: overlap}. So
\begin{align*}
     \frac{1}{T^2}\sum_{t = 1}^{T}\text{Var}\left(\widehat{\text{Var}}(\hat{\tau}^{k, k^{'}}_t)\right) \leq \frac{1}{T^2}\sum_{t = 1}^{T}M = \frac{M}{T} \rightarrow 0
\end{align*}
as $T \rightarrow 0$. So in the regime where both $n$ and $T$ go to infinity (i.t. Theorem \ref{thm: atecclt} holds) or $T$ goes to infinity (i.e., Theorem \ref{thm: ateclyapunov} holds), (\ref{var_eq: atec}) $\rightarrow 0$ and we are done.
\end{proof}

\begin{proof} [Proof of Theorem \ref{thm: tec_mixed_inter}]
This should be exactly the same as our proof of Theorem \ref{thm: tec_clt_pure_spat_inter}.
\end{proof}

\begin{proof} [Proof of Proposition \ref{prop: k=2cvxbias}]
\begin{align}
    |\mathbb{E}[\hat{\tau}^c_t] - \tau_t^{TE}| &= |\mathbb{E}[\alpha\hat{\tau}_t^{TE} + (1 - \alpha)\hat{\tau}_{t - 1}^{TE}] - \tau_t^{TE}| \notag \\
    &= |\alpha\tau_t^{TE} + (1 - \alpha)\tau_{t - 1}^{TE} - \tau_t^{TE}| \notag \\
    &= |(1-\alpha)(\tau_{t - 1}^{TE} - \tau_t^{TE})| \notag \\
    &= (1 - \alpha)|\tau_{t - 1}^{TE} - \tau_t^{TE}| \notag
\end{align}
The second equality follows from unbiasedness of $\hat{\tau}_t^{TE}$ and $\hat{\tau}_{t - 1}^{TE}$. To further bound the bias, we need to bound $|\tau_{t - 1}^{TE} - \tau_t^{TE}|$. We do this below.
\begin{align}
    |\tau_{t - 1}^{TE} - \tau_t^{TE}| &= \bigg|\left(\frac{1}{n}\sum_{i = 1}^{n}Y_{i, t}(h_i^1) - \frac{1}{n}\sum_{i = 1}^{n}Y_{i, t}(h_i^0)\right) - \left(\frac{1}{n}\sum_{i = 1}^{n}Y_{i, t-1}(h_i^1) - \frac{1}{n}\sum_{i = 1}^{n}Y_{i, t-1}(h_i^0)\right)\bigg| \notag \\
    &= \bigg|\frac{1}{n}\sum_{i = 1}^{n}(Y_{i, t}(h_i^1) - Y_{i, t-1}(h_i^1)) - \frac{1}{n}\sum_{i = 1}^{n}(Y_{i, t}(h_i^0) - Y_{i, t-1}(h_i^0))\bigg| \notag \\
    &\leq \frac{1}{n}\sum_{i = 1}^{n}|Y_{i, t}(h_i^1) - Y_{i, t-1}(h_i^1)| + \frac{1}{n}\sum_{i = 1}^{n}|Y_{i, t}(h_i^0) - Y_{i, t-1}(h_i^0)| \notag \\
    &\leq 2\epsilon, \notag
\end{align}
by our $\epsilon$-weak-stability assumption. Hence, 
\begin{equation*}
    |\mathbb{E}[\hat{\tau}^c_t] - \tau_t^{TE}| \leq 2(1 - \alpha)\epsilon.
\end{equation*}
\end{proof}
\begin{remark}
Note that following the exact derivation, we can know that
\begin{equation}
    |\tau_t^{TE} - \tau_{t^{'}}^{TE}| \leq 2|t - t^{'}|\epsilon \label{equa: generaldiff}
\end{equation}
\end{remark}
\begin{proposition} [Variance and Covariance of Horvitz-Thompson Type Estimators] 
\label{prop: var_and_cov_of_tauhat}
For each $i \in \{1, \cdots, n\}, t \in \{1, \cdots, T\}$, we let $\mathbb{P}(H_{i, t} = h_i^1) = \pi_{i, t}^1$, $\mathbb{P}(H_{i, t} = h_i^0) = \pi_{i, t}^{0}$, $\mathbb{P}(H_{j, t} = h_j^1) = \pi_{j, t}^{1}$ and $\mathbb{P}(H_{j, t} = h_j^0) = \pi_{j, t}^{0}$. Moreover, for each $i \neq j$ and $t$, we let $\mathbb{P}(H_{i, t} = h_i^1, H_{j, t} = h_j^1) = \pi_{i, j, t}^{1, 1}$, $\mathbb{P}(H_{i, t} = h_i^0, H_{j, t} = h_j^1) = \pi_{i, j, t}^{0, 1}$, $\mathbb{P}(H_{i, t} = h_i^1, H_{j, t} = h_j^0) = \pi_{i, j, t}^{1, 0}$ and $\mathbb{P}(H_{i, t} = h_i^0, H_{j, t} = h_j^0) = \pi_{i, j, t}^{0, 0}$, then
\small
\begin{equation}
\begin{split}
    Var(\hat{\tau}_t^{TE}) = \frac{1}{n^2}\sum_{i = 1}^{n}\left[\frac{Y_{i, t}^2(h_i^1)(1 - \pi_{i, t}^1)}{\pi_{i, t}^{1}} + \frac{Y_{i, t}^2(h_i^0)(1 - \pi_{i, t}^0)}{\pi_{i, t}^0} + 2Y_{i, t}(h_i^1)Y_{i, t}(h_i^0)\right] \\
    + \frac{2}{n^2}\sum_{1 \leq i < j \leq n}\left[\frac{Y_{i, t}(h_i^1)Y_{j, t}(h_j^1)(\pi_{i, j, t}^{1, 1} - \pi_{i, t}^{1}\pi_{j, t}^1)}{\pi_{i, t}^1\pi_{j, t}^1} - \frac{Y_{i, t}(h_i^0)Y_{j, t}(h_j^1)(\pi_{i, j, t}^{0, 1} - \pi_{i, t}^{0}\pi_{j, t}^1)}{\pi_{i, t}^0\pi_{j, t}^1} \right.\\
    - \left.\frac{Y_{i, t}(h_i^1)Y_{j, t}(h_j^0)(\pi_{i, j, t}^{1, 0} - \pi_{i, t}^{1}\pi_{j, t}^0)}{\pi_{i, t}^1\pi_{j, t}^0} + \frac{Y_{i, t}(h_i^0)Y_{j, t}(h_j^0)(\pi_{i, j, t}^{0, 0} - \pi_{i, t}^{0}\pi_{j, t}^0)}{\pi_{i, t}^0\pi_{j, t}^0}\right] \label{equa: var_formula}
\end{split}
\end{equation}
\normalsize
As for $Cov(\hat{\tau}_t^{TE}, \hat{\tau}_{t^{'}}^{TE})$, if we let $\mathbb{P}(H_{i, t} = h_i^1, H_{i, t^{'}} = h_i^1) = \pi_{i, t, t^{'}}^{1, 1}$, $\mathbb{P}(H_{i, t} = h_i^0, H_{i, t^{'}} = h_i^1) = \pi_{i, t, t^{'}}^{0, 1}$, $\mathbb{P}(H_{i, t} = h_i^1, H_{i, t^{'}} = h_i^0) = \pi_{i, t, t^{'}}^{1, 0}$, $\mathbb{P}(H_{i, t} = h_i^0, H_{i, t^{'}} = h_i^0) = \pi_{i, t, t^{'}}^{0, 0}$ and $\mathbb{P}(H_{i, t} = h_i^1, H_{j, t^{'}} = h_j^1) = \pi_{i, t, j, t^{'}}^{1, 1}$, $\mathbb{P}(H_{i, t} = h_i^0, H_{j, t^{'}} = h_j^1) = \pi_{i, t, j, t^{'}}^{0, 1}$, $\mathbb{P}(H_{i, t} = h_i^1, H_{j, t^{'}} = h_j^0) = \pi_{i, t, j, t^{'}}^{1, 0}$, $\mathbb{P}(H_{i, t} = h_i^0, H_{j, t^{'}} = h_j^0) = \pi_{i, t, j, t^{'}}^{0, 0}$, then we have the following expression for $Cov(\hat{\tau}_t^{TE}, \hat{\tau}_{t^{'}}^{TE})$:
\small
\begin{multline}
    \frac{1}{n^2}\sum_{i = 1}^{n}\left[\frac{Y_{i, t}(h_i^1)Y_{i, t^{'}}(h_i^1)(\pi_{i, t, t^{'}}^{1, 1} - \pi_{i, t}^{1}\pi_{i, t^{'}}^1)}{\pi_{i, t}^{1}\pi_{i, t^{'}}^1} - 
    \frac{Y_{i, t}(h_i^0)Y_{i, t^{'}}(h_i^1)(\pi_{i, t, t^{'}}^{0, 1} - \pi_{i, t}^{0}\pi_{i, t^{'}}^1)}{\pi_{i, t}^{0}\pi_{i, t^{'}}^1} \right.\\
    - \left.\frac{Y_{i, t}(h_i^1)Y_{i, t^{'}}(h_i^0)(\pi_{i, t, t^{'}}^{1, 0} - \pi_{i, t}^{1}\pi_{i, t^{'}}^0)}{\pi_{i, t}^{1}\pi_{i, t^{'}}^0} + \frac{Y_{i, t}(h_i^0)Y_{i, t^{'}}(h_i^0)(\pi_{i, t, t^{'}}^{0, 0} - \pi_{i, t}^{0}\pi_{i, t^{'}}^0)}{\pi_{i, t}^{0}\pi_{i, t^{'}}^0}\right] \\
    + \frac{2}{n^2}\sum_{1 \leq i < j \leq n}\left[\frac{Y_{i, t}(h_i^1)Y_{j, t^{'}}(h_j^1)(\pi_{i, t, j, t^{'}}^{1,1} - \pi_{i, t}^{1}\pi_{j, t^{'}}^{1})}{\pi_{i, t}^{1}\pi_{j, t^{'}}^{1}} - \frac{Y_{i, t}(h_i^0)Y_{j, t^{'}}(h_j^1)(\pi_{i, t, j, t^{'}}^{0,1}
    - \pi_{i, t}^{0}\pi_{j, t^{'}}^{1})}{\pi_{i, t}^{0}\pi_{j, t^{'}}^{1}} \right.\\
    - \left.\frac{Y_{i, t}(h_i^1)Y_{j, t^{'}}(h_j^0)(\pi_{i, t, j, t^{'}}^{1,0} - \pi_{i, t}^{1}\pi_{j, t^{'}}^{0})}{\pi_{i, t}^{1}\pi_{j, t^{'}}^{0}}
    + \frac{Y_{i, t}(h_i^0)Y_{j, t^{'}}(h_j^0)(\pi_{i, t, j, t^{'}}^{0,0} - \pi_{i, t}^{0}\pi_{j, t^{'}}^{0})}{\pi_{i, t}^{0}\pi_{j, t^{'}}^{0}}\right] \label{equa: cov_formula}
\end{multline}
\normalsize
\end{proposition}
\begin{proof} [Proof of Proposition \ref{prop: var_and_cov_of_tauhat}]
This can be done by direct calculations.
\end{proof}

\begin{proof} [Proof of Proposition \ref{prop: msereduction, k=2}]
We'd like to have reduction in MSE by using $\hat{\tau}_t^c$. By the bias-variance decomposition and note that $\hat{\tau}_t^{TE}$ is unbiased, this boils down to
\begin{equation*}
    Var(\hat{\tau}_t^c) + |\mathbb{E}[\hat{\tau}^c_t] - \tau_t^{TE}|^2 \leq Var(\hat{\tau}_t^{TE})
\end{equation*}
By Proposition \ref{prop: k=2cvxbias}, it suffices to have
\begin{equation*}
    Var(\hat{\tau}_t^c) + 4(1 - \alpha)^2\epsilon^2 \leq Var(\hat{\tau}_t^{TE}),
\end{equation*}
which is further equivalent to
\begin{equation}
\begin{split}
    \alpha^2 Var(\hat{\tau}_t^{TE}) + (1 - \alpha)^2 Var(\hat{\tau}_{t - 1}^{TE}) \\
    + 2\alpha(1 - \alpha)Cov(\hat{\tau}_t^{TE}, \hat{\tau}_{t-1}^{TE}) + 4(1 - \alpha)^2\epsilon^2 \\
    \leq Var(\hat{\tau}_t^{TE}) \label{varineq}
\end{split}
\end{equation}
Rewrite (\ref{varineq}), we have
\begin{multline}
    \left(4\epsilon^2 + Var(\hat{\tau}_t^{TE}) + Var(\hat{\tau}_{t - 1}^{TE}) - 2Cov(\hat{\tau}_t^{TE}, \hat{\tau}_{t-1}^{TE})\right)\alpha^2 \\
    - \left(8\epsilon^2 + 2Var(\hat{\tau}_{t - 1}^{TE}) - 2Cov(\hat{\tau}_t^{TE}, \hat{\tau}_{t-1}^{TE})\right)\alpha\\
    + \left(4\epsilon^2 + Var(\hat{\tau}_{t - 1}^{TE}) - Var(\hat{\tau}_{t}^{TE})\right) \leq 0 \label{msereducineq}
\end{multline}
Now we look at the left hand side of (\ref{msereducineq}), which is quadratic in $\alpha$. To ease notations, let $A = Var(\hat{\tau}_t^{TE})$, $B = Var(\hat{\tau}_{t - 1}^{TE})$ and $C = Cov(\hat{\tau}_t^{TE}, \hat{\tau}_{t-1}^{TE})$. It's easy to see that the left hand side achieves its minimum at $\alpha = \delta = 1 - \frac{2(A - C)}{8\epsilon^2 + 2A + 2B - 4C}$ and is 0 at $\alpha = 1$. So if we have $\delta < 1$, then for some $\alpha \in (0, 1)$, we have reduction in MSE. Moreover, if $\delta < \frac{1}{2}$, we then know that for $\alpha = \frac{1}{2}$, we also have smaller MSE by the property of quadratic functions. And simple algebra shows that $\delta < \frac{1}{2}$ is equivalent to $A - B > 4\epsilon^2$.
\end{proof}

\begin{proposition} [Estimators of variance]
\label{prop: totaleffectvarest}
We define two estimators of the variance:
\small
\begin{multline}
    \widehat{Var}^u(\hat{\tau}_t^{TE}) = \frac{1}{n^2}\sum_{i = 1}^{n}\left[\mathbf{1}(H_{i, t} = h_i^1)(1 - \pi_{i, t}^1)\left(\frac{Y_{i, t}}{\pi_{i, t}^1}\right)^2
    + \mathbf{1}(H_{i, t} = h_i^0)(1 - \pi_{i, t}^0)\left(\frac{Y_{i, t}}{\pi_{i, t}^0}\right)^2\right. \\
    + \left.\frac{Y_{i, t}^2}{\pi_{i, t}^1}\mathbf{1}(H_{i, t} = h_i^1) + \frac{Y_{i, t}^2}{\pi_{i, t}^0}\mathbf{1}(H_{i, t} = h_i^0)\right] \\
    + \frac{2}{n^2}\sum_{1 \leq i < j \leq n}\left[\mathbf{1}(\pi_{i,j,t}^{1,1} \neq 0)\mathbf{1}(H_{i, t} = h_i^1)\mathbf{1}(H_{j, t} = h_j^1)\frac{(\pi_{i, j, t}^{1, 1} - \pi_{i, t}^1\pi_{j, t}^1)Y_{i, t}Y_{j, t}}{\pi_{i, t}^1\pi_{j, t}^1\pi_{i, j, t}^{1, 1}} \right.\\
    - \left(\mathbf{1}(\pi_{i, j, t}^{0, 1} \neq 0)\mathbf{1}(H_{i, t} = h_i^0)\mathbf{1}(H_{j, t} = h_j^1)\frac{(\pi_{i, j, t}^{0, 1} - \pi_{i, t}^0\pi_{j, t}^1)Y_{i, t}Y_{j, t}}{\pi_{i, t}^0\pi_{j, t}^1\pi_{i, j, t}^{0, 1}} \right.\\
    \left.-\mathbf{1}(\pi_{i, j, t}^{0, 1} = 0)\left(\frac{\mathbf{1}(H_{i, t} = h_i^0)Y_{i, t}^2}{2\pi_{i, t}^{0}} + \frac{\mathbf{1}(H_{j, t} = h_j^1)Y_{j, t}^2}{2\pi_{j, t}^1}\right)\right) \\
    - \left(\mathbf{1}(\pi_{i, j, t}^{1, 0} \neq 0)\mathbf{1}(H_{i, t} = h_i^1)\mathbf{1}(H_{j, t} = h_j^0)
    \times\frac{(\pi_{i, j, t}^{1, 0} - \pi_{i, t}^1\pi_{j, t}^0)Y_{i, t}Y_{j, t}}{\pi_{i, t}^1\pi_{j, t}^0\pi_{i, j, t}^{1, 0}} \right.\\
    \left.-\mathbf{1}(\pi_{i, j, t}^{1, 0} = 0)\left(\frac{\mathbf{1}(H_{i, t} = h_i^1)Y_{i, t}^2}{2\pi_{i, t}^{1}} + \frac{\mathbf{1}(H_{j, t} = h_j^0)Y_{j, t}^2}{2\pi_{j, t}^0}\right)\right) \\
    \left.+ \mathbf{1}(\pi_{i, j, t}^{0, 0} \neq 0)\frac{\mathbf{1}(H_{i, t} = h_i^0)\mathbf{1}(H_{j, t} = h_j^0)(\pi_{i, j, t}^{0, 0} - \pi_{i, t}^0\pi_{j, t}^0)Y_{i, t}Y_{j, t}}{\pi_{i, t}^0\pi_{j, t}^0\pi_{i, j, t}^{0, 0}}\right]
\end{multline}
\normalsize
and
\small
\begin{multline}
    \widehat{Var}^d(\hat{\tau}_t^{TE}) = \frac{1}{n^2}\sum_{i = 1}^{n}\left[\mathbf{1}(H_{i, t} = h_i^1)(1 - \pi_{i, t}^1)\left(\frac{Y_{i, t}}{\pi_{i, t}^1}\right)^2 + \mathbf{1}(H_{i, t} = h_i^0)(1 - \pi_{i, t}^0)\left(\frac{Y_{i, t}}{\pi_{i, t}^0}\right)^2 \right] \\
    + \frac{2}{n^2}\sum_{1 \leq i < j \leq n}
    \left[\left(\mathbf{1}(\pi_{i, j, t}^{1, 1} \neq 0)\frac{\mathbf{1}(H_{i, t} = h_i^1)\mathbf{1}(H_{j, t} = h_j^1)(\pi_{i, j, t}^{1, 1} - \pi_{i, t}^1\pi_{j, t}^1)Y_{i, t}Y_{j, t}}{\pi_{i, t}^1\pi_{j, t}^1\pi_{i, j, t}^{1, 1}}\right.\right. \\
    \left.-\mathbf{1}(\pi_{i, j, t}^{1, 1} = 0)\left(\frac{\mathbf{1}(H_{i, t} = h_i^1)Y_{i, t}^2}{2\pi_{i, t}^{1}} + \frac{\mathbf{1}(H_{j, t} = h_j^1)Y_{j, t}^2}{2\pi_{j, t}^1}\right)\right) \\
    \left.- \mathbf{1}(\pi_{i,j,t}^{0, 1} \neq 0)\frac{\mathbf{1}(H_{i, t} = h_i^0)\mathbf{1}(H_{j, t} = h_j^1)(\pi_{i, j, t}^{0, 1} - \pi_{i, t}^0\pi_{j, t}^1)Y_{i, t}Y_{j, t}}{\pi_{i, t}^0\pi_{j, t}^1\pi_{i, j, t}^{0, 1}}\right. \\
    \left.- \mathbf{1}(\pi_{i,j,t}^{1, 0} \neq 0)\frac{\mathbf{1}(H_{i, t} = h_i^1)\mathbf{1}(H_{j, t} = h_j^0)(\pi_{i, j, t}^{1, 0} - \pi_{i, t}^1\pi_{j, t}^0)Y_{i, t}Y_{j, t}}{\pi_{i, t}^1\pi_{j, t}^0\pi_{i, j, t}^{1, 0}} \right. \\
    + \left.\left(\mathbf{1}(\pi_{i, j, t}^{0, 0} \neq 0)\frac{\mathbf{1}(H_{i, t} = h_i^0)\mathbf{1}(H_{j, t} = h_j^0)(\pi_{i, j, t}^{0, 0} - \pi_{i, t}^0\pi_{j, t}^0)Y_{i, t}Y_{j, t}}{\pi_{i, t}^0\pi_{j, t}^0\pi_{i, j, t}^{0, 0}}\right.\right. \\
    \left.-\mathbf{1}(\pi_{i, j, t}^{0, 0} = 0)\left(\frac{\mathbf{1}(H_{i, t} = h_i^0)Y_{i, t}^2}{2\pi_{i, t}^{0}} + \frac{\mathbf{1}(H_{j, t} = h_j^0)Y_{j, t}^2}{2\pi_{j, t}^0}\right)\right].
\end{multline}
\normalsize
Assuming all the potential outcomes are non-negative, we then have that
\begin{equation*}
    \mathbb{E}\left[\widehat{Var}^u(\hat{\tau}_t^{TE})\right] \geq Var(\hat{\tau}_t^{TE})
\end{equation*}
and 
\begin{equation*}
    \mathbb{E}\left[\widehat{Var}^d(\hat{\tau}_t^{TE})\right] \leq Var(\hat{\tau}_t^{TE}).
\end{equation*}
\end{proposition}

\begin{proposition} [Estimator of the covariance]
\label{prop: totaleffectcovest}
We have the following unbiased estimator of $Cov(\hat{\tau}_t^{TE}, \hat{\tau}_{t^{'}}^{TE})$:
\small
\begin{multline*}
    \widehat{Cov}(\hat{\tau}_t^{TE}, \hat{\tau}_{t^{'}}^{TE}) = \frac{1}{n^2}\sum_{i = 1}^{n}
    \left[\frac{\mathbf{1}(H_{i, t} = h_i^1)\mathbf{1}(H_{i, t^{'}} = h_i^1)Y_{i, t}Y_{i, t^{'}}(\pi_{i, t, t^{'}}^{1, 1} - \pi_{i, t}^{1}\pi_{i, t^{'}}^1)}{\pi_{i, t, t^{'}}^{1, 1}\pi_{i, t}^{1}\pi_{i, t^{'}}^1} \right.\\
    - 
    \left. \frac{\mathbf{1}(H_{i, t} = h_i^0)\mathbf{1}(H_{i, t^{'}} = h_i^1)Y_{i, t}Y_{i, t^{'}}(\pi_{i, t, t^{'}}^{0, 1} - \pi_{i, t}^{0}\pi_{i, t^{'}}^1)}{\pi_{i, t, t^{'}}^{0, 1}\pi_{i, t}^{0}\pi_{i, t^{'}}^1}\right. \\
    - \frac{\mathbf{1}(H_{i, t} = h_i^1)\mathbf{1}(H_{i, t^{'}} = h_i^0)Y_{i, t}Y_{i, t^{'}}(\pi_{i, t, t^{'}}^{1, 0} - \pi_{i, t}^{1}\pi_{i, t^{'}}^0)}{\pi_{i, t, t^{'}}^{1, 0}\pi_{i, t}^{1}\pi_{i, t^{'}}^0}\\ + 
    \left.\frac{\mathbf{1}(H_{i, t} = h_i^0)\mathbf{1}(H_{i, t^{'}} = h_i^0)Y_{i, t}Y_{i, t^{'}}(\pi_{i, t, t^{'}}^{0, 0} - \pi_{i, t}^{0}\pi_{i, t^{'}}^0)}{\pi_{i, t, t^{'}}^{0, 0}\pi_{i, t}^{0}\pi_{i, t^{'}}^0}\right] \\
    + \frac{2}{n^2}\sum_{1 \leq i < j \leq n}
    \left[\frac{\mathbf{1}(H_{i,t} = h_i^1)\mathbf{1}(H_{j, t^{'}} = h_j^1)Y_{i, t}Y_{j, t^{'}}(\pi_{i, t, j, t^{'}}^{1,1} - \pi_{i, t}^{1}\pi_{j, t^{'}}^{1})}{\pi_{i, t, j, t^{'}}^{1,1}\pi_{i, t}^{1}\pi_{j, t^{'}}^{1}} \right.\\
    - \frac{\mathbf{1}(H_{i,t} = h_i^0)\mathbf{1}(H_{j, t^{'}} = h_j^1)Y_{i, t}Y_{j, t^{'}}(\pi_{i, t, j, t^{'}}^{0,1} - \pi_{i, t}^{0}\pi_{j, t^{'}}^{1})}{\pi_{i, t, j, t^{'}}^{0,1}\pi_{i, t}^{0}\pi_{j, t^{'}}^{1}} \\
    - \frac{\mathbf{1}(H_{i,t} = h_i^1)\mathbf{1}(H_{j, t^{'}} = h_j^0)Y_{i, t}Y_{j, t^{'}}(\pi_{i, t, j, t^{'}}^{1,0} - \pi_{i, t}^{1}\pi_{j, t^{'}}^{0})}{\pi_{i, t, j, t^{'}}^{1,0}\pi_{i, t}^{1}\pi_{j, t^{'}}^{0}} \\
    + \left.\frac{\mathbf{1}(H_{i,t} = h_i^0)\mathbf{1}(H_{j, t^{'}} = h_j^0)Y_{i, t}Y_{j, t^{'}}(\pi_{i, t, j, t^{'}}^{0,0} - \pi_{i, t}^{0}\pi_{j, t^{'}}^{0})}{\pi_{i, t, j, t^{'}}^{0,0}\pi_{i, t}^{0}\pi_{j, t^{'}}^{0}}\right]
\end{multline*}
\end{proposition}
\normalsize

Proving Theorem~\ref{thm: mixedgeneral} relies on results in $m$-dependence central limit theorem. We need the following result as a lemma.
\begin{lemma}
\label{lemma: romanowolf}
Let $\{X_{n, i}\}$ be a triangular array of mean zero random variables. For each $n = 1, 2, \cdots$ let $d = d_n$, and suppose $X_{n, 1}, \cdots, X_{n, d}$ is an $m$-dependent sequence of random variables for some $m \in \mathbb{N}$. Define
\begin{equation*}
    B_{n, k, a}^2 = \text{Var}\left(\sum_{i = a}^{a + k - 1}X_{n, i}\right),
    B_n^2 = B_{n, d, 1} = \text{Var}\left(\sum_{i = 1}^{d}X_{n, i}\right).
\end{equation*}
Assume the following conditions hold. For some $\delta > 0$, $-1 \leq \gamma < 1$ and $g = g_n > 2m$ is such that $\frac{m}{g} \rightarrow 0$:
\begin{equation}
    \mathbb{E}|X_{n, i}|^{2 + \delta} \leq \Delta_n \,\, \text{for all} \,\, i,
\end{equation}
\begin{equation}
    B_{n, k, a}^2/(k^{1 + \gamma}) \leq K_n \,\, \text{for all} \,\, a \,\, \text{and for all}\,\, k \geq m,
\end{equation}
\begin{equation}
    B_n^2/(dm^\gamma) \geq L_n,
\end{equation}
\begin{equation}
\frac{K_n}{L_n}\cdot\frac{m}{g} \rightarrow 0, \label{romanocondition4}
\end{equation}
\begin{equation} \label{romanocondition5}
\frac{K_n}{L_n}\cdot\left(\frac{m}{g}\right)^{(1 - \gamma)/2} \rightarrow 0,
\end{equation}
\begin{equation}
    \Delta_nL_n^{-(2 + \delta)/2}g^{\delta/2 + (1 - \gamma)(2 + \delta)/2}d^{-\delta/2}\left(\frac{m}{g}\right)^{(1 - \gamma)(2 + \delta)/2} \rightarrow 0. \label{romanocondition6}
\end{equation}
Then, $B_n^{-1}(X_{n, 1} + \cdots X_{n, d}) \xrightarrow{d} \mathcal{N}(0, 1)$.
\end{lemma}
\begin{proof}
This is essentially Theorem 2.1 in \cite{ROMANO2000115}. We replace the original conditions 4, 5 and 6 by the last three conditions. In fact, the last three conditions are needed to establish the theorem and the conditions 4, 5 and 6 in Theorem 2.1 in \cite{ROMANO2000115} are sufficient conditions.
\end{proof}
Now, we are ready to prove Theorem~\ref{thm: mixedgeneral}.
\begin{proof} [Proof of Theorem~\ref{thm: mixedgeneral}]
We define $\tilde{\tau}_{i, t} = \frac{\mathbf{1}(H_{i,t} = k)}{\mathbb{P}(H_{i, t} = k)}Y_{i, t} - \frac{\mathbf{1}(H_{i, t} = k')}{\mathbb{P}(H_{i, t} = k')}Y_{i, t} = \frac{\mathbf{1}(H_{i,t} = k)}{\mathbb{P}(H_{i, t} = k)}Y_{i, t}(k) - \frac{\mathbf{1}(H_{i, t} = k')}{\mathbb{P}(H_{i, t} = k')}Y_{i, t}(k')$. Then the ATEC can be written as
\begin{equation*}
    \hat{\bar{\tau}}^{k, k^{'}} = \sum_{i=1}^{n}\sum_{t=1}^T\frac{1}{nT}\tilde{\tau}_{i, t}.
\end{equation*}
Similarly, we define $\tau_{i, t} = Y_{i, t}(k) - Y_{i, t}(k')$, which is the true individual exposure contrast. Now,
\begin{equation*}
    \sqrt{nT}(\hat{\bar{\tau}}^{k, k^{'}} - \bar{\tau}^{k, k^{'}}) = \sum_{i=1}^{n}\sum_{t=1}^T\frac{1}{\sqrt{nT}}(\tilde{\tau}_{i, t} - \tau_{i, t}).
\end{equation*}
To proceed, we let $X_{n, i, t} = \frac{1}{\sqrt{nT}}(\tilde{\tau}_{i, t} - \tau_{i, t})$. We view $\{X_{n, i, t}\}$ as a single sequence of random variables by enumerating $X_{n, i, t}$ following the order $X_{n, 1, 1}, \cdots, X_{n, 1, T}, X_{n, 2, 1}, \cdots, X_{n, 2, T}, \cdots, X_{n, n, T}$. Using the language in the lemma, $d = nT$. Since $\{H_{i, t}\}_{i=1}^{n}$ is a sequence of $s$-dependent random variables and $X_{n, i, t}$ is a function of $H_{i, t}$, we know that $\{X_{n, i, t}\}$ is a $sT$-dependent sequence of random variables. In other words, $m = sT$ in the above lemma. Note that $|X_{n, i, t}| \leq \frac{C_1}{\sqrt{nT}}$ by uniform boundedness of potential outcomes. Hence, for any $\delta > 0$, $\Delta_n = C_2(nT)^{-1-\delta/2}$. Now, we calculate $B_{n, k, a}^2$ and $B_n^2$. We start with $B_{n, k, a}^2$. For all $(i_1, t_1)$ and $k \geq m$, let $(i_2, t_2)$ be the index such that when we order $X$'s there are exactly $k$ indices from $(i_1, t_1)$ to $(i_2, t_2)$.
\begin{align*}
    B_{n, k, a}^2 &= \text{Var}\left(\sum_{(i, t) = (i_1, t_1)}^{(i, t) = (i_2, t_2)}X_{n, i, t}\right) \\
    &= \frac{1}{nT} \text{Var}\left(\sum_{(i, t) = (i_1, t_1)}^{(i, t) = (i_2, t_2)}\tilde{\tau}_{i, t}\right) \\
    &= \frac{1}{nT}\left[\sum_{(i, t) = (i_1, t_1)}^{(i, t) = (i_2, t_2)}\text{Var}(\tilde{\tau}_{i, t}) + 2\sum_{(u, v) \neq (p, q)}\text{Cov}(\tilde{\tau}_{u, v}, \tilde{\tau}_{p, q})\right]
\end{align*}
Since $k \geq m = sT$, we know that at most $mk$ covariance terms are non-zero. Given uniform boundedness of potential outcomes and overlap, all the variance and covariance terms are upper bounded by constants $M_1 > 0$ and $M_2 > 0$ respectively. Hence,
\begin{equation*}
    B_{n, k, a}^2 \leq \frac{1}{nT}(kM_1 + 2mkM_2) \leq M_3\frac{mk}{nT} = M_3\frac{sk}{n}.
\end{equation*}
Therefore,
\begin{equation*}
    B_{n, k, a}^2/k \leq M_3\frac{sk}{n}/k = M_3\frac{s}{n} = K_n.
\end{equation*}
Now we look at $B_n^2$. By Assumption~\ref{assump: nonvanishvarmixed}, $\text{Var}(\sqrt{nT}\hat{\bar{\tau}}^{k, k^{'}}) \geq \epsilon > 0$, hence, for sufficiently large $n$,
\begin{equation*}
    B_n^2 = \text{Var}(\sqrt{nT}\hat{\bar{\tau}}^{k, k^{'}}) \geq \epsilon > 0,
\end{equation*}
and
\begin{equation*}
    B_n^2/d = B_n^2/(nT) \geq \epsilon/(nT) = L_n.
\end{equation*}
We let $\gamma = 0, \delta = 2$. Pick $g = g_n = s^3T^3n^{\alpha}$. With such $g$, $m/g$ obviously goes to 0. Now,
\begin{equation*}
    \frac{K_n}{L_n}\cdot\frac{m}{g} = \epsilon M_3 sT\cdot\frac{1}{s^2T^2n^{\alpha}} \rightarrow 0,
\end{equation*}
\begin{equation*}
    \frac{K_n}{L_n}\cdot\left(\frac{m}{g}\right)^{(1 - \gamma)/2} = \epsilon M_3 sT\cdot\frac{1}{sTn^{0.5\alpha}} \rightarrow 0,
\end{equation*}
\begin{align*}
    &\Delta_nL_n^{-(2 + \delta)/2}g^{\delta/2 + (1 - \gamma)(2 + \delta)/2}d^{-\delta/2}\left(\frac{m}{g}\right)^{(1 - \gamma)(2 + \delta)/2} \\
    &= C_2(nT)^{-1-\delta/2}\epsilon^{-(2 + \delta)/2}(nT)^{(2 + \delta)/2}g^{\delta/2}(nT)^{-\delta/2}(sT)^{1 + \delta/2} \\
    &= C_4gs^2T/n \qquad \text{when } \delta = 2 \\
    &= C_4s^5T^4n^{\alpha}/n.
\end{align*}
Since $s^5T^4 = o(n^{1-\alpha})$, $s^5T^4n^{\alpha} = o(n)$ and hence $\Delta_nL_n^{-(2 + \delta)/2}g^{\delta/2 + (1 - \gamma)(2 + \delta)/2}d^{-\delta/2}\left(\frac{m}{g}\right)^{(1 - \gamma)(2 + \delta)/2} = o(1)$. Having checked all the conditions, by Lemma~\ref{lemma: romanowolf}, we are done.
\end{proof}

\begin{proof} [Proof of Theorem \ref{thm: groupclt}]
As in the above proof, we check the six conditions in Lemma~\ref{lemma: romanowolf} are satisfied with $\gamma = 0$ and $\delta = 2$. Note that since now $X_{n, i, t}$ and $X_{n, j, t}$ are correlated if and only if $i$ and $j$ are in the same group, we can reorder $X_{n, i, t}$'s as follows:
\begin{equation*}
    X_{n, 1, 1}, \cdots, X_{n, r, 1}, X_{n, 1, 2}, \cdots, X_{n, r, 2}, \cdots, X_{n, r, T}, X_{n, r+1, 1}, \cdots, X_{n, nr, T}.
\end{equation*}
Now, this sequence is actually $(2r)$-dependent, i.e., $m = 2r, s = r$. Then
\begin{equation*}
    K_n = M_4/(nT), \quad L_n = \epsilon/(nrT).
\end{equation*}
Hence $K_n/L_n = M_5r$. Pick $g = g_n$ such that $g \rightarrow \infty$ and $g = (nT)^{3/4}$. Then with $r = o((nT)^{\frac{1}{4}})$, $r^2/g \rightarrow 0$ and $r^3/g \rightarrow 0$.
\begin{equation*}
    \frac{K_n}{L_n}\cdot\frac{m}{g} = M_5r\cdot\frac{2r}{g} \rightarrow 0,
\end{equation*}
\begin{equation*}
    \frac{K_n}{L_n}\cdot\left(\frac{m}{g}\right)^{(1 - \gamma)/2} = M_5r\cdot\sqrt{\frac{2r}{g}} \rightarrow 0
\end{equation*}
and
\begin{align*}
    &\Delta_nL_n^{-(2 + \delta)/2}g^{\delta/2 + (1 - \gamma)(2 + \delta)/2}d^{-\delta/2}\left(\frac{m}{g}\right)^{(1 - \gamma)(2 + \delta)/2} \\
    &= M_6g^{3}(nrT)^{-1}\left(\frac{r}{g}\right)^{2} \\
    &= M_6rg/(nT) \\
    &= o(nT)/(nT) \rightarrow 0
\end{align*}
Hence all the conditions are satisfied. It is also easy to see that instead of just 2 time steps, any finite $p$ time steps would work.
\end{proof}

\begin{proof} [Proof of Proposition \ref{prop: householdclt-ci}]
Let $X_{n, t} = \sqrt{\frac{nr}{T}}(\hat{\tau}^{k, k^{'}}_{t} - \tau^{k, k^{'}}_{t})$. The key ingredients are the following two expressions:
\begin{multline}
    \text{Var}(X_{n, t}) = \frac{1}{nrT}\left[\sum_{l = 1}^{n}\sum_{q = 1}^{r}(2^{2r} - 1)Y_{(l, q), t}(k)^2 \right.\\
    + \sum_{l = 1}^{n}\sum_{q = 1}^{r}(2^{2r} - 1)Y_{(l, q), t}(k^{'})^2 + 2\sum_{l = 1}^n\sum_{q = 1}^{r}Y_{(l, q), t}(k)Y_{(l, q), t}(k^{'})\\
    + \sum_{l = 1}^{n}\sum_{q_1 = 1}^{r}\sum_{q_2 \neq q_1}\left((2^{2r} - 1)Y_{(l, q_1), t}(k)Y_{(l, q_2), t}(k) + (2^{2r} - 1)Y_{(l, q_1), t}(k^{'})Y_{(l, q_2), t}(k^{'})\right)\\
    \left. + 2\sum_{l = 1}^{n}\sum_{q_1 = 1}^{r}\sum_{q_2 \neq q_1}Y_{(l, q_1), t}(k)Y_{(l, q_2), t}(k^{'})\right] \label{varxnt}
\end{multline}
and
\begin{multline}
    \text{Cov}(X_{n, t}, X_{n, t+1}) = \frac{1}{nrT}\sum_{l = 1}^n\sum_{q_1 = 1}^r\sum_{q_2 = 1}^r\\
    \left((2^r - 1)Y_{(l, q_1), t}(k)Y_{(l, q_2), t+1}(k)\right.\\
    \left.+ (2^r - 1)Y_{(l, q_1), t}(k^{'})Y_{(l, q_2), t+1}(k^{'}) \right.\\
    + \left.Y_{(l, q_1), t}(k^{'})Y_{(l, q_2), t+1}(k) + Y_{(l, q_1), t}(k)Y_{(l, q_2), t+1}(k^{'})\right) \label{covnt}
\end{multline}
We have that
\begin{equation*}
    B_n^2 = \sum_{t = 1}^{T}\text{Var}(X_{n, t}) + 2\sum_{t = 1}^{T - 1}\text{Cov}(X_{n, t}, X_{n, t + 1})
\end{equation*}
Plugging in \eqref{varxnt} and \eqref{covnt}, we have the expression of $B_n^2$. The estimator is obtained by replacing the non-identifiable terms by corresponding upper bound.
\end{proof}




\section{\texorpdfstring{$k-$}{k-}steps convex estimator}
\label{suppl_sec_b}

The approach we have described in Section~\ref{subsec: epsilonstab} naturally extends 
to using the $k - 1$ previous time steps, yielding the weighted combination estimator:
\begin{equation*}
    \hat{\tau}_t^c = \alpha_1\hat{\tau}_{t-k+1}^{TE} + \cdots + \alpha_k\hat{\tau}_{t}^{TE},
\end{equation*}
which exhibits the following absolute bias bound:
\begin{proposition} [Bound on the bias of $\hat{\tau}_t^c$]
\label{prop: biasgeneralk}
\begin{equation*}
    |\mathbb{E}[\hat{\tau}_t^c] - \tau_t^{TE}| \leq 2\left[(k - 1)\alpha_1 + (k - 2)\alpha_2 + \cdots + \alpha_{k - 1}\right]\epsilon
\end{equation*}
\end{proposition}
As in the previous section, we can estimate $\alpha_1, \cdots, \alpha_k$ by 
solving the following convex optimization problem:
\begin{align*}
    &\argmin_{\alpha_1, \cdots, \alpha_k} \qquad \alpha_1^2\widehat{\text{Var}}(\hat{\tau}_{t-k+1}^{TE}) + \cdots + \alpha_k^2\widehat{\text{Var}}(\hat{\tau}_{t}^{TE}) \\
    &\qquad \qquad\qquad+ 4\left[(k - 1)\alpha_1 + \cdots + \alpha_{k - 1}\right]^2\epsilon^2 \\
    &\text{subject to} \qquad \alpha_1 + \cdots + \alpha_k = 1,
\end{align*}
where $\widehat{\text{Var}}(\hat{\tau}_{t-k+1}^{TE}), \cdots, \widehat{\text{Var}}(\hat{\tau}_{t}^{TE})$ are estimators of the associated 
variance terms, and are provided in Appendix \ref{appA}. This then 
suggests the following plug-in estimator:
\begin{equation*}
    \hat{\tau}_t^c = \hat{\alpha}_1\hat{\tau}_{t-k+1}^{TE} + \cdots + \hat{\alpha}_k\hat{\tau}_{t}^{TE}.
\end{equation*}
%
We can assert stronger control over the bias of $\hat{\tau}^c$ by incorporating 
an additional constraint to the optimization problem:
\begin{align*}
    \argmin_{\alpha_1, \cdots, \alpha_k} &\quad \alpha_1^2\widehat{\text{Var}}(\hat{\tau}_{t-k+1}^{TE}) + \cdots + \alpha_k^2\widehat{\text{Var}}(\hat{\tau}_{t}^{TE}) \\
    &\qquad + 4\left[(k - 1)\alpha_1 + \cdots + \alpha_{k - 1}\right]^2\epsilon^2 \\
    \text{subject to} &\quad \alpha_1 + \cdots + \alpha_k = 1, \\
    &\quad 2\left[(k - 1)\alpha_1 + (k - 2)\alpha_2 + \cdots + \alpha_{k - 1}\right]\epsilon \leq \delta.
\end{align*}
Numerical solutions for either optimization problem are straightforward to 
obtain using standard numerical solvers. Variance estimator and confidence interval of $\tau_t^{TE}$ can be constructed in exactly the same way as in the case $k = 2$.


\begin{proof} [Proof of Proposition \ref{prop: biasgeneralk}]
\begin{align*}
    &|\mathbb{E}[\hat{\tau}_t^c] - \tau_t| \\
    &= |\alpha_1\tau_{t - k + 1}^{TE} + \cdots + \alpha_k\tau_{t}^{TE} - \tau_t^{TE}| \\
    &= |\alpha_1\tau_{t - k + 1}^{TE} + \cdots + \alpha_{k - 1}\tau_{t - 1}^{TE} - (1 - \alpha_k)\tau_t^{TE}| \\
    &= |\alpha_1\tau_{t - k + 1}^{TE} + \cdots + \alpha_{k - 1}\tau_{t - 1}^{TE} - (\alpha_1 + \cdots + \alpha_{k - 1})\tau_t^{TE}| \\
    &= |\alpha_1(\tau_{t - k + 1}^{TE} - \tau_t^{TE}) + \cdots + \alpha_{k - 1}(\tau_{t - 1}^{TE} - \tau_{t}^{TE})| \\
    &\leq \alpha_1|\tau_{t - k + 1}^{TE} - \tau_t^{TE}| + \cdots + \alpha_{k - 1}|\tau_{t - 1}^{TE} - \tau_{t}^{TE}| \\
    &\leq 2\alpha_1(k - 1)\epsilon + \cdots + 2\alpha_{k - 1}\epsilon \\
    &= 2\left[(k - 1)\alpha_1 + \cdots + \alpha_{k - 1}\right]\epsilon
\end{align*}
\end{proof}

We first give the optimization problem for the general case that assignments may be correlated across time:
\begin{align*}
    \argmin_{\alpha_1, \cdots, \alpha_k} &\quad \alpha_1^2\widehat{\text{Var}}(\hat{\tau}_{t-k+1}^{TE}) + \cdots + \alpha_k^2\widehat{\text{Var}}(\hat{\tau}_{t}^{TE}) \\
    & \quad + 2\alpha_{i}\alpha_{j}\sum_{1 \leq i < j \leq k}\widehat{\text{Cov}}(\hat{\tau}_{t -k + i}^{TE}, \hat{\tau}_{t - k + j}^{TE}) \\
    &\quad + 4\left[(k - 1)\alpha_1 + \cdots + \alpha_{k - 1}\right]^2\epsilon^2 \\
    \text{subject to} &\quad \alpha_1 + \cdots + \alpha_k = 1,
\end{align*}
where $\widehat{\text{Var}}(\hat{\tau}_{t-k+1}^{TE}), \cdots, \widehat{\text{Var}}(\hat{\tau}_{t}^{TE})$ and $\widehat{\text{Cov}}(\hat{\tau}_{t -k + i}^{TE}, \hat{\tau}_{t - k + j}^{TE})$ can be any estimator in Proposition \ref{prop: totaleffectvarest} and \ref{prop: totaleffectcovest}. Moreover, suppose that the assignments are independent across time, we know that $\text{Cov}(\hat{\tau}_{t -k + i}^{TE}, \hat{\tau}_{t - k + j}^{TE}) = 0$, hence we have an even simpler optimization problem as stated in the main text.
\begin{proof}
[Derivation of the optimization problem]
We first calculate the variance. 
\begin{align*}
    \text{Var}(\hat{\tau}_t^c) &= \text{Var}(\alpha_1\hat{\tau}_{t-k+1}^{TE} + \cdots + \alpha_k\hat{\tau}_{t}^{TE}) \\
    &= \alpha_1^2\text{Var}(\hat{\tau}_{t-k+1}^{TE}) + \cdots + \alpha_k^2\text{Var}(\hat{\tau}_{t}^{TE}) \\
    &\qquad + 2\alpha_{i}\alpha_{j}\sum_{1 \leq i < j \leq k}\text{Cov}(\hat{\tau}_{t -k + i}^{TE}, \hat{\tau}_{t - k + j}^{TE})
\end{align*}
Suppose we want to have smaller MSE by using $\hat{\tau}_t^c$, we need to have
\begin{equation*}
    \text{Var}(\hat{\tau}_t^c) + |\mathbb{E}[\hat{\tau}_t^c] - \tau_t^{TE}|^2 \leq \text{Var}(\hat{\tau}_t^{TE})
\end{equation*}
By Proposition \ref{prop: biasgeneralk}, it suffices to have
\begin{equation}
\begin{multlined}
    \alpha_1^2\text{Var}(\hat{\tau}_{t-k+1}^{TE}) + \cdots + \alpha_k^2\text{Var}(\hat{\tau}_{t}^{TE}) \\
    + 2\alpha_{i}\alpha_{j}\sum_{1 \leq i < j \leq k}\text{Cov}(\hat{\tau}_{t -k + i}^{TE}, \hat{\tau}_{t - k + j}^{TE}) \\
    + 4\left[(k - 1)\alpha_1 + \cdots + \alpha_{k - 1}\right]^2\epsilon^2
    \leq \text{Var}(\hat{\tau}_t^{TE})
\end{multlined} \label{eq: varineq, generalk}
\end{equation}
Now, the left hand side of (\ref{eq: varineq, generalk}) is convex in $\alpha_1, \cdots, \alpha_k$.
\end{proof}

\newpage 

\section{Additional simulation results} \label{appC}

\subsection{Simulation results for estimation under stability assumption}
\subsubsection{Parameters for Erd\H{o}s-R\'{e}nyi Model}
For the simulation study in Section~\ref{subsubsec:simul_stab_est}, we use $p = 0.1$ for $n = 50$ and then scale the probability $p$ accordingly for larger $n$ so that each unit has the same expected number of neighbors.
\subsubsection{The effect of estimated stability parameter}
Recall that our $\hat{\epsilon}$ is only a lower bound of the true $\epsilon$, hence may underestimate $\epsilon$. To investigate how our estimate of $\epsilon$ affects the results, we fix $n = 50$ and generate the social network according to Erd\H{o}s-R\'{e}nyi Model with $p = 0.1$. We generate 500 realizations of assignments and plug in $\hat{\epsilon}$, $1.5\hat{\epsilon}$, $2\hat{\epsilon}$, $2.5\hat{\epsilon}$ and $3\hat{\epsilon}$ for three kinds of estimators considered above.
\begin{table} 
\centering
\scalebox{0.8} {
\begin{tabular}{c l c c c c c c c} 
\toprule 
\multicolumn{2}{c}{Estimate of $\epsilon$} & $\hat{\epsilon}$ & $1.5\hat{\epsilon}$ & $2\hat{\epsilon}$ & $2.5\hat{\epsilon}$ & $3\hat{\epsilon}$\\
\midrule
    \multirow{1}{*}{RMSE for $\hat{\tau}_{20}^{TE}$} 
            & & 33.27 & 33.27 & 33.27 & 33.27 & 33.27\\
            \midrule
    \multirow{1}{*}{RMSE for $\hat{\tau}^c_{20}$, $k$ = 2} 
            & & 8.93 & 8.81 & 8.69 & 8.55 & 8.42\\
            \midrule
    \multirow{1}{*}{RMSE for $\hat{\tau}^c_{20}$, $k$ = 5} 
            & & 5.12 & 6.20 & 7.11 & 7.91 & 8.64\\
\bottomrule 
\end{tabular}}
\caption{Root mean squared errors (RMSE) for $\hat{\tau}_{20}^{TE}$, $\hat{\tau}^c_{20}$ with $k = 2$ and $\hat{\tau}^c_{20}$ with $k = 5$} 
\label{table:cvxest_diffeps} 
\end{table}
Table \ref{table:cvxest_diffeps} shows the results. We see that the convex combination type estimator with $k = 2$ is not sensitive to the estimate of $\epsilon$ while the convex combination type estimator with $k = 5$ is. Even we use $3\hat{\epsilon}$, two convex combination type estimators still show better performance in terms of root mean squared error.

\begin{figure}[H]
\centering
    \includegraphics[totalheight=6cm]{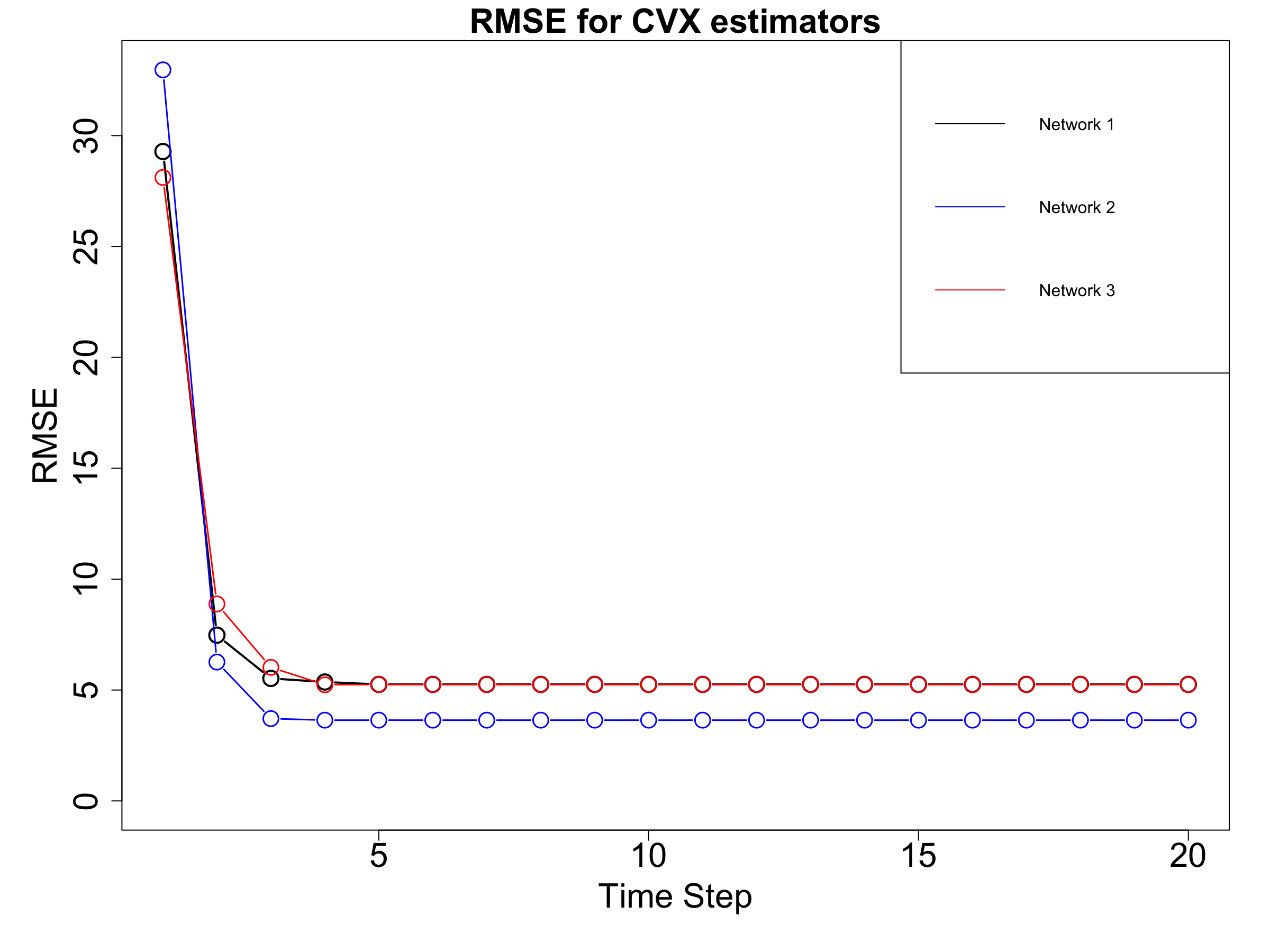}
    \caption{Root mean squared errors (RMSE) for $\hat{\tau}_{20}^{TE}$ and $\hat{\tau}^c$}
    \label{fig:rmse}
\end{figure}
\subsubsection{The effect of the number of time steps}
Finally, we investigate how $k$ affects the results. We generate three different social networks, and for each one, we plot the root mean squared errors of using 1 time step (i.e., the Horvitz-Thompson type estimator) to 20 time steps (i.e., we use all time steps to estimate the total effect at time step 20). From Figure \ref{fig:rmse} we can see that the RMSE curves stay flat after a certain value of $k$. Hence, we do not need to worry about using too many time steps as the optimization problem intrinsically pick the right $k$.

\subsubsection{Lengths of approximate confidence intervals}
\begin{table} 
\centering
\scalebox{0.75} {
\begin{tabular}{c l c c c c c} 
\toprule 
\multicolumn{2}{c}{Confidence Interval} & Network 1 & Network 2 & Network 3\\
\midrule
    \multirow{1}{*}{Gaussian CI with variance estimated by $\widehat{\text{Var}}^d$} 
            & & 27.38 & 26.62 & 27.02\\
            \midrule
    \multirow{1}{*}{Gaussian CI with variance estimated by $\widehat{\text{Var}}^u$} 
            & & 34.04 & 32.34 & 33.33\\
            \midrule
    \multirow{1}{*}{Chebyshev CI with variance estimated by $\widehat{\text{Var}}^d$} 
            & & 62.47 & 60.75 & 61.66\\
            \midrule
    \multirow{1}{*}{Chebyshev CI with variance estimated by $\widehat{\text{Var}}^u$} 
            & & 77.67 & 73.79 & 76.04\\
\bottomrule 
\end{tabular}}
\caption{Lengths of two approximate confidence intervals for $\tau_t^{TE}$ with $k = 2$} 
\label{table:cvx_est_ci_length} 
\end{table}
Table~\ref{table:cvx_est_ci_length} shows the average lengths of approximate confidence intervals. As expected, Gaussian confidence intervals are shorter.

\end{document}